\documentclass[12pt]{article}
\usepackage{amsmath}
\usepackage[round]{natbib}
\usepackage{amssymb,color}
\usepackage{appendix}
\usepackage{amsthm}
\usepackage{hyperref}
\usepackage{enumitem}
\usepackage{dsfont}
\usepackage{mathtools} % need for `show only references'
\mathtoolsset{showonlyrefs=true}  % only equations which are labeled AND referenced will be numbered.                                                                                                 % IMPORTANT NOTE...must use \eqref{} instead of (\ref{})

\usepackage{caption}

\usepackage{ mathrsfs }
\usepackage[margin=0.6in]{geometry}

\numberwithin{equation}{section}
\usepackage{comment}

\usepackage{titling}
\usepackage{xcolor}

\usepackage{tikz}
\usetikzlibrary{matrix}

\usepackage{varwidth}

\newtheorem{thm}{Theorem}[section]

\newtheorem{defi}{Definition}[section]
\newtheorem{assume}{Assumption}[section]
\newtheorem{prop}[thm]{Proposition}

\newtheorem{lemma}[thm]{Lemma}

\newtheorem{remark}{Remark}[section]

\begin{document}

	\author{Jaehyun Kim\thanks{jaehyun107@snu.ac.kr} and Hyungbin Park\thanks{hyungbin@snu.ac.kr, hyungbin2015@gmail.com}     \\ \\ \normalsize{Department of Mathematical Sciences} \\ 
		\normalsize{Seoul National University}\\
		\normalsize{1, Gwanak-ro, Gwanak-gu, Seoul, Republic of Korea} 
	}

\title{A $G$-BSDE approach to the long-term decomposition of robust pricing kernels}	
	\maketitle

\abstract{
	This study proposes a BSDE approach to the long-term decomposition of pricing kernels under the $G$-expectation framework. We establish the existence, uniqueness, and regularity of solutions to three types of quadratic $G$-BSDEs: finite-horizon $G$-BSDEs, infinite-horizon $G$-BSDEs, and ergodic $G$-BSDEs. Moreover, we explore the Feynman--Kac formula associated with these three types of quadratic $G$-BSDEs. Using these results, a pricing kernel is uniquely decomposed into four components: an exponential discounting component, a transitory component, a symmetric $G$-martingale, and a decreasing component that captures the volatility uncertainty of the $G$-Brownian motion. Furthermore, these components are represented through the solution to a second-order PDE. This study extends previous findings obtained under a single fixed probability framework to the $G$-expectation context.
	}

\section{Introduction}

A long-term analysis of pricing kernels plays a central role in financial economics, especially in understanding asset prices over extended horizons. 
Pricing kernels serve as valuation operators in dynamic stochastic environments, reflecting how economic agents discount uncertain future cash flows depending on current and future states of the world.
This paper explores a long-term decomposition of pricing kernels, which allows researchers to study how valuation evolves over time and how it depends on the structure of the underlying economic model. The decomposition characterizes the asymptotic behavior of the pricing kernel and reveals how long-term risks are priced. 
This framework offers key insights into the persistence of risk premia and the intertemporal structure of asset values.
This analysis is especially helpful when pricing long-lived assets such as long-term bonds and equities.

$G$-Brownian motion has been a valuable tool for modeling financial markets under volatility uncertainty. The mathematical foundations of $G$-expectation theory were rigorously developed in \cite{peng2007g,peng2008multi}. \cite{epstein2013ambiguous}
provided an economic rationale for employing $G$-Brownian motion in such settings, deriving arbitrage-free pricing rules based on hedging arguments and analyzing equilibrium asset prices within the framework of a consumption-based capital asset pricing model. The pricing of contingent claims under the $G$-expectation framework has been investigated in \cite{biagini2019robust}, \cite{holzermann2024pricing}, and \cite{vorbrink2014financial}.
Portfolio optimization problems within this framework have been explored in
\cite{fouque2016portfolio}, \cite{pun2021g}, and \cite{lin2021optimal}.

This study examines pricing kernels of the form \begin{equation}\label{eqn:pk} D_s=e^{-\int_0^sr(X_u)\,du-\int_0^sk_{ij}(X_u)\,d\langle B^i,B^j\rangle_u-\int_0^sv(X_u)\,dB_u},\quad s\ge0\,, \end{equation} 
where $r,k_{ij},v$ are continuous functions, $B=(B^1,B^2,\cdots,B^d)$ is a $G$-Brownian motion, and $X$ is a solution to a $G$-SDE. The pricing kernel \begin{equation} \label{eqn:our_decomp} D_s=e^{\lambda s}e^{u(X_0)-u(X_s)} M_se^{K_s},\quad s\ge0 
\end{equation} 
is uniquely decomposed 
into a discounting component with a long-term exponential rate $\lambda$, a transitory component $e^{u(X_0)-u(X)}$ for some function $u$, a positive symmetric $G$-martingale $M$, and a decreasing process $e^{K}$ that captures the volatility uncertainty of the $G$-Brownian motion. Furthermore, 
the function $u$ is characterized by the solution to a PDE, and the components $M$ and $K$ are expressed in terms of $u$.
To achieve this, the study employs a quadratic $G$-BSDE approach, investigating both finite and infinite horizon $G$-BSDEs, along with ergodic $G$-BSDEs.

This type of long-term decomposition has been studied by several authors.
Under the standard Brownian environment, \cite{hansen2009long} and \cite{hansen2012dynamic} decomposed
the pricing kernel as 
\begin{equation} \label{HS}
D_s=e^{\lambda s}e^{u(X_0)-u(X_s)}
M_s\,,\;s\ge0\,.
\end{equation}
The long-term exponential
rate $\lambda$ and 
the process $e^{u(X_0)-u(X)}$ 
correspond to the eigenvalue and eigenfunction of the Perron--Frobenius operator and the process $M$ is a martingale whose logarithm has stationary increments.
\cite{qin2017long} 
conducted the long-term decomposition under a general semimartingale framework.
The constant $-\lambda$ is the yield on the long bond and the process  
$e^{u(X_0)-u(X)}$ characterizes
gross holding
period returns on the long bond net of the long-term discount rate, and $M$
is a positive martingale that defines a long-term forward measure. A similar work was performed in \cite{qin2018along}
under the Heath--Jarrow--Morton
models.
\cite{severino2025long}
employed the long-term decomposition to reconcile the stochastic discount factor growth with the instantaneous relations between returns and rates in stochastic rate markets.

This study makes three key contributions.
First, our work extends previous results established
under a single fixed probability 
framework to the $G$-expectation context.
Unlike the literature mentioned above,
our $G$-BSDE approach does not rely on a single fixed expectation but operates with a sublinear expectation, representing the supremum of all possible expectations.
This robust approach is beneficial for conservative trading agents who evaluate upper prices with the sublinear expectation under volatility uncertainty. The process $K$ in the decomposition \eqref{eqn:our_decomp} is a decreasing $G$-martingale that reflects the effect of volatility uncertainty of the $G$-Brownian motion.
In the specific case where the $G$-Brownian motion is the standard Brownian motion, we have $K=0,$  which simplifies equation \eqref{eqn:our_decomp} to \eqref{HS}.
Consequently, our results 
can be interpreted as 
a robust version of long-term decompositions under volatility uncertainty.

Second, this study presents a novel BSDE approach that provides novel insights compared to previous literature. This is the first study to apply the BSDE framework to long-term decomposition. 
Our approach differs significantly from 
the operator approach proposed by \cite{hansen2009long} and \cite{hansen2012dynamic}, as well as the martingale approach suggested by \cite{qin2017long}.
Although the operator approach offers a methodology for achieving long-term decomposition, our BSDE approach explicitly identifies a sufficient condition for achieving long-term decomposition based on the specified parameters of pricing kernels. Additionally, our method establishes the uniqueness of the decomposition, a feature not addressed by the operator approach. In contrast to the martingale approach, 
our method characterizes the components of the decomposition via the solution of a second-order parabolic PDE.

Third, this paper develops a theory for a specific class of $G$-BSDEs. More precisely,  we investigate $G$-BSDEs
whose drivers are  uniformly Lipschitz continuous and possibly unbounded in the state variable and exhibit quadratic growth in the martingale integrand.
This specific class of $G$-BSDEs plays a crucial role in establishing the long-term decomposition in standard market models.  
We focus on  three main classes: finite-horizon,  infinite-horizon and ergodic $G$-BSDEs.
For each class, we establish the existence, uniqueness, and regularity of solutions.
These results are then employed, as outlined in \eqref{pic}, to  obtain the long-term decomposition. Furthermore, we investigate the related Feynman--Kac formula, which characterizes the components of the long-term  decomposition through the solution to a second-order parabolic PDE.

Our main results on the long-term decomposition are presented in Theorems \ref{sdf thm}, \ref{thm:unique}, and \ref{feynmac kac formula reverse} in Section \ref{section 6}.
Theorem \ref{sdf thm} establishes a sufficient condition for achieving the long-term decomposition \eqref{eqn:our_decomp} based on the specified parameters of pricing kernels. The uniqueness of this decomposition is demonstrated in Theorem \ref{thm:unique}. Additionally, Theorem \ref{feynmac kac formula reverse} offers an alternative characterization of the decomposition through a solution to a PDE.
We would like to emphasize that these results are novel, even in the specific case where the $G$-Brownian motion is the standard Brownian motion.

\begin{equation}\label{pic}
\begin{tikzpicture}  [>=stealth,sloped]
\matrix (tree) [
matrix of nodes,
minimum size=0.5cm,
column sep=0.4cm,
row sep=0.1cm,
]
{
	\fbox{Finite-horizon $G$-BSDE}   &     &   \\
	&  & 	\fbox{\begin{varwidth}{\textwidth}
		\centering
		Long-term decomposition \\ of pricing kernels
		\end{varwidth}} \\
	\fbox{Infinite-horizon $G$-BSDE} &   	\fbox{Ergodic $G$-BSDE}  &   \\ 
};
\draw[->] (tree-1-1) -- (tree-2-3) node [midway,above] {};
\draw[->] (tree-3-1) -- (tree-3-2) node [midway,above] {};
\draw[->] (tree-3-2) -- (tree-2-3) node [midway,above] {};
\draw[->] (tree-1-1) -- (tree-3-1) node [midway,above] {};
\end{tikzpicture}
\end{equation}

Several authors have studied $G$-BSDEs in contexts closely related to our work.
\cite{hu2018quadratic} and \cite{hu2022quadratic} investigated the existence and uniqueness of solutions to finite-horizon   quadratic $G$-BSDEs.
\cite{hu2018ergodic} developed a framework for infinite-horizon and ergodic $G$-BSDEs with Lipschitz continuous drivers, while \cite{sun2024g} investigated infinite-horizon and ergodic $G$-BSDEs with quadratic drivers in connection with $G$-forward performance processes. 
Our results are  related to these works, especially to \cite{hu2018ergodic} and \cite{sun2024g}, and several parts of our theorems and proofs are inspired by and build upon their framework. 
Nonetheless, these existing results do not fully encompass our setting. Further comparisons and distinctions are provided in Remarks \ref{remark:finite_BSDEs} and \ref{remark:infinite_BSDEs}.

The long-term decomposition is a powerful tool for analyzing risk premia over extended horizons.
In equilibrium pricing, risk-averse investors require compensation for their risk exposure, which gives rise to risk premia. These premia depend on both risk exposure and the market price of that exposure. 
The decomposition  provides 
insight into how the
exposure of cash flows and the price of that exposure evolve over the long term.
In addition to the previously cited references,
 the 
literature on the long-term decomposition and its applications includes
\cite{borovivcka2011risk},
\cite{bakshi2012variance},
\cite{hansen2012pricing},
\cite{borovivcka2016term},
\cite{hansen2017stochastic},
\cite{backus2018term} and
\cite{qin2018blong}.

The remainder of this paper is organized as follows. Section \ref{section 2} presents the fundamental concepts of $G$-expectation and demonstrates key results related to Lipschitz $G$-BSDEs. In Section \ref{section 3}, we explore the existence, uniqueness, and regularity of solutions to finite-horizon quadratic $G$-BSDEs. These results are extended to infinite-horizon quadratic $G$-BSDEs in Section \ref{section 4} and to ergodic quadratic $G$-BSDEs in Section \ref{section 5}. Finally, Section \ref{section 6} investigates the existence and uniqueness of the long-term decomposition of pricing kernels and provides a PDE characterization of this decomposition.
The final section summarizes this study.

\section{Preliminary}\label{section 2}
\subsection{\textit{G}-expectation}\label{sec:G_preli}
In this section, we state basic notions and  results of $G$-expectation theory briefly, which are needed in the sequel. The readers may refer to  \cite{peng2007g}, \cite{peng2008multi} and \cite{peng2019nonlinear}. 
Let $\Omega=C_0([0,\infty),\mathbb{R}^d)$ be the space of all $\mathbb{R}^d$-valued continuous functions $(\omega_t)_{t\in[0,\infty) }$ with $\omega_0=0,$ endowed with the distance 
$$\rho(\omega^{(1)},\omega^{(2)}):=\sum_{n=1}^\infty \frac{1}{2^n} (1\wedge \max_{t\in [0,n]}|\omega^{(1)}_t-\omega^{(2)}_t|)$$
and $B$ be the canonical process, namely $B_t(\omega)=\omega(t)$.
For each $t>0$, define a function space 
\begin{align}
Lip(\Omega_t):=\{\varphi(B_{t_1}-B_{t_0},B_{t_2}-B_{t_1},\cdots,B_{t_n}-B_{t_{n-1}})\,|&
\,n\in \mathbb{N},
0\leq t_0 \leq \cdots \le t_n\leq t,\\ &\hspace{20mm}\varphi\in C_{l,Lip}(\mathbb{R}^{d\times n})\} \,,
\end{align}
where
$C_{l,Lip}(\mathbb{R}^{d\times n})$ is the space of all continuous functions $\varphi$ such that there are constants   $C>0$ and $k\in \mathbb{N}$   satisfying $|\varphi(x)-\varphi(y)|\leq C(1+|x|^k+|y|^k)|x-y|$ for all $x,y\in \mathbb{R}^{d\times n}$. 
It is easy to show that $Lip(\Omega_t) \subseteq Lip(\Omega_T)$  for $t\leq T$. We define 
$Lip(\Omega):= \bigcup\limits_{t\geq 0}Lip(\Omega_t).$

Let $\mathbb{S}^d$ be the collection of all $d\times d$ symmetric matrices.
For any given monotonic and sublinear function   $G:\mathbb{S}^d\to \mathbb{R}$,
there is a sublinear expectation space 
$(\Omega, Lip(\Omega), \hat{\mathbb{E}})$ such that
the canonical process $B$ is a $G$-Brownian motion satisfying
$\textit{G}(A)=\frac{1}{2}\hat{\mathbb{E}}[ \langle AB_1, B_1\rangle].$
We say $G$ is the generator of the $G$-Brownian motion.
It can be shown that there exists a bounded and closed set $\Sigma$ of $d\times d$ symmetric positive definite matrices such that
\begin{equation}\begin{aligned}\label{gamma}
G(A)=\frac{1}{2}\sup_{Q\in\Sigma}\text{tr}(AQ)\,.
\end{aligned}\end{equation}
Occasionally, for $A = (a_{ij})_{ i, j =1}^d$, we write $G(a_{ij})$ instead of $G(A)$ for notational simplicity.
We define a  sublinear conditional expectation $\hat{\mathbb{E}}_t$ of $X\in Lip(\Omega)$   as
$$\hat{\mathbb{E}}_t[X] =\hat{\mathbb{E}}[\varphi(x_1,\cdots,x_j,B_{t_{j+1}}-B_{t_j} ,\cdots,B_{t_n}-B_{t_{n-1}})]|_{x_1=B_{t_1}-B_{t_0},\cdots, x_j=B_{t_{j}}-B_{t_{j-1}}}\,,$$
where $X$ is expressed as $X=\varphi(B_{t_1}-B_{t_0},B_{t_2}-B_{t_1},\cdots,B_{t_{j+1}}-B_{t_j} ,\cdots,B_{t_n}-B_{t_{n-1}})$
with $\varphi\in C_{l,Lip}(\mathbb{R}^{d\times n})$ and $t=t_j$ for some $j=0,\cdots, n-1.$
The sublinear expectation space $(\Omega, Lip(\Omega), \hat{\mathbb{E}})$,
the sublinear expectation $\hat{\mathbb{E}}$ 
and the sublinear conditional expectation $\hat{\mathbb{E}}_t$ 
are called a $G$-expectation space, a  $G$-expectation, a conditional $G$-expectation  respectively.

For $p\geq 1$, we define $\mathbb{L}_G^p(\Omega_t)$ (respectively, $\mathbb{L}_G^p(\Omega)$) as the completion of $Lip(\Omega_t)$ (respectively, $Lip(\Omega)$) with respect to the norm $|\!|\xi|\!|_{p,G}:= (\hat{\mathbb{E}}[|\xi|^p])^{\frac{1}{p}}$. For each $t\geq 0,$ a sublinear conditional expectation can be extended continuously to $\hat{\mathbb{E}_t}:\mathbb{L}^p_G(\Omega)\to \mathbb{L}^p_G(\Omega_t)$.
The product space of $m$ copies of 
$\mathbb{L}_G^p(\Omega_t)$
is denoted as 
$\mathbb{L}_G^p(\Omega_t;\mathbb{R}^m).$
In this study, we shall only
consider   $G$-Brownian motions satisfying the strong ellipticity condition, that is, 
there exist  strictly positive constants $\overline{\sigma}, \underline{\sigma}$ such that
\begin{equation}\label{ellip}
\frac{1}{2}\underline{\sigma}^2\text{tr}(A-B)\leq G(A)-G(B)\leq \frac{1}{2}\overline{\sigma}^2\text{tr}(A-B)\text{ for } A\geq B\,.
\end{equation}

The conditional $G$-expectation has a probabilistic representation.
Refer to \cite{hu2009representation}, \cite{denis2011function}, and \cite{hu2021extended}       
for the following representation theorem.
Let $(\mathcal{F}_t)_{t\ge0}$ be the natural filtration on the canonical space $\Omega$  and  $\mathcal{F}$ be the $\sigma$-algebra generated by  $\cup_{t\ge0}\mathcal{F}_t.$

\begin{thm}\label{conditonal sublinear representation} 
	There exists a weakly compact set $\mathcal{P}$ of probability measures   on $(\Omega,\mathcal{F})$  such that 
	\begin{equation}
	\hat{\mathbb{E}}[X]=\sup_{P\in \mathcal{P}}E^P[X] 
    \text{\, for all } X \in \mathbb{L}^1_G(\Omega). 
	\end{equation}
	We say $\mathcal{P}$ is a set that represents $\hat{\mathbb{E}}$. Moreover, the sublinear conditional expectation $\hat{\mathbb{E}}_s$ satisfies
	\begin{equation} 
	\begin{aligned}
	\hat{\mathbb{E}}_s[X]=\operatorname*{ess\,\sup}_{Q \in \mathcal{P}(s,P)}E^Q[X|\mathcal{F}_s]\hspace{1cm} P\textnormal{-almost surely\,,}   
	\end{aligned}	
	\end{equation}
	where $\mathcal{P}(s,P):= \{Q\in \mathcal{P}\,|\, E^Q[X]=E^P[X]\text{\, for all } X\in Lip(\Omega_s)\}$.
\end{thm}

We define the capacity as 
\begin{equation}
c(A):= \sup_{P\in \mathcal{P}} P(A),\hspace{0.2cm} A\in \mathcal{F}\,.
\end{equation}
A set $A\in \mathcal{F}$ is called polar if $c(A)=0$. We say a property holds ``quasi-surely" if it holds outside a polar set. Unless explicitly stated otherwise, all subsequent inequalities and equations hold quasi-surely.
We say two sublinear expectations
	$\hat{\mathbb{E}}$ and $\tilde{\mathbb{E}}$ on $(\Omega, Lip(\Omega))$ are equivalent if $\hat{\mathbb{E}}[\mathds{1}_A]=0$ if and only if $\tilde{\mathbb{E}}[\mathds{1}_A]=0$ for all $A\in \mathcal{F}$. 
A random variable $X:\Omega \to \mathbb{R}$ is said  to  be quasi-continuous if for any $\epsilon>0$, there exists an open set $O$ such that $c(O)<\epsilon$ and $X|_{O^c}$ is continuous.
Let $\mathcal{B}_b(\Omega)$ be the set of all bounded $\mathcal{F}$-measurable real-valued functions, and define $\mathbb{L}^p_b(\Omega)$ as the completion of $\mathcal{B}_b(\Omega)$ with respect to the norm $|\!| \cdot|\!|_{p,G}$. \cite{denis2011function}  proved that  
\begin{equation} 
\begin{aligned}
\mathbb{L}^p_b(\Omega)=\{X:\Omega\to\mathbb{R}\,|\,X\text{ is }\mathcal{F}\text{-measurable, } \hat{\mathbb{E}}[|X|^p]<\infty\text{ and } \lim_{n\rightarrow\infty}\hat{\mathbb{E}}[|X|^p\mathds{1}_{\{|X|>n\}}]=0\}.
\end{aligned}
\end{equation}
This implies that to demonstrate $Y\in \mathbb{L}_b^1(\Omega)$,  it suffices to show that  $Y\in \mathbb{L}_G^p(\Omega)$ for $p>1$.  We now present several properties of random variables, which are used later.
Refer to \cite{hu2018quadratic} and \cite{hu2019convergences}
for the proofs of the following two propositions.

\begin{prop}\label{relation P and sublinear} Let $X$ and $Y$ be random variables with $E^P[|X|+|Y|]<\infty$ for all $P\in\mathcal{P}.$ Then  following statements are equivalent.
	\begin{enumerate}[label=(\roman*)]
		\item $X=Y$ $P$-almost surely for all $P\in \mathcal{P}$. 
		\item  $\hat{\mathbb{E}}[|X-Y|]=0$.
		\item $c(\{X\neq Y\})=0$.
	\end{enumerate}
\end{prop}

\begin{prop}\label{in sublinear DCT} 
	\sloppy Let $(X_n)_{n=1}^\infty$ be a sequence of random variables satisfying $\sup_{P\in\mathcal{P}}\mathbb{E}^P[|X_n|]<\infty$ for all $n\ge1.$ Suppose there exists $Y\in \mathbb{L}^1_b(\Omega)$ such that $|X_n|\leq Y$ for all $n\geq 1$ and $X_n\rightarrow X$ as $n\to\infty$ quasi-surely. Then $X_n$ converges to $X$ in $\mathbb{L}_{b}^1(\Omega).$
\end{prop}

\begin{thm}\label{similar doob maximal}  For any $\alpha\geq1$, $\delta>0$, $1<\gamma <\beta:=\frac{\alpha+\delta}{\alpha}$, $\gamma \leq 2$, we have 
	\begin{equation}
	\hat{\mathbb{E}}[\sup_{0\leq s\leq T}\hat{\mathbb{E}}_s[|\xi|^{\alpha}]]\leq \frac{\gamma}{\gamma-1} ((\hat{\mathbb{E}}[|\xi|^{\alpha+\delta}])^\frac{\alpha}{\alpha+\delta} +14^{\frac{1}{\gamma}}C_{\beta,\gamma}(\hat{\mathbb{E}}[|\xi|^{\alpha+\delta}])^\frac{1}{\gamma})  
	\end{equation}
	for all   $\xi \in \mathbb{L}^{\alpha+\delta}_G(\Omega_T)$, where $C_{\beta,\gamma}=\sum_{j=1}^{\infty}j^{-\frac{\beta}{\gamma}}$. 
\end{thm}

\begin{proof}
	The inequality was proven in  \cite[Theorem 3.4]{song2011some} for all 
	$\xi\in Lip(\Omega_T),$
	and it can be easily extended for all $\xi \in \mathbb{L}^{\alpha+\delta}_G(\Omega_T)$.
\end{proof}

We now consider several process spaces and norms. 
For $T>0$ and $p\ge1,$ define
\begin{equation} 
\begin{aligned}
\mathbb{M}^{p,0}(0,T)&=\big\{ \eta_t(\omega)=\sum_{k=1}^{n}\xi_k(\omega) \mathds{1}_{[t_k,t_{k+1})}(t) \,|\,  n \in \mathbb{N},\,k=1,2,\cdots,n , \\
&\hspace{60mm} 0\le t_k< t_{k+1}\leq T,\; \xi_k\in \mathbb{L}_G^p(\Omega_{t_k}) \big\}\,,\\
\mathbb{S}^0(0,T)&=\big\{\varphi(\,\cdot\,,B_{t_1\wedge \cdot},\cdots,B_{t_n\wedge \cdot})\,|\, n\in\mathbb{N}, \,  t_1,\cdots,t_n\in[0,T], \varphi\in C_{b,Lip}(\mathbb{R}^{d\times n+1}) \big\}
\end{aligned}
\end{equation}
with norms   
\begin{equation}\begin{aligned}
|\!|\eta|\!|_{\mathbb{M}_G^p}&=(\hat{\mathbb{E}}[\int_0^T|\eta_u|^p\,du])^{\frac{1}{p}},\,
|\!|\eta|\!|_{\mathbb{H}_G^p}=(\hat{\mathbb{E}}[(\int_0^T|\eta_u|^2\,du)^{\frac{p}{2}}])^{\frac{1}{p}},\,
|\!|\eta|\!|_{\mathbb{S}_G^p}=(\hat{\mathbb{E}}[\sup_{0\leq u\leq T}|\eta_u|^p])^{\frac{1}{p}} \,.
\end{aligned}\end{equation}
Here, $C_{b,Lip}(\mathbb{R}^{d\times n+1})$ is the space of all bounded Lipschitz continuous functions on $\mathbb{R}^{d\times n+1}.$ 
Denote the  completions of $\mathbb{M}^{p,0}(0,T)$ with respect to norms $|\!|\cdot|\!|_{\mathbb{M}_G^p}$ and  $|\!|\cdot|\!|_{\mathbb{H}_G^p}$ by  $\mathbb{M}_G^p(0,T)$ and  $\mathbb{H}_G^p(0,T)$ respectively and  the  completion of $\mathbb{S}^{0}(0,T)$ with respect to norm $|\!|\cdot|\!|_{\mathbb{S}_G^p}$  by $\mathbb{S}_G^p(0,T)$. 
We can also define the spaces $\mathbb{M}_G^p(t,T),$  $\mathbb{H}_G^p(t,T),$ $\mathbb{S}_G^{p}(t,T)$ for $0\le t<T$ similarly.
The product spaces of $m$ copies of 
$\mathbb{M}_G^p(t,T),$  $\mathbb{H}_G^p(t,T),$ $\mathbb{S}_G^{p}(t,T)$
are denoted as $\mathbb{M}_G^p(t,T;\mathbb{R}^m),$  $\mathbb{H}_G^p(t,T;\mathbb{R}^m),$ $\mathbb{S}_G^{p}(t,T;\mathbb{R}^m),$ respectively.
Moreover, we define 
$
\mathbb{M}^p_G(0,\infty;\mathbb{R}^m)=\cap_{T>0}\mathbb{M}^p_G(0,T;\mathbb{R}^m),$
$\mathbb{H}^p_G(0,\infty;\mathbb{R}^m)=\cap_{T>0}\mathbb{H}^p_G(0,T;\mathbb{R}^m),$
$\mathbb{S}^p_G(0,\infty;\mathbb{R}^m)=\cap_{T>0}\mathbb{S}^p_G(0,T;\mathbb{R}^m)
$.
For a $G$-Brownian motion   $B=(B^1,\cdots, B^d)$,  the quadratic covariation process 
of $B^i$ and $B^j$ is 
defined as  $\langle B^i,B^j\rangle=(\langle B^i,B^j\rangle_s)_{s\ge0}$, where
	\begin{align}
	\langle B^i,B^j\rangle_s:=\lim_{|\pi^n|\to 0}\sum_{k=0}^{\ell(n)-1}(B^i_{s^n_{k+1}}-B^i_{s^n_k})(B^j_{s^n_{k+1}}-B^j_{s^n_k})\textnormal{ in } \mathbb{L}^2_G(\Omega)\,.
	\end{align}
Here, $\pi^n = \{ 0 = s^n_0 < s^n_1 < \cdots < s^n_{\ell(n)} = s \}$ is a sequence of partitions of the interval $[0, s]$, and the mesh size is defined as
$
|\pi^n| := \max\{ |s^n_{k+1} - s^n_k| \;|\; k = 0, 1, \dots, \ell(n) - 1 \}.$
From  \cite{peng2019nonlinear}, for  $p\geq 1$ and $\eta\in \mathbb{H}_G^p(t,T)$, we can define the  integrals    $\int_t^T\eta_u\,dB^i_u$ and $\int_t^T\eta_u\,d \langle B^i,B^j\rangle_u$ for  $i,j=1,2,\cdots d$.
The following proposition states the  Burkholder-Davis-Gundy inequality under the $G$-expectation.

\begin{prop}\label{BDG}
	For $p\geq1$, there are positive constants $c_p$ and $C_p$ such that
	\begin{equation}
	\underline{\sigma}^pc_p\hat{\mathbb{E}}\Big[\Big(\int_0^T|\eta_u|^2\,du\Big)^{\frac{p}{2}}\Big]\leq \hat{\mathbb{E}}\Big[\sup_{0\leq s\leq T}\Big|\int_0^s\eta_u\,dB_u\Big|^p\Big]\leq \overline{\sigma}^pC_p\hat{\mathbb{E}}\Big[\Big(\int_0^T|\eta_u|^2\,du\Big)^{\frac{p}{2}}\Big]  
	\end{equation}
	for all $\eta \in \mathbb{H}_G^p(0,T;\mathbb{R}^d).$
\end{prop}

We now state the definition and a useful property of
$G$-martingales. 
The proof of Lemma \ref{int with respect to K is martingale} is stated in \cite{hu2014backward}.
\begin{defi}
	We say a process $M=(M_t)_{t\ge0}$ is a $G$-martingale if
	$M_t\in \mathbb{L}_G^1(\Omega_t)$ for   $t\ge0$ and  $\hat{\mathbb{E}}_s[M_t]=M_s$ for  $0\le s\leq t.$ A process $M$ is called a symmetric $G$-martingale if both $M$ and $-M$ are $G$-martingales.
\end{defi}

\begin{lemma} \label{int with respect to K is martingale}
	Let $X\in \mathbb{S}_G^{\alpha}(0,T)$ for $\alpha> 1$ and $\beta:=\frac{\alpha}{\alpha-1}$. If    $K^i$ is a decreasing $G$-martingale with $K^i_0=0$ and $K_T^i\in \mathbb{L}^{\beta}_G(\Omega_T)$  for $i=1,2$, then the   process  
	\begin{equation}
	\int_0^s X_u^{+}\,dK_u^1+\int_0^s X_u^{-}\,dK_u^2\,,\;0\le s\le T
	\end{equation}
	is a decreasing $G$-martingale, where $X^{+}=\max(0,X)$ and $X^{-}=\max(0,-X)$.
\end{lemma}

\subsection{Lipschitz G-BSDE}

This section states existing results on Lipschitz $G$-BSDEs. 
Given $0\le t\le T$ and $\xi\in \mathbb{L}_G^2(\Omega_t;\mathbb{R}^m)$,
consider the $G$-SDE and $G$-BSDE
\begin{align}\label{FKformula SDE}
X_s^{t,\xi}&=\xi+\int_t^s b(u,X_u^{t,\xi})\,du+\int_t^s h_{ij}(u,X_u^{t,\xi})\,d\langle B^i,B^j\rangle_u+\int_t^s \sigma(u,X_u^{t,\xi})\,dB_u\,,
\\
\label{FKformula BSDE}
Y_s^{t,\xi}&=\Phi(X_T^{t,\xi})+\int_s^Tf(u,X_u^{t,\xi},Y_u^{t,\xi},Z_u^{t,\xi})\,du\\&
\quad+\int_s^Tg_{ij}(u,X_u^{t,\xi},Y_u^{t,\xi},Z_u^{t,\xi})\,d\langle B^i,B^j\rangle_u-\int_s^TZ_u^{t,\xi}\,dB_u-(K_T^{t,\xi}-K_s^{t,\xi}) 
\end{align}
for $t\le s\le T$,
where  $b,h_{ij}:[0,T]\times\mathbb{R}^m\to\mathbb{R}^m$, $\sigma:[0,T]\times\mathbb{R}^m\to \mathbb{R}^{m\times d},$ 
$\Phi:\mathbb{R}^m\to \mathbb{R},$   $f,g_{ij}:[0,T]\times \mathbb{R}^m\times \mathbb{R}\times \mathbb{R}^d\to
\mathbb{R}.$
We also write $X^{t,\xi},Y^{t,\xi},Z^{t,\xi},K^{t,\xi}$ as $X,Y,Z,K,$ respectively, omitting the superscripts $t,\xi.$ 
The following theorem 
demonstrates the existence and uniqueness of solutions to the above $G$-SDE and $G$-BSDE with the  Feynman-Kac formula. The proof is stated in 
\cite{hu2014backward} and \cite{hu2014comparison}.

\begin{thm}\label{linear FeyKac formula}
	Suppose $h_{ij}=h_{ji},$ $g_{ij}=g_{ji}$ and $b,h_{ij},\sigma,g_{ij},f,\Phi$ are continuous in time variable and uniformly Lipschitz in variables $(x,y,z)$. Then  there exists a unique solution $(X,Y,Z,K)=(X^{t,\xi},Y^{t,\xi},Z^{t,\xi},K^{t,\xi})$ to \eqref{FKformula SDE} and \eqref{FKformula BSDE} such that $(X,Y, Z)\in\mathbb{S}_G^2(t,T;\mathbb{R}^m)\times\mathbb{S}_G^2(t,T) \times \mathbb{H}_G^2(t,T;\mathbb{R}^d)$ and $K$ is a decreasing $G$-matingale with $K_0=0$ and $K_T\in \mathbb{L}^2_G(\Omega_T).$
	Moreover, if we define $u(t,x)=Y_t^{t,x}$ for $(t,x)\in [0,T]\times\mathbb{R}^m$, then $u$ is a unique viscosity solution to the PDE
	\begin{equation}\begin{aligned}
	\begin{cases}\; 
	\partial_t u+F(t,x,u,D_xu,D_x^2u)=0\,, \\
	\; u(T,x)=\Phi(x)\,,
	\end{cases}
	\end{aligned}\end{equation}
	such that for some constant $L>0$   
	\begin{equation}\begin{aligned}\label{eqn:reg_growth}
	|u(t,x)-u(t,x')|\leq L|x-x'|\;,\quad
	|u(t,x)|\leq L(1+|x|)
	\end{aligned}\end{equation}
	for all $t\in [0,T]$, $x,x'\in \mathbb{R}^m$,	
	where 
	\begin{equation}\begin{aligned}
	F(t,x,u,D_xu,D_x^2u)&=G(H(t,x,u,D_xu,D_x^2u))+\langle b(t,x),D_xu\rangle\\&
    \quad+f(t,x,u,\langle\sigma^1(t,x),D_xu\rangle,\cdots,\langle\sigma^d(t,x),D_xu\rangle)\,,\\
	H_{ij}(t,x,u,D_xu,D_x^2u)&=\langle D^2_xu\sigma^i(t,x),\sigma^j(t,x)\rangle+2\langle D_xu,h_{ij}(t,x)\rangle\\
	&\quad+2 g_{ij}(t,x,u,\langle\sigma^1(t,x),D_xu\rangle,\cdots,\langle\sigma^d(t,x),D_xu\rangle)\,,
	\end{aligned}\end{equation}
	and $\sigma^i$  is the $i$-th column vector of $\sigma$ for $i=1,2,\cdots,d.$
\end{thm}

As a special case, we consider linear $G$-BSDEs.  Suppose $f(s,x,y,z)=a_sy+b_sz+m_s$, 
$g_{ij}(s,x,y,z)=c_s^{ij}y+d_s^{ij}z+n^{ij}_s$, where $a, c^{ij},m,n^{ij}\in \mathbb{M}^2_G(0,T)$, $b,d^{ij}\in \mathbb{M}^2_G(0,T;\mathbb{R}^d).$
Moreover, we assume that the processes $a,b,c^{ij},d^{ij}$ are bounded and $\Phi(X_T^{t,\xi})\in\mathbb{L}^2_G(\Omega_T)$.
Then,
a solution $Y^{t,\xi}$
to \eqref{FKformula BSDE}  can be expressed as follows.
Construct an auxiliary extended $\Tilde{\text{G}}$-expectation space $(\Tilde{\Omega}, {\mathbb{L}}_G^1(\tilde{\Omega}_T), \Tilde{\mathbb{E}})$  with  $\tilde{\Omega}=C_0([0,\infty),\mathbb{R}^{2d}) $ and 
\begin{equation}\begin{aligned}
\Tilde{G}(A)=\frac{1}{2}\sup_{Q\in \Sigma}\text{tr}\Big(A\Big[\begin{matrix}
Q&I_d\\I_d&Q^{-1} 
\end{matrix}\Big]\Big)\,,\;A \in \mathbb{S}^{2d} \label{tildeG}\,,
\end{aligned}\end{equation}
where $\Sigma$ is the set stated in \eqref{gamma}.
Then the canonical process $(B,\Tilde{B})$ is a $\Tilde{G}$-Brownian motion and $\langle B, \Tilde{B} \rangle_s =s I_d$ for $s\ge0.$
Let $\Gamma=(\Gamma_s)_{0\le s\le T}$ be a solution to the $\tilde{G}$-SDE 
\begin{equation}\begin{aligned}\label{linear SDE}
&d\Gamma_s=a_s\Gamma_s\,ds+c^{ij}_s\Gamma_s\,d\langle B^i,B^j\rangle _s+d^{ij}_s\Gamma_s\,dB_s+b_s\Gamma_s\,d\tilde{B}_s\,, \; 0\le s\le T \,,\\
&\Gamma_0=1\,.
\end{aligned}\end{equation}
Then we have
\begin{equation}\label{linear BSDE soln}
Y_s^{t,\xi}=(\Gamma_s)^{-1}\tilde{\mathbb{E}}_s\Big[\Gamma_T\Phi(X_T^{t,\xi})+\int_s^Tm_u\Gamma_u\,du+\int_s^Tn_u^{ij}\Gamma_u\,d\langle B^i,B^j \rangle _u\Big]\,,\;t\le s\le T\,.
\end{equation}
It is noteworthy that  if $\xi \in \mathbb{L}_G^1(\Omega_T)$ then $\tilde{\mathbb{E}}[\xi]= 
\hat{\mathbb{E}}[\xi]$.

Moreover, we can define a new sublinear expectation and a sublinear conditional expectation on the $\tilde{G}$-expectation space $(\Tilde{\Omega}, {\mathbb{L}}_G^1(\tilde{\Omega}_T), \Tilde{\mathbb{E}})$. Let $M^{\tilde{b},\tilde{d}^{ij}}$ be a solution to the linear $\tilde{G}$-SDE
\begin{equation}\begin{aligned}\label{b,d induced Radon density}
M^{\tilde{b},\tilde{d}^{ij}}=1+\int_0^\cdot\,\tilde{d}_u^{ij}M^{\tilde{b},\tilde{d}^{ij}}_u\,dB_u+\int_0^\cdot \tilde{b}_uM^{\tilde{b},\tilde{d}^{ij}}_u\,d\tilde{B}_u\,,
\end{aligned}\end{equation}
where $\tilde{b},\tilde{d}^{ij}\in \mathbb{M}_G^2(0,T;\mathbb{R}^d)$ for given $T>0$ and are bounded. Then the process $M^{\tilde{b},\tilde{d}^{ij}}$
is a symmetric $G$-martingale. Define a sublinear expectation  $\tilde{\mathbb{E}}^{\tilde{b},\tilde{d}^{ij}}$ and a sublinear conditional expectation 
$\tilde{\mathbb{E}}_s^{\tilde{b},\tilde{d}^{ij}}$ as
\begin{equation}\begin{aligned}\label{M_z}
\tilde{\mathbb{E}}^{\tilde{b},\tilde{d}^{ij}}[\xi]=\tilde{\mathbb{E}}[M^{\tilde{b},\tilde{d}^{ij}}_T \xi] \,,\;\tilde{\mathbb{E}}^{\tilde{b},\tilde{d}^{ij}}_s[\xi]=(M_s^{\tilde{b},\tilde{d}^{ij}})^{-1}\tilde{\mathbb{E}}_s[M^{\tilde{b},\tilde{d}^{ij}}_T \xi] 
\end{aligned}\end{equation} 
for $\xi\in Lip(\tilde{\Omega}_T)$.  
We say  $\tilde{\mathbb{E}}^{\tilde{b},\tilde{d}^{ij}}$ (respectively, $\tilde{\mathbb{E}}^{\tilde{b},\tilde{d}^{ij}}_s$) is the sublinear expectation (respectively, the sublinear conditional expectation)   induced by $\tilde{b}$ and $\tilde{d}^{ij}.$
The sublinear expectations
$\tilde{\mathbb{E}}$ and $\tilde{\mathbb{E}}^{\tilde{b},\tilde{d}^{ij}}$ are equivalent, moreover
the process
\begin{equation}\label{new_GBM}
B_s^{\tilde{b},\tilde{d}^{ij}}:=B_s-\int_0^s \tilde{b}_u\, du-\int_0^s\tilde{d}^{ij}_u\,d\langle B^i,B^j\rangle _u\,,\;0\le s\le T
\end{equation}  is a $G$-Brownian motion under the sublinear expectation $\tilde{\mathbb{E}}^{\tilde{b},\tilde{d}^{ij}}$. It can be easily verified that the quadratic variations of $B$ and $B^{\tilde{b},\tilde{d}^{ij}}$ coincide under both $\tilde{\mathbb{E}}$ and $\tilde{\mathbb{E}}^{\tilde{b},\tilde{d}^{ij}}.$
We can extend consistently the   sublinear expectation
$\tilde{\mathbb{E}}^{\tilde{b},\tilde{d}^{ij}}$
and the sublinear conditional expectation 
$\tilde{\mathbb{E}}^{\tilde{b},\tilde{d}^{ij}}_s$ induced by $\tilde{b}$ and $\tilde{d}^{ij}$
to the $\tilde{G}$-expectation space
$(\Tilde{\Omega}, {\mathbb{L}}_G^1(\tilde{\Omega}), \Tilde{\mathbb{E}})$ 
when     $\tilde{b},\tilde{d}^{ij}\in \mathbb{M}_G^2(0,\infty;\mathbb{R}^d)$  are bounded.
This can be  achieved by the tower property 
$$ \tilde{\mathbb{E}}_t[\tilde{\mathbb{E}}_s[\xi]]=\tilde{\mathbb{E}}_{t}[\xi]\text{ for all } 0\le t<s, \,\xi\in \mathbb{L}_G^1(\tilde{\Omega})$$
of the sublinear conditional expectation.
Refer to \cite{hu2014comparison} and \cite{hu2018stochastic} for more details.

The following lemma states that the decreasing $G$-martingale property remains invariant under the induced sublinear expectation.
We omit the proof of this lemma, as it is similar to that of \cite[Lemma 3.4]{hu2018quadratic}.
 
\begin{lemma}\label{K is MG in new sublinear}
	\sloppy Let $K$ be a decreasing $\tilde{G}$-martingale on the $\tilde{G}$-expectation space 
  $(\tilde{\Omega}, {\mathbb{L}}_G^1(\tilde{\Omega}_T), \tilde{\mathbb{E}})$ with $K_0=0$ and  $K_T\in \mathbb{L}_G^p(\tilde{\Omega}_T)$ for some $p>1,$ and let $\tilde{\mathbb{E}}^{\tilde{b},\tilde{d}^{ij}}$ be the sublinear expectation induced by bounded processes  $\tilde{b},\tilde{d}^{ij}\in \mathbb{M}_G^2(0,T;\mathbb{R}^d)$ satisfying $\tilde{d}^{ij}=\tilde{d}^{ji}$. Then $K$ is a decreasing $\tilde{G}$-martingale under the sublinear expectation $\tilde{\mathbb{E}}^{\tilde{b},\tilde{d}^{ij}}$.
\end{lemma}

\section{Finite-horizon quadratic $G$-BSDEs}\label{section 3}

Given 
$0\le t\le T,$ $\xi\in \mathbb{L}_G^2(\Omega_t;\mathbb{R}^m)$
and measurable functions
$b,h_{ij}:[0,T]\times\mathbb{R}^m\to \mathbb{R}^m$,  $\sigma:[0,T]\times \mathbb{R}^m\to \mathbb{R}^{m\times d} $,
$\Phi:\mathbb{R}^m\to \mathbb{R},$   $f,g_{ij}:[0,T]\times \mathbb{R}^m\times \mathbb{R}\times \mathbb{R}^d\to 
\mathbb{R},$ 
consider the $G$-SDE and 
$G$-BSDE 
\begin{equation}\begin{aligned}
\label{SDE} X_s^{t,\xi}&=\xi+\int_t^s b(u,X_u^{t,\xi})\,du+\int_t^s h_{ij}(u,X_u^{t,\xi})\,d\langle B^i,B^j\rangle_u+\int_t^s \sigma(u,X_u^{t,\xi})\,dB_u\,,
\end{aligned}\end{equation} 
\begin{equation}\begin{aligned}
\label{QBSDE}Y_s^{t,\xi}&=\Phi(X_T^{t,\xi})+\int_s^Tf(u,X_u^{t,\xi},Y_u^{t,\xi},Z_u^{t,\xi})\,du\\
&\quad+\int_s^Tg_{ij}(u,X_u^{t,\xi},Y_u^{t,\xi},Z_u^{t,\xi})\,d\langle B^i,B^j\rangle_u-\int_s^TZ_u^{t,\xi}\,dB_u-(K_T^{t,\xi}-K_s^{t,\xi}) 
\end{aligned}\end{equation}
for $t\le s\le T.$
Occasionally, we write $X^{t,\xi},Y^{t,\xi},Z^{t,\xi},K^{t,\xi}$ as $X,Y,Z,K,$ respectively, omitting the superscripts $t,\xi.$

\begin{assume}\label{assumption} Assume the functions $b,h_{ij},  \sigma, \Phi, f,g_{ij}$ satisfy the following properties.  
	\begin{enumerate}[label=(\roman*)]%[font={\bfseries},label=({A\arabic*})		] 
		\item For $1\leq i,j \leq d$, $h_{ij}=h_{ji}$ and $g_{ij}=g_{ji}.$    \label{finite assumption first}
		\item The functions $b,h_{ij}, \sigma, f,$ $g_{ij}$ are continuous in $s$.\label{finite assumption second}
		\item There are positive constants $C_1,  C_2,C_3, C_{\sigma}, C_{\Phi}, M_\sigma$ such that 
		\begin{equation}\begin{aligned}
		&|b(s,x)-b(s,x')|+\sum_{i,j=1}^d|h_{ij}(s,x)-h_{ij}(s,x')|\leq C_1|x-x'|\,,\\
		&|\sigma(s,x)-\sigma(s,x')|\leq C_\sigma|x-x'|\,,\\
		&|\sigma(s,x)|\leq M_{\sigma}\,,\\
		&|\Phi(x)-\Phi(x')|\leq C_{\Phi}|x-x'|\,,\\
		&|f(s,x,y,z)-f(s,x',y',z')|+\sum_{i,j=1}^d|g_{ij}(s,x,y,z)-g_{ij}(s,x',y',z')|\\ &\hspace{2cm}\leq C_1|x-x'|+C_2|y-y'|+C_3(1+|z|+|z'|)|z-z'|
		\end{aligned}\end{equation}
		for $s\in[t,T],$ $x,x'\in \mathbb{R}^m$, $y,y'\in \mathbb{R}$,   $z,z'\in \mathbb{R}^d.$
		\label{finite assmption third}
		\item There are some constants $\mu$ and $\eta$ such that
		\begin{equation} 
		\begin{aligned}
		&(f(s,x,y,z)-f(s,x,y',z))(y-y')+2G((g_{ij}(s,x,y,z)-g_{ij}(s,x,y',z))(y-y'))\\
		&\quad \leq -\mu|y-y'|^2\,,\\
		&G\Big( \sum_{j=1}^m(\sigma_j(s,x)-\sigma_j(s,x'))^{\top} (\sigma_j(s,x)-\sigma_j(s,x'))  \\&\quad +2(\langle x-x',h_{ij}(s,x)-h_{ij}(s,x')\rangle)_{i,j=1}^d\Big)+\langle x-x',b(s,x)-b(s,x')\rangle\leq -\eta|x-x'|^2\,,\\
		&\mu+\eta>(1+\overline{\sigma}^2)\Bigg(C_\sigma C_3+4C_\Phi C_\sigma C_3\frac{\overline{\sigma} M_\sigma}{\underline{\sigma}}+4\sqrt{C_\sigma C_1C_3\frac{\overline{\sigma}M_{\sigma}}{\underline{\sigma}}}\,\Bigg)
		\end{aligned}		
		\end{equation}
		for $s\in[t,T],$ $x,x'\in \mathbb{R}^m$, $y,y'\in\mathbb{R},z\in\mathbb{R}^d$, where $\sigma_j$ is the $j$-th row vector of $\sigma$ for $j=1,2,\cdots,m.$
	\end{enumerate}
\end{assume}

\begin{defi}
	We say a process $X$ is a solution to the $G$-SDE \eqref{SDE} if 
	$X\in \mathbb{S}_G^2(t,T;\mathbb{R}^m)$ and the process $X$ satisfies \eqref{SDE} quasi-surely.
\end{defi}

The following theorem states the existence, uniqueness and regularity of solutions to the $G$-SDE \eqref{SDE}.
Compared to \cite[Proposition 5.3.1, Corollary 5.3.2]{peng2019nonlinear}, the following theorem provides more detailed estimates illuminating the dependence of the parameters $C_1,C_\sigma,\overline{\sigma},M_{\sigma},C_4,|\!|b(s,0)|\!|,|\!|h_{ij}(s,0)|\!|,T.$
 These detailed estimates are necessary for our analysis later.
We extend the  $G$-expectation space $(\Omega,\mathbb{L}_G^1(\Omega_T),\hat{\mathbb{E}})$ 
\sloppy to the $\tilde{G}$-expectation space $(\tilde{\Omega},\mathbb{L}_{\tilde{G}}(\tilde{\Omega}_T),\tilde{\mathbb{E}})$, where  $\tilde{\Omega}=C_0([0,\infty),\mathbb{R}^{2d})$  and $\tilde{G}$ is given as in  \eqref{tildeG}.
Then the canonical process $(B,\Tilde{B})$ is a $\Tilde{G}$-Brownian motion.
For any bounded processes $b,d^{ij}\in \mathbb{M}^2_G(0,T;\mathbb{R}^d),$
let $M^{b,d^{ij}}$ be the unique solution to \begin{equation}\begin{aligned}\label{b,d density}
M^{b,d^{ij}}=1+\int_0^\cdot\,d^{ij}_uM^{b,d^{ij}}_u\,dB_u+\int_0^\cdot b_uM^{b,d^{ij}}_u\,d\tilde{B}_u
\end{aligned}\end{equation} 
and define  
the sublinear conditional expectation $\tilde{\mathbb{E}}_t^{b,d^{ij}}$ by the process $M^{b,d^{ij}}$ as presented in \eqref{M_z}.

\begin{thm}\label{SDE sublinear property} 
	Let Assumptions \ref{assumption}  hold. Then for $t>0,$ $p\geq 1$ and $\xi\in \mathbb{L}_G^2(\Omega_t;\mathbb{R}^m)$,
	there exists a unique solution $X^{t,\xi}$ to \eqref{SDE}.   For any processes $b,d^{ij}\in \mathbb{M}^2_G(0,T;\mathbb{R}^d)$ satisfying $d^{ij}=d^{ji}$ and $|b|,|d^{ij}|\le C_4$ for some constant $C_4>0,$
	there exists a constant $L>0$ depending only on $C_1,C_\sigma,\overline{\sigma},M_{\sigma},C_4,|\!|b(s,0)|\!|,|\!|h_{ij}(s,0)|\!|$ and $T$ 
	such that 
    \begin{equation}
    \begin{aligned}
	&\tilde{\mathbb{E}}^{b,d^{ij}}_t[|X_{s}^{t,\xi}-X_{s}^{t,\xi'}|]\leq 
	e^{-(\eta-(1+\overline{\sigma}^2)C_\sigma C_4)(s-t)}|\xi-\xi'|\,,\\
	&\tilde{\mathbb{E}}^{b,d^{ij}}_t[|X_{s}^{t,\xi}|]\leq L(1+|\xi|)\,,\\
	&\tilde{\mathbb{E}}^{b,d^{ij}}_t[|X_s^{t,\xi}-\xi|]\leq L(1+|\xi|)(s-t)+M_{\sigma}\overline{\sigma}(s-t)^{\frac{1}{2}}\,
	\end{aligned}\label{eqn:x esi finite}
    \end{equation}
	for all $\xi,\xi'\in  \mathbb{L}_G^p(\Omega_t;\mathbb{R}^m)$
	and $t\le s\le T.$ 
\end{thm}
\begin{proof}
The existence and uniqueness of solutions to the SDE \eqref{SDE} can be established by the standard arguments (see, for example, \cite[Proposition 4.1]{hu2014comparison} and \cite[Theorem 5.1.3]{peng2019nonlinear}). 
The proof of the first inequality in \eqref{eqn:x esi finite} is similar to that of \cite[Lemma 4.1]{hu2018ergodic} and the second inequality follows from \cite[Proposition 4.1]{hu2014comparison}.

For the last inequality, we assume for simplicity  that $\xi:=x\in \mathbb{R}^m$ is a constant vector.
Observe that \eqref{SDE} yields
	\begin{align}
	\tilde{\mathbb{E}}_t^{b,d^{ij}}[|X_s^{t,x}-X_t^{t,x}|]	&\leq \tilde{\mathbb{E}}_t^{b,d^{ij}}\Big[\int_t^s|b(u,X_u^{t,x})+\sigma(u,X_u^{t,x})b_u|\,du\Big]\\
    &\quad
    +\tilde{\mathbb{E}}_t^{b,d^{ij}}\Big[\Big|\int_t^s h_{k\ell}(u,X_u^{t,x})+\sigma(u,X_u^{t,x})d^{k\ell}_u\,d\langle B^{b,d^{ij},k}, B^{b,d^{ij},\ell}\rangle_u\Big|\Big]\\
    &\quad+\tilde{\mathbb{E}}_t^{b,d^{ij}}\Big[\Big|\int_t^s\sigma(u,X_u^{t,x})\,dB^{b,d^{ij}}_u\Big|\Big]\,,
	\end{align} 
   where $B^{b,d^{ij}}=B-\int_0^\cdot b_u\,du-\int_0^\cdot d^{ij}_u\langle B^i,B^j\rangle _u$.
  Since $b$ and $h_{k\ell}$ have linear growth in $x$ and $\sigma$ is bounded, it follows that 
   \begin{align}
       &\tilde{\mathbb{E}}_t^{b,d^{ij}}\Big[\int_t^s|b(u,X_u^{t,x})+\sigma(u,X_u^{t,x})b_u|\,du\Big]\le  L(1+|X_t^{t,x}|)(s-t) \,,\\
       &\tilde{\mathbb{E}}_t^{b,d^{ij}}\Big[\Big|\int_t^s h_{k\ell}(u,X_u^{t,x})+\sigma(u,X_u^{t,x})d^{k\ell}_u\,d\langle B^{b,d^{ij},k}, B^{b,d^{ij},\ell}\rangle_u\Big|\Big]
       \le L(1+|X_t^{t,x}|)(s-t)       
         \end{align}
       and
       \begin{align}
     \tilde{\mathbb{E}}_t^{b,d^{ij}}\Big[\Big|\int_t^s\sigma(u,X^{t,x}_u)\,dB^{b,d^{ij}}_u\Big|\Big]&\le \Big(\tilde{\mathbb{E}}_t^{b,d^{ij}}\Big[\Big|\int_t^s\sigma(u,X_u^{t,x})\,dB^{b,d^{ij}}_u\Big|^2\Big]\Big)^{\frac{1}{2}}\\
    &\le\Big(\tilde{\mathbb{E}}_t^{b,d^{ij}}\Big[\int_t^s\Big((\sigma^{\top}\sigma)(u,X^{t,x}_u)\Big)_{k\ell}\,d\langle B^{b,d^{ij},k},B^{b,d^{ij},\ell}\rangle_u\Big]\Big)^{\frac{1}{2}}\\
       &\le \Big(\tilde{\mathbb{E}}_t^{b,d^{ij}}\Big[\int_t^s\overline{\sigma}^2|\sigma(u,X_u^{t,x})|^2\,du \Big]\Big)^{\frac{1}{2}}\le M_\sigma\overline{\sigma}(s-t)^{\frac{1}{2}}\,.
   \end{align}
Here we have used that for each component of $(\zeta^{k\ell v}_u)_{1\le v\le m}=h_{k\ell}(u,X_u^{t,x})+\sigma(u,X_u^{t,x})d^{k\ell}_u\in \mathbb{H}^2_G(0,T;\mathbb{R}^m)$,
        \begin{align}
        \sum_{1\le k,\ell\le d}\int_t^s \zeta^{k\ell v}_u\,d\langle B^{b,d^{ij},k},B^{b,d^{ij},\ell}\rangle_u
        &\le\int_t^s2G(\zeta^{v}_u)\,du\\
        &\le\int_t^s2 G(D^v_u)\,du\le \sum_{1\le k,\ell\le d}\int_t^s \overline{\sigma}^2|\zeta^{k\ell v}_u|\,du\,,
        \end{align}
        and
        \begin{align}   
        \sum_{1\le k,\ell\le d} \int_t^s \zeta^{k\ell v}_u\,d\langle B^{b,d^{ij},k},B^{b,d^{ij},\ell}\rangle_u\ge - \sum_{1\le k,\ell\le d}\int_t^s \overline{\sigma}^2|\zeta^{k\ell v}_u|\,du\,,
    \end{align}
    where $D_u^v$ is a diagonal matrix whose $(k,k)$-entry is $\sum_{\ell=1}^d|\zeta^{k\ell v}_u|$ for $k=1,2,\cdots d$. 
Then
	\begin{equation}\begin{aligned}
	\tilde{\mathbb{E}}_t^{b,d^{ij}}[|X_s^{t,x}-X_t^{t,x}|] \leq  L(1+|X_t^{t,x}|)(s-t)+M_{\sigma}\overline{\sigma}(s-t)^{\frac{1}{2}}\,,
	\end{aligned}\end{equation}
	where $L$ depends only on $C_1,C_4,M_{\sigma},\overline{\sigma},|\!|b(s,0)|\!|,|\!|h_{ij}(s,0)|\!|$ and $T.$
\end{proof}

We now focus on the $G$-BSDE \eqref{QBSDE} and demonstrate several properties of solutions to the $G$-BSDE.

\begin{defi} We say a triplet $(Y,Z,K)$ is a solution to the $G$-BSDE \eqref{QBSDE} if 
	\begin{enumerate}[label=(\roman*)]
		\item the tuple $(Y,Z)\in \mathbb{S}_G^{2}(t,T)\times \mathbb{H}_G^{2}(t,T;\mathbb{R}^d)$,
		\item the process $K$ is a decreasing $G$-martingale with $K_0=0$ and $K_T\in \mathbb{L}_G^2(\Omega_T)$,
		\item the triplet $(Y,Z,K)$ satisfies \eqref{QBSDE} quasi-surely.
	\end{enumerate}
\end{defi}

The following theorem 
states the existence, uniqueness and regularity of solutions to the $G$-BSDE \eqref{QBSDE}.
The independence of the parameter 
$T$
on the Lipschitz constant $L_1$ 
in \eqref{eqn:Lip} is crucial for our analysis later.

\begin{thm}\label{Finite QBSDE exis and unique}
	Let Assumptions \ref{assumption} hold. For any $0\le t\le T$ and $\xi\in \mathbb{L}_G^p(\Omega_t;\mathbb{R}^m),$ there exists a unique solution $(Y^{t,\xi},Z^{t,\xi},K^{t,\xi})$ to the $G$-BSDE \eqref{QBSDE} such that $Z^{t,\xi}$ is  bounded. Moreover if we define $u(t,x)=Y_t^{t,x}$ for $(t,x)\in[0,T]\times \mathbb{R}^m,$ then $u$ satisfies the inequalities
	\begin{equation}\begin{aligned}\label{eqn:Lip}
	|u(t,x)-u(t,x')|&\leq L_1|x-x'|\,,\\
	|u(t+\delta,x)-u(t,x)|&\leq L_2(1+|x|)\delta+L_1\delta^{\frac{1}{2}}\,,
	\end{aligned}\end{equation}
	where $L_1$ is a constant  depending only on  $\mu,\eta, C_1,C_3,C_\Phi,C_\sigma,M_{\sigma},\overline{\sigma}$ and 
	$L_2$ is a constant  depending only on  $\mu,\eta, C_1,C_3,C_\Phi,C_\sigma,M_{\sigma},\overline{\sigma},|\!|b(s,0)|\!|,|\!|h_{ij}(s,0)|\!|,$ $T$.
\end{thm}

\begin{proof}
    For simplicity, assume $\xi=x\in\mathbb{R}^m$ is a constant vector.
    We first consider the case with $\mu<0.$
	Let $$M:=\frac{\eta+\mu-(1+\overline{\sigma}^2)C_\sigma C_3+4C_\Phi(1+\overline{\sigma}^2)C_{\sigma}C_3\frac{\overline{\sigma}M_\sigma }{\underline{\sigma}}   }{4(1+\overline{\sigma}^2)C_\sigma C_3}\,,$$
	then  $M$ is a positive constant satisfying $\eta+\mu-(1+\overline{\sigma}^2)C_{\sigma}C_3(1+2M)>0.$
	Define $f^M(s,x,y,z):=f(s,x,y,z^M)$, $g_{ij}^M(s,x,y,z):=g_{ij}(s,x,y,z^M)$, where $z^{M}:=\frac{|z|\wedge M}{|z|}z$ with the convention $\frac{\,0\,}{0}=0.$  Then the Lipschitz  $G$-BSDE
	\begin{equation}\begin{aligned}\label{eqn: truncation M}
	Y_s^{M,t,x}=\Phi(X^{t,x}_T)&+\int_s^Tf^M(u,X^{t,x}_u,Y^{M,t,x}_u,Z_u^{M,t,x}) \,du\\&+\int_s^Tg_{ij}^M(u,X^{t,x}_u,Y^{M,t,x}_u,Z_u^{M,t,x}) \,d\langle B^i,B^j\rangle_u \\&-\int_s^TZ_u^{M,t,x}\,dB_u-(K^{M,t,x}_T-K^{M,t,x}_s)\,,\;t\le s\le T
	\end{aligned}\end{equation}
	has
	a unique solution $(Y^{M,t,x},Z^{M,t,x},K^{M,t,x})$
	by Theorem \ref{linear FeyKac formula}.
	Define $u^M(t,x)=Y_t^{M,t,x}$ for $(t,x)\in [0,T]\times\mathbb{R}^m.$

	First we show that the function $u^{M}(t,\cdot)$ is  Lipschitz uniformly in $t$ with Lipschitz constant  depending only on 
	$\mu,\eta, C_1,C_3,C_\Phi,C_\sigma,M_{\sigma},\overline{\sigma}.$ 
	We denote $(X^{t,x},Y^{M,t,x},Z^{M,t,x},K^{M,t,x})$ as $(X^x,Y^{M,x},Z^{M,x},K^{M,x})$ for convenience.
	Let $(\hat{Y},\hat{Z},\hat{K}):=(Y^{M,x}-Y^{M,x'},Z^{M,x}-Z^{M,x'},K^{M,x}-K^{M,x'})$, $\hat{\Phi}:=\Phi(X_T^x)-\Phi(X_T^{x'})$ and  $l^\epsilon(x):=\mathds{1}_{\{|x|\geq \epsilon\}}+\frac{|x|}{\epsilon}\mathds{1}_{\{|x|< \epsilon\}}$ for $\epsilon>0.$  We define the  auxiliary processes
	\begin{equation} 
	\begin{aligned}
	a_s^{\epsilon}&=l^\epsilon(\hat{Y_s})\frac{f^M(s,X_s^x,Y^{M,x}_s,Z^{M,x}_s)-f^M(s,X_s^x,Y_s^{M,x'},Z^{M,x}_s)}{\hat{Y}_s} -\frac{\mu}{1+\underline{\sigma}^2}(1-l^\epsilon(\hat{Y_s})) \,,
	\\
	b_s^{\epsilon}&=l^\epsilon(\hat{Z}_s)\frac{f^M(s,X_s^x,Y^{M,x'}_s,Z^{M,x}_s)-f^M(s,X_s^x,Y_s^{M,x'},Z_s^{M,x'})}{|\hat{Z}_s|^2}\hat{Z}_s\\
    &\quad+C_3(1+2M)(1-l^\epsilon(\hat{Z}_s))\,,\\
	m_s^{\epsilon}&=f^M(s,X_s^x,Y^{M,x}_s,Z^{M,x}_s)-f^M(s,X_s^x,Y_s^{M,x'},Z_s^{M,x'})-a^{\epsilon}_s\hat{Y}_s-b_s^{\epsilon}\hat{Z_s}\,,\\
	h_s&=f^M(s,X_s^x,Y_s^{M,x'},Z_s^{M,x'})-f^M(s,X_s^{x'},Y_s^{M,x'},Z_s^{M,x'})\,,\\
	c_s^{\epsilon,ij}&=l^\epsilon(\hat{Y_s})\frac{g_{ij}^M(s,X_s^x,Y^{M,x}_s,Z^{M,x}_s)-g_{ij}^M(s,X_s^x,Y_s^{M,x'},Z^{M,x}_s)}{\hat{Y}_s} \\&\quad-\frac{\mu}{(1+\underline{\sigma}^2)d}(1-l^\epsilon(\hat{Y_s})) \,,\\
	d_s^{\epsilon,ij}&=l^\epsilon(\hat{Z}_s)\frac{g_{ij}^M(s,X_s^x,Y^{M,x'}_s,Z^{M,x}_s)-g_{ij}^M(s,X_s^x,Y_s^{M,x'},Z_s^{M,x'})}{|\hat{Z}_s|^2}\hat{Z}_s\\
    &\quad+\frac{C_3(1+2M)}{d^2}(1-l^\epsilon(\hat{Z}_s))\,,\\
	n_s^{\epsilon,ij}&=g_{ij}^M(s,X_s^x,Y^{M,x}_s,Z^{M,x}_s)-g_{ij}^M(s,X_s^x,Y_s^{M,x'},Z_s^{M,x'})-c^{\epsilon,ij}_s\hat{Y}_s-d_s^{\epsilon,ij}\hat{Z_s}\,,\\
	k_s^{ij}&=g_{ij}^M(s,X_s^x,Y_s^{M,x'},Z_s^{M,x'})-g_{ij}^M(s,X_s^{x'},Y_s^{M,x'},Z_s^{M,x'})\,.
	\end{aligned}	
	\end{equation}
	It can be easily verified that $ a^\epsilon,  c^{\epsilon,ij}, m^\epsilon,h, n^{\epsilon,ij},k^{ij}$ are in $\mathbb{M}^2_G(0,T)$ and 
    $b^\epsilon,d^{\epsilon,ij}$ are in $\mathbb{M}^2_G(0,T;\mathbb{R}^d).$ By simple calculations, we obtain $|m^\epsilon|\leq (2C_2+2C_3(1+2M))\epsilon,$ 
	$|n^{\epsilon,ij}|\leq (C_2(1+d)+2C_3(1+2M))\epsilon$, $|b^\epsilon|\leq C_3(1+2M)$, $|d^{\epsilon,ij}|\leq C_3(1+2M)$ and $ a^\epsilon+2G(c^{\epsilon,ij})\leq -\mu\,.$ Let $\tilde{\mathbb{E}}^{b^\epsilon,d^{\epsilon,ij}}$ be the sublinear expectation induced by $b^{\epsilon}$ and $d^{\epsilon,ij}$. Then $B^{\epsilon}:=B-\int_0^\cdot b_u^{\epsilon}\,du-\int_0^\cdot d^{\epsilon,ij}_u\,d\langle B^i,B^j\rangle_u$ is a $G$-Brownian motion under $\tilde{\mathbb{E}}^{b^\epsilon,d^{\epsilon,ij}}$
	by \eqref{new_GBM}.
	It follows that
	\begin{equation}\begin{aligned}
	\hat{Y_s}+K_s^{M,x'}&=\hat{\Phi}+K_T^{M,x'}+\int_s^T (a_u^{\epsilon}\hat{Y}_u+m^{\epsilon}_u+h_u)\,du\\
    &\quad+\int_s^T(c_u^{\epsilon,ij}\hat{Y}_u+n^{\epsilon,ij}_u+k_u^{ij})\,d\langle B^{\epsilon,i},B^{\epsilon,j}\rangle_u\\
    &\quad-\int_s^T\hat{Z}_u\,dB^{\epsilon}_u-(K_T^{M,x}-K^{M,x}_s)\,.
	\end{aligned}\end{equation}    
	Let $\Gamma$ be a solution to the $G$-SDE
	\begin{equation}\begin{aligned}
	\Gamma_\cdot=1+\int_0^\cdot a_u^{\epsilon}\Gamma_u\,du+\int_0^\cdot c_u^{\epsilon,ij}\Gamma_u\,d\langle B^{\epsilon,i},B^{\epsilon,j}\rangle_u\,.
	\end{aligned}\end{equation}
	Applying It\^{o}'s formula to $(\hat{Y}+K^{M,x'})\Gamma$, we have
	\begin{equation} 
	\begin{aligned}
	\Gamma_s(\hat{Y}_s+K_s^{M,x'})&=\Gamma_T(\hat{\Phi}+K_T^{M,x'})+\int_s^T\Gamma_u(m_u^{\epsilon}+h_u)\,du-\int_s^Ta_u^{\epsilon} \Gamma_u K_u^{M,x'}\,du\\
    &\quad+\int_s^T \Gamma_u(n_u^{\epsilon,ij}+k_u^{ij})\,d\langle B^{\epsilon,i},B^{\epsilon,j}\rangle_u-\int_s^T\Gamma_uc_u^{\epsilon,ij}K_u^{M,x'}\,d\langle B^{\epsilon,i},B^{\epsilon,j}\rangle_u\\
    &\quad-\int_s^T\Gamma_u\hat{Z}_u\,dB^{\epsilon}_u-\int_s^T\Gamma_u\,dK_u^{M,x}\,.
	\end{aligned}
	\end{equation}
	Then
	$\int_0^\cdot \Gamma_u\hat{Z}_u\,dB^{\epsilon}_u$ is a symmetric $G$-martingale
	and $\int_0^\cdot \Gamma_u\,dK_u^{M,x}$ is a decreasing $G$-martingale
	by Lemmas \ref{int with respect to K is martingale} and \ref{K is MG in new sublinear}
	since $\Gamma\in \mathbb{S}^2_G(0,T)$ is positive and $K^{M,x}_T\in \mathbb{L}^p_G(\Omega_T)$ for every $p\geq 1$.
	It follows that
	\begin{equation}\begin{aligned}\label{eqn:gamma_K}
	&\quad\Gamma_s\hat{Y}_s+\Gamma_sK_s^{M,x'}\\&\leq\tilde{\mathbb{E}}_s^{b^\epsilon,d^{\epsilon,ij}}\Big[\Gamma_T\hat{\Phi}+\int_s^T\Gamma_u(m_u^{\epsilon}+h_u)\,du+\int_s^T \Gamma_u(n_u^{\epsilon,ij}+k_u^{ij})\,d\langle B^{\epsilon,i},B^{\epsilon,j}\rangle_u\Big]\\
	&\quad +\tilde{\mathbb{E}}_s^{b^\epsilon,d^{\epsilon,ij}}\Big[\Gamma_TK_T^{M,x'}-\int_s^Ta_u^{\epsilon} \Gamma_u K_u^{M,x'}\,du-\int_s^T\Gamma_uc_u^{\epsilon,ij}K_u^{M,x'}\,d\langle B^{\epsilon,i},B^{\epsilon,j}\rangle_u\Big]\,.
	\end{aligned}\end{equation}
	Observe that
	\begin{equation}\begin{aligned}
	\Gamma_sK_s^{M,x'}=\tilde{\mathbb{E}}_s^{b^\epsilon,d^{\epsilon,ij}}\Big[\Gamma_TK_T^{M,x'}-\int_s^Ta_u^{\epsilon} \Gamma_u K_u^{M,x'}\,du-\int_s^T\Gamma_uc_u^{\epsilon,ij}K_u^{M,x'}\,d\langle B^{\epsilon,i},B^{\epsilon,j}\rangle_u\Big]\,,
	\end{aligned}\end{equation}
	which is obtained from \eqref{linear BSDE soln} and the fact that
	a triplet $(K^{M,x'},0,K^{M,x'})$ is the unique solution to the linear  $G$-BSDE
	\begin{equation} 
	\begin{aligned}
	\tilde{Y}_s&=K_T^{M,x'}+\int_s^Ta_u^{\epsilon}\tilde{Y}_u-a_u^{\epsilon}K_u^{M,x'}\,du+\int_s^Tc_u^{\epsilon,ij}\tilde{Y}_u-c_u^{\epsilon,ij}K_u^{M,x'}\,d\langle B^{\epsilon,i},B^{\epsilon,j}\rangle_u\\&\quad -\int_s^T\tilde{Z}_u\,dB^{\epsilon}_u-(\tilde{K}_T-\tilde{K}_s)\,.
	\end{aligned}
	\end{equation}
	Thus, \eqref{eqn:gamma_K} becomes
	\begin{equation}
	\begin{aligned}
	\hat{Y}_s\leq\tilde{\mathbb{E}}_s^{b^\epsilon,d^{\epsilon,ij}}\Big[\frac{\Gamma_T}{\Gamma_s}\hat{\Phi}+\int_s^T\frac{\Gamma_u}{\Gamma_s}(m_u^{\epsilon}+h_u)\,du+\int_s^T \frac{\Gamma_u}{\Gamma_s}(n_u^{\epsilon,ij}+k_u^{ij})\,d\langle B^{\epsilon,i},B^{\epsilon,j}\rangle_u\Big]\,.    
	\end{aligned}
	\end{equation} Observe that
  \begin{align}
        \int_s^T \frac{\Gamma_u}{\Gamma_s}(n_u^{\epsilon,ij}+k_u^{ij})\,d\langle B^{\epsilon,i},B^{\epsilon,j}\rangle_u&\le\int_s^T2\frac{\Gamma_u}{\Gamma_s}G( n_u^{\epsilon,ij}+k_u^{ij})\,du\\
        &\le\int_s^T2 \frac{\Gamma_u}{\Gamma_s} G(D_u)\,du\\
        &\le \int_s^T \overline{\sigma}^2\frac{\Gamma_u}{\Gamma_s}  \Big(|n_u^{\epsilon,ij}|+|k_u^{ij}|\Big)\,du\,,
    \end{align}
    where $D_u$ is the diagonal matrix with entries $D_u^{ii}=\sum_{j=1}^d|n_u^{\epsilon,ij}|+|k_u^{ij}|$ for $i=1,2,\cdots d$.  
    From
	$|b^\epsilon|, |d^{\epsilon,ij}|\leq C_3(1+2M)$,
	Theorem \ref{SDE sublinear property}, the Lipschitz property of $\hat{\Phi},h,k^{ij}$ and the inequality   
	\begin{equation}\begin{aligned}
	\frac{\Gamma_u}{\Gamma_s}&=e^{\int_s^ua_v^{\epsilon}\,dv+\int_s^u c_v^{\epsilon,ij}\,d\langle B^{\epsilon,i},B^{\epsilon,j}\rangle_v}=e^{\int_s^ua_v^{\epsilon}+2G(c_v^{\epsilon,ij})\,dv+\int_s^u c_v^{\epsilon,ij}\,d\langle B^{\epsilon,i},B^{\epsilon,j}\rangle_v -\int_s^u2G(c_v^{\epsilon,ij})\,dv}\\
    &\leq e^{-\mu(u-s)}\textnormal{ for }s\leq u\,,
	\end{aligned}\end{equation}

	we have
	\begin{equation}\begin{aligned}
	\hat{Y_s}&\leq  \Big(C_{\Phi}+\frac{C_1(1+\overline{\sigma}^2)}{\eta+\mu-(1+\overline{\sigma}^2)C_{\sigma}C_3(1+2M)}\Big)|X_s^x-X_s^{x'}|\\&
    \quad+\tilde{\mathbb{E}}_s^{b^\epsilon,d^{\epsilon,ij}}\Big[\int_s^T\frac{\Gamma_u}{\Gamma_s}m_u^{\epsilon}\,du+\int_s^T \frac{\Gamma_u}{\Gamma_s}n_u^{\epsilon,ij}\,d\langle B^{\epsilon,i}, B^{\epsilon,j}\rangle_u\Big]\,.
	\end{aligned}\end{equation}
	From   $|m^\epsilon|,|n^{\epsilon,ij}|\leq (C_2(1+d)+2C_3(1+2M))\epsilon$, it follows that  
	\begin{equation}\begin{aligned}
	&\tilde{\mathbb{E}}_s^{b^\epsilon,d^{\epsilon,ij}}\Big[\int_s^T\frac{\Gamma_u}{\Gamma_s}m_u^{\epsilon}\,du+\int_s^T \frac{\Gamma_u}{\Gamma_s}n_u^{\epsilon,ij}\,d\langle B^{\epsilon,i},B^{\epsilon,j}\rangle_u\Big]\\
    &\leq  (1+\overline{\sigma}^2)(C_2(1+d)+2C_3(1+2M))\epsilon\int_s^Te^{-\mu (u-s)}\,du\,.
	\end{aligned}
    \end{equation}
	Letting $\epsilon \rightarrow 0$, we have $Y_s^{M,x}-Y_s^{M,x'}\leq L_1|X_s^x-X_s^{x'}|$, 
	where $$L_1:=C_{\Phi}+\frac{C_1(1+\overline{\sigma}^2)}{\eta+\mu-(1+\overline{\sigma}^2)C_{\sigma}C_3(1+2M)}\,.$$
	Using the same method, $Y_s^{M,x'}-Y_s^{M,x}\leq L_1|X_s^x-X_s^{x'}|,$
	and therefore 
	\begin{equation}\begin{aligned}\label{eqn: lip BSDE }
	|u^M(t,x)-u^M(t,x')|=|Y_t^{M,t,x'}-Y_t^{M,t,x}|=|Y_t^{M,x'}-Y_t^{M,x}|\leq L_1|x-x'| 
	\end{aligned}\end{equation} 
	for all $x,x'\in \mathbb{R}^m.$
	
    Second we show that the function $u^M=u^M(t,x)$ is $\frac{1}{2}$-H\"{o}lder continuous with respect to $t$.   We define the   auxiliary processes  
	\begin{equation} 
	\begin{aligned}
	\tilde{a}_s^{\epsilon}&=l^\epsilon({Y_s^{M,x}})\frac{f^M(s,X^x_s,Y^{M,x}_s,Z^{M,x}_s)-f^M(s,X^x_s,0,Z^{M,x}_s)}{{Y}^{M,x}_s} -\frac{\mu}{1+\underline{\sigma}^2}(1-l^\epsilon({Y_s}^{M,x})) \,,\\
	\tilde{b}_s^{\epsilon}&=l^\epsilon({Z}^{M,x}_s)\frac{f^M(s,X^x_s,0,Z^{M,x}_s)-f^M(s,X^x_s,0,0)}{|{Z}^{M,x}_s|^2}{Z}^{M,x}_s\\
    &\quad+C_3(1+2M)(1-l^\epsilon({Z}_s^{M,x}))\,,\\
	\tilde{m}_s^{\epsilon}&=f^M(s,X^x_s,Y^{M,x}_s,Z^{M,x}_s)-f^M(s,X^x_s,0,0)-\tilde{a}^{\epsilon}_s{Y}^{M,x}_s-\tilde{b}_s^{\epsilon}{Z_s^{M,x}}\,,\\
	\tilde{c}_s^{\epsilon,ij}&=l^\epsilon({Y_s^{M,x}})\frac{g_{ij}^M(s,X^x_s,Y^{M,x}_s,Z^{M,x}_s)-g_{ij}^M(s,X^x_s,0,Z^{M,x}_s)}{{Y}^{M,x}_s} \\&\quad-\frac{\mu}{(1+\underline{\sigma}^2)d}(1-l^\epsilon({Y_s}^{M,x})) \,,\\
	\tilde{d}_s^{\epsilon,ij}&=l^\epsilon({Z}^{M,x}_s)\frac{g_{ij}^M(s,X^x_s,0,Z^{M,x}_s)-g_{ij}^M(s,X^x_s,0,0)}{|{Z}^{M,x}_s|^2}{Z}^{M,x}_s\\
	&\quad+\frac{C_3(1+2M)}{d^2}(1-l^\epsilon({Z}_s^{M,x}))\,,\\
	\tilde{n}_s^{\epsilon,ij}&=g_{ij}^M(s,X^x_s,Y^{M,x}_s,Z^{M,x}_s)-g_{ij}^M(s,X^x_s,0,0)-\tilde{c}^{\epsilon,ij}_s{Y}^{M,x}_s-\tilde{d}_s^{\epsilon,ij}{Z_s^{M,x}}\,.\\
	\end{aligned}	
	\end{equation}
	It can be easily proven that $ \tilde{a}^\epsilon,  \tilde{c}^{\epsilon,ij}, \tilde{m}^{\epsilon},\tilde{n}^{\epsilon,ij}\in\mathbb{M}^2_G(0,T)$, $\tilde{b}^{\epsilon},\tilde{d}^{\epsilon,ij} \in\mathbb{M}^2_G(0,T;\mathbb{R}^d)$
	and  
	$ |\tilde{m}^\epsilon|\leq (2C_2+2C_3(1+2M))\epsilon$,
	$|\tilde{n}^{\epsilon,ij}|\leq (C_2(1+d)+2C_3(1+2M))\epsilon$, $|\tilde{b}^\epsilon|\leq C_3(1+2M)$, $|\tilde{d}^{\epsilon,ij}|\leq C_3(1+2M)$, $ \tilde{a}^\epsilon+2G(\tilde{c}^{\epsilon,ij})\leq -\mu.$ Let $\tilde{\mathbb{E}}^{\tilde{b}^\epsilon,\tilde{d}^{\epsilon,ij}}$ be the sublinear expectation induced by $\tilde{b}^{\epsilon}$ and $\tilde{d}^{\epsilon,ij}$. By \eqref{new_GBM}, $\tilde{B}^{\epsilon}:=B-\int_0^\cdot \tilde{b}_u^{\epsilon}\,du-\int_0^\cdot \tilde{d}^{\epsilon,ij}_u\,d\langle B^i,B^j\rangle_u$ is a $G$-Brownian motion under $\tilde{\mathbb{E}}^{\tilde{b}^\epsilon,\tilde{d}^{\epsilon,ij}}$. 
	Since $Y^{M,x}_{s+\delta}=u^M(s+\delta,X^x_{s+\delta})$ for any $0\leq \delta\leq T-s$,
	\eqref{eqn: truncation M} can be written as
	\begin{equation}
    \begin{aligned}
	u^M(s,X^x_s)&=u^M(s+\delta,X^x_{s+\delta})+\int_s^{s+\delta}  \tilde{a}_u^{\epsilon}Y_u^{M,x}+\tilde{m}^{\epsilon}_u+f^M(u,X^x_u,0,0)\,du\\
    &\quad+\int_s^{s+\delta}  \tilde{c}_u^{\epsilon,ij}Y_u^{M,x}+\tilde{n}^{\epsilon,ij}_u+g_{ij}^M(u,X^x_u,0,0)\,d\langle \tilde{B}^{\epsilon,i},\tilde{B}^{\epsilon,j}\rangle_u\\
    &\quad-\int_s^{s+\delta} Z_u^{M,x}\,d\tilde{B}^\epsilon_u-(K_{s+\delta}^{M,x}-K_s^{M,x})\,.
	\end{aligned}
    \end{equation}
	By \eqref{linear BSDE soln}, it follows that
	\begin{equation}\begin{aligned}
	u^M(s,X^x_s)=\tilde{\mathbb{E}}_s^{\tilde{b}^{\epsilon},\tilde{d}^{\epsilon,ij}}\Big[&\frac{\tilde{\Gamma}_{s+\delta}}{\tilde{\Gamma}_s}u^M({s+\delta},X^x_{s+\delta})+\int_s^{s+\delta}   \frac{\tilde{\Gamma}_u}{\tilde{\Gamma}_s}(\tilde{m}_u^\epsilon+f^M(u,X^x_u,0,0))\,du\\&+\int_s^{s+\delta}\frac{\tilde{\Gamma}_u}{\tilde{\Gamma}_s}(\tilde{n}_u^{\epsilon,ij}+g_{ij}^M(u,X^x_u,0,0))\,d\langle \tilde{B}^{\epsilon,i},\tilde{B}^{\epsilon,j} \rangle_u\Big]\,,
	\end{aligned}\end{equation}
	where $\tilde{\Gamma}$ is a solution to the $G$-SDE $\tilde{\Gamma}_\cdot=1+\int_0^\cdot\tilde{a}_u^{\epsilon}\tilde{\Gamma}_u\,du+\int_0^\cdot\tilde{c}_u^{\epsilon,ij}\tilde{\Gamma}_u\,d\langle \tilde{B}^{\epsilon,i},\tilde{B}^{\epsilon,j}\rangle_u$. Then
	\begin{equation}\begin{aligned}
	&u^M(s,X^x_s)-u^M(s+\delta,X^x_s)\\&=\tilde{\mathbb{E}}_s^{\tilde{b}^{\epsilon},\tilde{d}^{\epsilon,ij}}\Big[\Big(\frac{\tilde{\Gamma}_{s+\delta}}{\tilde{\Gamma}_s}-1\Big)u^M({s+\delta},X^x_{s+\delta})+\Big(u^M({s+\delta},X^x_{s+\delta})-u^M(s+\delta,X^x_s)\Big)\\
	&\quad\quad\quad\quad\quad+\int_s^{s+\delta}   \frac{\tilde{\Gamma}_u}{\tilde{\Gamma}_s}(\tilde{m}_u^\epsilon+f^M(u,X^x_u,0,0))\,du\\
    &\quad\quad\quad\quad\quad+\int_s^{s+\delta}\frac{\tilde{\Gamma}_u}{\tilde{\Gamma}_s}(\tilde{n}_u^{\epsilon,ij}+g_{ij}^M(u,X^x_u,0,0))\,d\langle \tilde{B}^{\epsilon,i},\tilde{B}^{\epsilon,j} \rangle_u\Big]\,.
	\end{aligned}\end{equation}
	We  obtain
	\begin{equation}\begin{aligned}\label{eqn: 1/2 hodler conti }
	|u^M(s,X^x_s)-u^M(s+\delta,X^x_s)|\leq L_1M_\sigma\overline{\sigma}\delta^{\frac{1}{2}}+L_2(1+|X^x_s|)\delta
	\end{aligned}\end{equation}
	with a constant $L_2$ depending only on  $\mu, C_1,C_3,M,M_{\sigma},\overline{\sigma},|\!|b(s,0)|\!|,|\!|h_{ij}(s,0)|\!|,T$ 
	by estimating the terms 
	\begin{equation}\begin{aligned}
	&\tilde{\mathbb{E}}_s^{\tilde{b}^{\epsilon},\tilde{d}^{\epsilon,ij}}\Big[|u^M({s+\delta},X^x_{s+\delta})-u^M(s+\delta,X^x_s)|\Big]
	\leq L_1M_\sigma\overline{\sigma}\delta^{\frac{1}{2}}+L_1L(1+|X^x_s|)\delta \,,\\
	& \tilde{\mathbb{E}}_s^{\tilde{b}^{\epsilon},\tilde{d}^{\epsilon,ij}}\Big[\Big(\frac{\tilde{\Gamma}_{s+\delta}}{\tilde{\Gamma}_s}-1\Big)|u^M({s+\delta},X^x_{s+\delta})|\Big]\leq L(1+|X^x_s|)\delta\,,\\
	&\tilde{\mathbb{E}}_s^{\tilde{b}^{\epsilon},\tilde{d}^{\epsilon,ij}}\Big[\int_s^{s+\delta}  \Big| \frac{\tilde{\Gamma}_u}{\tilde{\Gamma}_s}(\tilde{m}_u^\epsilon+f^M(u,X^x_u,0,0))\Big|\,du\Big]\leq L(1+|X^x_s|)\delta \,,\\
	&\tilde{\mathbb{E}}_s^{\tilde{b}^{\epsilon},\tilde{d}^{\epsilon,ij}}\Big[\Big|\int_s^{s+\delta}\frac{\tilde{\Gamma}_u}{\tilde{\Gamma}_s}(\tilde{n}_u^{\epsilon,ij}+g_{ij}^M(u,X^x_u,0,0))\,d\langle \tilde{B}^{\epsilon,i},\tilde{B}^{\epsilon,j} \rangle_u\Big|\Big]\leq L(1+|X^x_s|)\delta\,,
	\end{aligned}\end{equation}
	which are derived from Theorem \ref{SDE sublinear property}, \eqref{eqn: lip BSDE } and \cite[Proposition 4.2]{hu2014comparison}.  
	
We now prove that 
	the triplet $(Y^M,Z^M,K^M):=(Y^{M,x},Z^{M,x},K^{M,x})$ is a solution to \eqref{QBSDE}. 
	From \eqref{eqn: lip BSDE } and \eqref{eqn: 1/2 hodler conti },  it follows that
	\begin{equation}\begin{aligned}
	\underline{\sigma}^2\int_s^t|Z_u^M|^2\,du&\leq \int_s^t Z_u^{M,i} Z_u^{M,j}\,d\langle B^i,B^j\rangle_u =\lim_{|\!|\pi|\!|\rightarrow 0}\sum_{i=0}^{n-1}(u^M(t_{i+1},X_{t_{i+1}})-u^M(t_{i},X_{t_i}))^2\\
	&\leq 2\lim_{|\!|\pi|\!|\rightarrow 0}\sum_{i=0}^{n-1}(u^M(t_{i+1},X_{t_{i+1}})-u^M(t_{i+1},X_{t_i}))^2\\
	&\quad+2\lim_{|\!|\pi|\!|\rightarrow 0}\sum_{i=0}^{n-1}(u^M(t_{i+1},X_{t_{i}})-u^M(t_{i},X_{t_i}))^2  \\
	&\leq 4L_1^2M_\sigma^2\overline{\sigma}^2(t-s)\,,
	\end{aligned}\end{equation}
    where $Z^{M,i}$ is the $i$-th component of $Z^M$ and $\pi$ is a partition of $[s,t]$ with mesh size $|\!|\pi|\!|$.
	Thus, $Z^M$ is bounded by $M,$ that is, 
	\begin{equation}\begin{aligned}
	|Z^M|\leq 2L_1\frac{\overline{\sigma}M_\sigma}{\underline{\sigma}}\leq M\,.
	\end{aligned}\end{equation}
	Since 
	$f(s,X_s,Y_s^M,Z_s^M)=f^M(s,X_s,Y_s^M,Z_s^M)$ and 
	$g_{ij}(s,X_s,Y_s^M,Z_s^M)=g_{ij}^M(s,X_s,Y_s^M,Z_s^M),$
	combining with
	\eqref{eqn: truncation M}, we conclude that  
	$(Y^M,Z^M,K^M)$ is a solution to \eqref{QBSDE}, moreover $u(t,x):=Y_t^{M,t,x}=u^M(t,x)$ satisfies \eqref{eqn:Lip} by \eqref{eqn: lip BSDE } and \eqref{eqn: 1/2 hodler conti }.
	The case with $\mu>0$ can be proven similarly.

	The uniqueness of solutions can be proven as follows. Suppose $(Y,Z,K)$ and $(Y',Z',K')$ are two solutions to \eqref{QBSDE} such that $Z$ and $Z'$ are bounded. 
	Choosing $M>|Z|+|Z'|$ and defining
    $f^M(s,x,y,z):=f(s,x,y,z^M)$, $g_{ij}^M(s,x,y,z):=g_{ij}(s,x,y,z^M)$, where $z^{M}:=\frac{|z|\wedge M}{|z|}z$ with the convention $\frac{\,0\,}{0}=0$, 
	both
	$(Y,Z,K)$ and $(Y',Z',K')$
	are also solutions to 
	the Lipschitz  $G$-BSDE
	\begin{equation}\begin{aligned}
	Y_s=\Phi(X_T)&+\int_s^Tf^M(u,X_u,Y_u,Z_u) \,du+\int_s^Tg_{ij}^M(u,X_u,Y_u,Z_u) \,d\langle B^i,B^j\rangle_u \\&-\int_s^TZ_u\,dB_u-(K_T-K_s)\,,\;0\le s\le T\,.
	\end{aligned}\end{equation}
	By Theorem \ref{linear FeyKac formula}, the Lipschitz  $G$-BSDE has a unique solution, leading to the conclusion that $(Y,Z,K)=(Y',Z',K').$
\end{proof}

We now present the corresponding Feynman-Kac formula.
A PDE representation of
$u(t,x)=Y_t^{t,x}$ is an important issue.
\cite{hu2014comparison}
have studied this
for Lipschitz $f$ and $g_{ij}$. The following theorem  describes
the Feynman-Kac formula for quadratic $f$ and $g_{ij}$.

\begin{thm}\label{QBSDE feykac formula}
	Let Assumption \ref{assumption} hold and $(Y^{t,x},Z^{t,x},K^{t,x})$ be a solution to \eqref{QBSDE} such that $Z^{t,x}$ is bounded. Define $u(t,x)=Y_t^{t,x}$ for $(t,x)\in [0,T]
	\times\mathbb{R}^m.$ Then $u$ is a viscosity solution to the PDE
	\begin{equation}\begin{aligned}
	\begin{cases}
	\partial_t u+F(t,x,u,D_xu,D_x^2u)=0\,, \\
	u(T,x)=\Phi(x)\,,
	\end{cases}
	\end{aligned}\end{equation}
	where 
	\begin{equation}\begin{aligned}
	F(t,x,u,D_xu,D_x^2u)&=G(H(t,x,u,D_xu,D_x^2u))+\langle b(t,x),D_xu\rangle\\
	&\quad+f(t,x,u,\langle\sigma^1(t,x),D_xu\rangle,\cdots,\langle\sigma^d(t,x),D_xu\rangle)\,,\\
	H_{ij}(t,x,u,D_xu,D_x^2u)&=\langle D^2_xu\sigma^i(t,x),\sigma^j(t,x)\rangle+2\langle D_xu,h_{ij}(t,x)\rangle\\
	&\quad+2g_{ij}(t,x,u,\langle\sigma^1(t,x),D_xu\rangle,\cdots,\langle\sigma^d(t,x),D_xu\rangle) \,,
	\end{aligned}\end{equation}
	and $\sigma^i$  is the $i$-th column vector of $\sigma$ for $i=1,2,\cdots,d.$
\end{thm}

\begin{proof}
	It is evident that the function  
	$u$ is continuous by Theorem \ref{Finite QBSDE exis and unique}.
	We only prove $u$ is a viscosity subsolution.     
    For fixed $(t,x)\in(0,T)\times \mathbb{R}^m$, let $\psi$ be a smooth function on $(0,T)\times \mathbb{R}^m$ such that $\psi\geq u$, $\psi(t,x)=u(t,x)$ and $|\partial_{t,x_i}^2 \psi|+|\partial_{x_i} \psi|+|\partial_{x_i,x_j}^2 \psi|+|\partial_{x_i,x_j,x_k}^3 \psi|$ is bounded. Choose $M>|Z^{t,x}|
	+|(\langle \sigma^1(\cdot,{\cdot}),D_x\psi(\cdot,{\cdot})\rangle, \cdots ,\langle \sigma^d(\cdot,{\cdot}),D_x\psi(\cdot,{\cdot})\rangle)|$ and define 
	\begin{equation}\begin{aligned}
	f^M(s,x,y,z):=f(s,x,y,z^M),\quad g_{ij}^M(s,x,y,z):=g_{ij}(s,x,y,z^M)\,,
	\end{aligned}\end{equation}
    where $z^M:= \frac{|z|\wedge M}{|z|}z$ with the convention $\frac{\,0\,}{0}=0$.
	Consider the $G$-BSDE  
	\begin{equation}
    \begin{aligned}	\tilde{Y}_s&=\psi(t+\delta,X_{t+\delta}^{t,x})+\int_s^{t+\delta}f^M(u,X_u^{t,x},\tilde{Y}_u,\tilde{Z}_u)\,du\\
    &\quad+\int_s^{t+\delta} g^M_{ij}(u,X_u^{t,x},\tilde{Y}_u,\tilde{Z}_u)\,d\langle B^i,B^j\rangle_u
	-\int_s^{t+\delta}\tilde{Z}_u\,dB_u-(\tilde{K}_{t+\delta}-\tilde{K}_s)
	\end{aligned}
    \end{equation}
	for $t\le s\le t+\delta$. Since this is a Lipschitz $G$-BSDE, by Theorem \ref{linear FeyKac formula}, there exists a unique solution $(\tilde{Y},\tilde{Z},\tilde{K})$.
	Define $\hat{Y}_s:=\tilde{Y}_s-\psi(s,X_s^{t,x}),$
	$\hat{Z}_s:=\tilde{Z}_s-(\langle \sigma^1(s,X_s^{t,x}),D_x\psi(s,X_s^{t,x})\rangle, \cdots, \langle \sigma^d(s,X_s^{t,x}),D_x\psi(s,X_s^{t,x})\rangle),$ $\hat{K}_s:=\tilde{K}_s$ for $t\le s \le t+\delta$. Applying It\^{o}'s formula, we have
	\begin{align}
	\hat{Y}_s&=\int_s^{t+\delta} F_1(u,X^{t,x}_u,\hat{Y}_u,\hat{Z}_u)\,du+\int_s^{t+\delta} F_2^{ij}(u,X^{t,x}_u,\hat{Y}_u,\hat{Z}_u)\,d\langle B^i,B^j\rangle_u-\int_s^{t+\delta} \hat{Z}_u\,dB_u\\
    &\quad-(\hat{K}_{t+\delta}-\hat{K}_s)\,,
	\end{align}
where   \begin{align}
	&\quad\; F_1(s,x,y,z)\\
	&:=f^M(s,x,y+\psi(s,x),z+(\langle \sigma^1(s,x),D_x\psi(s,x)\rangle, \cdots ,\langle \sigma^d(s,x),D_x\psi(s,x)\rangle)) \\&\quad 
	+\partial_t\psi(s,x)+\langle b(s,x),D_x\psi(s,x)\rangle	
	\end{align}
and
	\begin{align} 
	&\quad\; F_2^{ij}(s,x,y,z)\\
	&:=g^M_{ij}(s,x,y+\psi(s,x),z+(\langle \sigma^1(s,x),D_x\psi(s,x)\rangle, \cdots ,\langle \sigma^d(s,x),D_x\psi(s,x)\rangle))\\
	&\quad +\langle D_x\psi(s,x),h_{ij}(s,x)\rangle+\frac{1}{2}\langle D_x^2\psi(s,x)\sigma^i(s,x),\sigma^j(s,x)\rangle\,.
	\end{align} 
	We compare $(\hat{Y},\hat{Z},\hat{K})$ with a solution $(\overline{Y},\overline{Z},\overline{K})$ to the $G$-BSDE
	
	\begin{align}
	\overline{Y}_s&=\int_s^{t+\delta}F_1(u,x,\overline{Y}_u,\overline{Z}_u)\,du+\int_s^{t+\delta}F_2^{ij}(u,x,\overline{Y}_u,\overline{Z}_u)\,d\langle B^i,B^j\rangle_u-\int_s^{t+\delta}\overline{Z}_u\,dB_u\\
    &\quad-(\overline{K}_{t+\delta}-\overline{K}_s)
	\end{align}
	for $t\le s \le t+\delta$.
	It can be easily verified that 
	\begin{align}
	\overline{Y}_s&=\int_s^{t+\delta}F_1(u,x,\overline{Y}_u,0)+2G(F_2(u,x,\overline{Y}_u,0))\,du\,,\\
	\overline{K}_s&=\int_t^s F_2^{ij}(u,x,\overline{Y}_u,0)\,d\langle B^i,B^j\rangle_u-\int_t^s2G(F_2(u,x,\overline{Y}_u,0))\,du\,,\\
	\overline{Z}_s&=0
	\end{align}
	with $F_2(u,x,y,0):=(F_2^{ij}(u,x,y,0))_{i,j=1}^d$ is a solution to the above $G$-BSDE. In the following, $L>0$ is a generic constant depending only on $C_1,C_2,C_3,C_\sigma, \psi, M_\sigma$  and may differ line by line. Since the functions $F_1,F_2^{ij}$ are Lipschitz with respect to $z$ by \cite[Proposition 3.8]{hu2014backward}, 
    \begin{align}
	&\quad|\hat{Y}_t-\overline{Y}_t|^2
	\le \hat{\mathbb{E}}[\sup_{t\le s\le t+\delta}|\hat{Y}_s-\overline{Y}_s|^2]\\
	&\le L\Big( \Big(\hat{\mathbb{E}}\Big[\sup_{t\le s\le t+\delta} \hat{\mathbb{E}}_s\Big[\Big(\int_t^{t+\delta} \hat{F}_u\,du\Big)^4\Big]\Big] \Big)^{\frac{1}{2}} + \hat{\mathbb{E}}\Big[\sup_{t\le s\le t+\delta} \hat{\mathbb{E}}_s\Big[\Big(\int_t^{t+\delta} \hat{F}_u\,du\Big)^4\Big]\Big]   \Big)\,,
	\end{align}
	where $\hat{F}_u=|F_1(u,X^{t,x}_u,\overline{Y}_u,0)-F_1(u,x,\overline{Y}_u,0)|+|F_2^{ij}(u,X^{t,x}_u,\overline{Y}_u,0)-F_2^{ij}(u,x,\overline{Y}_u,0)|$.
	Since $\hat{F}_u\le L|X^{t,x}_u-x|$ for $t\le u\le t+\delta$, we obtain $|\hat{Y}_t-\overline{Y}_t|\le L(1+|x|^{12})\delta^{\frac{3}{2}}$ for $\delta\leq 1$ by estimating the term
	\begin{equation}\begin{aligned}
	\hat{\mathbb{E}}\left[\sup_{t\leq s\leq t+\delta}\hat{\mathbb{E}}_s\left[\left(\int_t^{t+\delta} \hat{F}_u\,du\right)^4\right]\right]     &\leq\delta \hat{\mathbb{E}} \left[\sup_{t\le s\le t+\delta} \hat{\mathbb{E}}_s\left[ 
	\sup_{t\leq u\leq t+\delta}L|X_u^{t,x}-x|^4\right]\right]\\
    &\leq L(1+|x|^{12})\delta^3\,, 
	\end{aligned}\end{equation}
	which is derived from Theorems \ref{similar doob maximal} and \ref{SDE sublinear property}.
	Since the terminal conditions satisfy $\psi(t+\delta,X_{t+\delta}^{t,x})\ge u(t+\delta,X_{t+\delta}^{t,x}),$ by the comparison principle, 
	we have 
	$\tilde{Y}_t\geq u(t,x)$, and thus $\hat{Y}_t\geq 0$. It follows that
	\begin{equation}\begin{aligned}
	-L(1+|x|^{12})\delta^{\frac{1}{2}}\leq \frac{1}{\delta}(\overline{Y}_t-\hat{Y}_t)\leq\frac{1}{\delta} \overline{Y}_t= \frac{1}{\delta}\int_t^{t+\delta}F_1(u,x,\overline{Y}_u,0)+2G(F_2(u,x,\overline{Y_u},0))\,du\,.
	\end{aligned}\end{equation}
	Sending $\delta\rightarrow 0$, we obtain $u=u(t,x)$ is a viscosity subsolution.
\end{proof}

\begin{remark}\label{remark:finite_BSDEs}
	Theorems \ref{Finite QBSDE exis and unique} and \ref{QBSDE feykac formula} differ from existing results in the literature in several aspects. To apply the results of \cite{hu2018quadratic} to our setting, a boundedness condition is required: the terminal condition $\omega \mapsto \xi(\omega)$ and the drivers
	$$
	(s, \omega) \mapsto f(s, X_s(\omega), 0, 0), \quad (s, \omega) \mapsto g_{ij}(s, X_s(\omega), 0, 0)
	$$
	must be bounded. Additionally, a pathwise continuity condition must hold: there exists a modulus of continuity $w$ such that
	$$
	|f(s, X_s(\omega), y, z) - f(s', X_{s'}(\omega'), y, z)| \le w\Big( |s - s'| + \sup_{u \in [0, T]} |\omega_{u\wedge s} - \omega_{u'\wedge s'}| \Big)
	$$
	for all $s, s' \in [0,T]$, $y \in \mathbb{R}$, and $z \in \mathbb{R}^d$. This requirement is  quite  restrictive and excludes even simple cases such as $f(s, x, y, z) = x$ since the solution map $\omega \mapsto X(\omega)$ of an SDE is typically not continuous under the supremum norm topology.
	In contrast, our result does not require these boundedness or pathwise continuity assumptions.
	Similarly, the framework in \cite{hu2022quadratic} imposes the same pathwise continuity condition and further assumes that the driver $f(s, x, y, \cdot)$ is either convex or concave. Our analysis, however, remains valid without such structural constraints.
	
\end{remark}

\section{Infinite-horizon quadratic $G$-BSDEs}\label{section 4}

This section studies  infinite-horizon quadratic $G$-BSDEs.
We present the existence, uniqueness of solutions and the  corresponding Feyman-Kac formula.
Given $x\in \mathbb{R}^m$ and continuous functions 
$b,h_{ij}: \mathbb{R}^m\to \mathbb{R}^m$, $\sigma:\mathbb{R}^m\to \mathbb{R}^{m\times d}$,  $f,g_{ij}: \mathbb{R}^m\times \mathbb{R}\times \mathbb{R}^d\to
\mathbb{R}$,
consider the $G$-SDE and $G$-BSDE  
\begin{align} \label{ergodic SDE2}
X_s^x&=x+\int_0^s b(X_u^x)\,du+\int_0^s h_{ij}(X_u^x)\,d\langle B^i,B^j\rangle_u+\int_0^s\sigma(X_u^x) \,dB_u\,, 
\\
\label{ergodic QBSDE2} Y_s^x&=Y_T^x+\int_s^Tf(X_u^x,Y_u^x,Z_u^x)\,du+\int_s^Tg_{ij}(X_u^x,Y_u^x,Z_u^x)\,d\langle B^i,B^j\rangle_u\\&
\quad-\int_s^TZ_u^x\,dB_u-(K_T^x-K_s^x) 
\end{align}
for $0\leq s \leq T<\infty$. Occasionally, we write $X^x, Y^x, Z^x, K^x$ as $X,Y,Z,K$ respectively, omitting the superscripts $x$.
\begin{assume}\label{ergodic assumption1} Assume $b,h_{ij},\sigma,g_{ij}$ satisfy the following properties.
	\begin{enumerate}[label=(\roman*)]%[font={\bfseries},label=({A\arabic*})] 
		\item \label{ergodic assumption A1}For $1\leq i,j \leq d$, $h_{ij}=h_{ji}$ and $g_{ij}=g_{ji}$.
		\item \label{ergodic assumption A2} There exist  constants $C_1, C_{\sigma}, M_{\sigma}$ such that 
		\begin{equation}\begin{aligned}
		&|b(x)-b(x')|+\sum_{i,j=1}^d|h_{ij}(x)-h_{ij}(x')|\leq C_1|x-x'|\,,\\
		&|\sigma(x)-\sigma(x')|\leq C_{\sigma}|x-x'|\,,\\
		&|\sigma(x)|\le M_{\sigma}
		\end{aligned}\end{equation}
		for $x,x'\in \mathbb{R}^m$. 
	\end{enumerate}
\end{assume}

\begin{assume}\label{ergodic assumption2} Assume $b,h_{ij}, \sigma,f,g_{ij}$ satisfy the following properties.
	\begin{enumerate}[label=(\roman*)]%[font={\bfseries},label=({A\arabic*})] 
		\item \label{ergodic assumption A3} There exist  constants $C_1, C_2,C_3$ such that 
		\begin{equation}\begin{aligned}
		&|f(x,y,z)-f(x',y',z')|+\sum_{i,j=1}^d|g_{ij}(x,y,z)-g_{ij}(x',y',z')|\\ &\hspace{1cm}\leq C_1|x-x'|+C_2|y-y'|+C_3(1+|z|+|z'|)|z-z'|
		\end{aligned}\end{equation}
		for $x,x'\in \mathbb{R}^m$, $y,y'\in \mathbb{R}, $ $z,z'\in \mathbb{R}^d$.
		\item There exist constants $\mu>0$ and $\eta\in\mathbb{R}$ such that
		\begin{equation}\begin{aligned}
		&(f(x,y,z)-f(x,y',z))(y-y')+2G((g_{ij}(x,y,z)-g_{ij}(x,y',z))(y-y'))\\
		&\quad \leq -\mu|y-y'|^2\,,\\
		&G\big( \sum_{j=1}^m(\sigma_j(x)-\sigma_j(x'))^{\top} (\sigma_j(x)-\sigma_j(x')) +2(\langle x-x',h_{ij}(x)-h_{ij}(x')\rangle)_{i,j=1}^d\big)\\
		&+\langle x-x',b(x)-b(x')\rangle\leq -\eta|x-x'|^2\,
		\end{aligned}\end{equation}
		for $x,x'\in \mathbb{R}^m$, $y,y'\in\mathbb{R},z\in\mathbb{R}^d$,
		where $\sigma_j$ is the $j$-th row vector of $\sigma$ for $j=1,2,\cdots,m.$
		\item 
		$\mu+\eta>(1+\overline{\sigma}^2)\Bigg(C_\sigma C_3+4\sqrt{C_\sigma C_1C_3\frac{\overline{\sigma}M_\sigma}{\underline{\sigma}}} \Bigg)$.
	\end{enumerate}
\end{assume}

\begin{defi}
	We say a triplet $(Y,Z,K)=(Y_s^{x},Z_s^{x},K_s^{x})_{s\ge0}$ is a solution to the $G$-BSDE \eqref{ergodic QBSDE2} if
	\begin{enumerate}[label=(\roman*)]
		\item  $(Y,Z)\in\mathbb{S}_G^2(0,\infty)\times\mathbb{H}_G^{2}(0,\infty;\mathbb{R}^d)$,
		\item the process $K$ is a decreasing $G$-martingale with $K_0=0$ and $K_T\in \mathbb{L}^2_G(\Omega_T)$ for  all $T>0$,
		\item the triplet $(Y,Z,K)$ satisfies \eqref{ergodic QBSDE2} quasi-surely for all $0\leq s\leq T<\infty$.
	\end{enumerate}
\end{defi}

The following lemma is useful in extending the results in \cite{hu2018ergodic} to our cases. 
We recall 
the sublinear conditional expectation $\tilde{\mathbb{E}}_t^{b,d^{ij}}$  in Theorem \ref{SDE sublinear property}.

\begin{lemma}\label{infinite mean is finite} Let Assumptions \ref{ergodic assumption1}-\ref{ergodic assumption2} hold
	and  $X^{x}$ be a solution to \eqref{ergodic SDE2} for $x\in \mathbb{R}^m$.	
    Suppose  $b,d^{ij}\in \mathbb{M}^2_G(0,\infty;\mathbb{R}^d)$ are processes satisfying $d^{ij}=d^{ji}$ and  $|b|,|d^{ij}|\le C_4$ for some constant $C_4>0.$
For any constant $k>0$ with $k+\eta-(1+\overline{\sigma}^2)C_\sigma C_4>0,$ 
	there exists a constant
	$L>0$ depending only on $\underline{\sigma},$ $\overline{\sigma},$ $C_{\sigma},$ $C_4$, $k$, $\eta$ such that	
	\begin{equation}\begin{aligned}
	\tilde{\mathbb{E}}_t^{b,d^{ij}}[e^{-ks}|X_s^{x }|]\le L e^{-kt}(1+|{X_t}^x|)
	\end{aligned}\end{equation}
	for all $0\le t\le s<\infty$ and $x\in \mathbb{R}^m.$ 
\end{lemma}
\begin{proof}
Denote as $X=X^x$ for convenience.  Recall that $B_s^{{b},{d}^{ij}}=B_s-\int_0^s {b}_u\, du-\int_0^s{d}^{ij}_u\,d\langle B^i,B^j\rangle _u,s\ge0$ is a $G$-Brownian motion under the sublinear expectation $\tilde{\mathbb{E}}^{{b},{d}^{ij}}$.  By Jensen's inequality, we have
	\begin{equation}\begin{aligned}
	\tilde{\mathbb{E}}_t^{b,d^{ij}}[e^{-ks}|X_s|]
	\leq  \Big(\tilde{\mathbb{E}}^{b,d^{ij}}_t\Big[e^{-2ks}|X_s|^2\Big]\Big)^{\frac{1}{2}}\,.
	\end{aligned}\end{equation}
To estimate the term $e^{-2ks}|X_s|^2$, define $\tilde{L}=k+\eta-(1+\overline{\sigma}^2)C_\sigma C_4>0$ and  $\tilde{\psi}(s)=\psi(X_s)-\psi(0)$ for $\psi=b,h_{ij},\sigma_j$. Applying It\^{o}'s formula, we obtain
	\begin{equation}\begin{aligned}
	e^{-2ks}|X_s|^2
	\leq&\, e^{-2kt}|X_t|^2-2\int_t^s ke^{-2ku}|X_u|^2\,du+2\int_t^s e^{-2ku} \langle X_u,\tilde{b}(u)\rangle \,du\\
    &+2\int_t^se^{-2ku}G\Big(\Big(\sum_{j=1}^m \tilde{\sigma}_j^{\top}(u)\tilde{\sigma}_j(u)\Big) +2 \langle X_u,\tilde{h}_{ij}(u)\rangle \Big)\,du+\tilde{\Lambda}_s+\tilde{M}_s\,,
	\end{aligned}\end{equation}
	where
	\begin{equation} 
	\begin{aligned}
	\tilde{M}_s&=2\int_t^se^{-2ku}X_u^{\top}\sigma(X_u)\,d B^{{b},{d}^{ij}}_u\,,\\
	\tilde{\Lambda}_s&=2\int_t^s e^{-2ku}\langle X_u,\sigma(X_u)b_u\rangle + e^{-2ku}\langle X_u,b(0)\rangle\,du\\
    &\quad+4\int_t^se^{-2ku}G(\langle X_u,\sigma(X_u)d_u^{ij}\rangle )+e^{-2ku}G(\langle X_u,h_{ij}(0)\rangle)\,du\\
    &\quad+2\int_t^s e^{-2ku}G\Big(\sum_{j=1}^m\sigma_j^{\top}(X_u)\sigma_j(0)+\sigma_j^{\top}(0)\sigma_j(X_u)-\sigma_j^{\top}(0)\sigma_j(0)\Big) \,du\,.
	\end{aligned}
	\end{equation}  
   Observe that 
	\begin{equation}\begin{aligned}
	\tilde{\Lambda}_s&\leq 2\int_t^s e^{-2ku}|\langle X_u,\sigma(X_u)b_u\rangle|+e^{-2ku}|\langle X_u,b(0)\rangle|\,du\\
    &\quad+2\int_t^s \overline{\sigma}^2e^{-2ku} |\langle X_u,\sigma(X_u)d_u^{ij}\rangle|+\overline{\sigma}^2e^{-2ku}|\langle X_u,h_{ij}(0)\rangle|\,du\\
    &\quad+\int_t^s\overline{\sigma}^2e^{-2ku}\Big|\Big(\sum_{\ell =1}^m\sigma_\ell^{\top}(X_u)\sigma_\ell(0)+\sigma_\ell^{\top}(0)\sigma_\ell(X_u)\Big)_{ij}\Big|\,du\\
    &\le \int_t^s L_1e^{-2ku}(|X_u|+1)\,du\,,
	\end{aligned}\end{equation}
	where $L_1$ is a constant depending only on $\overline{\sigma},M_{\sigma},C_4$.
    From the inequality $ab\leq \frac{ca^2}{2}+\frac{b^2}{2c}$ for $c\geq 0$ with $a=X_u$, $b=L_1$, $c=4C_\sigma C_4$,
    we have
    \begin{equation}\begin{aligned}
	\tilde{\Lambda}_s\le2\int_t^sC_\sigma C_4 e^{-2ku}|X_u|^2\,du+\int_t^s\tilde{L}_1e^{-2ku}\,du\,
	\end{aligned}\end{equation}
 for some constant $\tilde{L}_1$  depending only on $\overline{\sigma},C_{\sigma},M_\sigma,C_4$.
    Thus,
	\begin{equation}\begin{aligned}
	e^{-2ks}|X_s|^2\leq e^{-2kt}|X_t|^2+\int_t^s\tilde{L}_1e^{-2ku}\,du+\tilde{M}_s\,.
	\end{aligned}\end{equation}
	Since $\tilde{M}$ is a symmetric ${G}$-martingale, we obtain
	\begin{equation}\begin{aligned}
	\tilde{\mathbb{E}}^{b,d^{ij}}_t[e^{-2ks}|X_s|^2]\le e^{-2kt}(|X_t|^2+\frac{\tilde{L}_1}{2k})\,,
	\end{aligned}\end{equation}
	which gives the desired result.
\end{proof}

\begin{thm}\label{infinite horizon Quadartic BSDE exist and unique}
	Suppose Assumptions \ref{ergodic assumption1}-\ref{ergodic assumption2} hold. Then for any $x\in \mathbb{R}^m$, there exists a unique solution $(Y^{x},Z^{x},K^{x})$ to the $G$-BSDE \eqref{ergodic QBSDE2}    such that 
	\begin{equation}\label{eqn:inq}
	\begin{aligned}
	&|Y^x|\leq L(1+|X^x|)\,,\\
	&|Z^x|\leq \frac{\mu+\eta-(1+\overline{\sigma}^2)C_\sigma C_3}{4(1+\overline{\sigma}^2)C_\sigma C_3}\,,\\
	&|Y_s^x-Y_s^{x'}|\leq L|X_s^x-X_s^{x'}|
	\end{aligned}
	\end{equation} 
	for a positive constant $L$. 
\end{thm}

\begin{proof} Define $$L_1:=\frac{\mu+\eta-(1+\overline{\sigma}^2)C_\sigma C_3}{4(1+\overline{\sigma}^2)C_\sigma C_3}\,.$$ We first  prove the uniqueness of solutions. Suppose $(Y^1,Z^1,K^1)$ and $(Y^2,Z^2,K^2)$ are two solutions to \eqref{ergodic QBSDE2}. Define $(\hat{Y},\hat{Z},\hat{K})=(Y^1-Y^2,Z^1-Z^2,K^1-K^2)$. Through a similar argument in the proof of Theorem \ref{Finite QBSDE exis and unique}, for each $\epsilon>0,$ we can construct processes  $ a^{\epsilon},c^{\epsilon,ij},m^{\epsilon},n^{\epsilon,ij}\in \mathbb{M}^2_G(0,T)$
    and
    $b^{\epsilon},d^{\epsilon,ij}\in \mathbb{M}^2_G(0,T;\mathbb{R}^d)$ 
    for $T\geq 0$ such that 
    \begin{align}
        \hat{Y}_s=\hat{Y}_T+\int_s^T (a_u^\epsilon \hat{Y}_u+b^\epsilon_u \hat{Z}_u+m_u^\epsilon)\,du&+\int_s^T (c_u^{\epsilon,ij} \hat{Y}_u+d^{\epsilon,ij}_u \hat{Z}_u+n_u^{\epsilon,ij})\,d\langle B^i,B^j\rangle_u\\&-\int_s^T\hat{Z}_u\,dB_u-(\hat{K}_T-\hat{K}_s)
    \end{align}
	and $ |m^\epsilon|\leq (2C_2+2C_3(1+2L_1))\epsilon$,
	$|n^{\epsilon,ij}|\leq (C_2(1+d)+2C_3(1+2L_1))\epsilon$, $|b^\epsilon|\leq C_3(1+2L_1)$, $|d^{\epsilon,ij}|\leq C_3(1+2L_1)$, $ a^\epsilon+2G(c^{\epsilon,ij})\leq -\mu.$ Let $\tilde{\mathbb{E}}^{b^\epsilon,d^{\epsilon,ij}}$ be the sublinear expectation  induced by $b^{\epsilon}$ and $d^{\epsilon,ij}$. By \eqref{new_GBM}, $B^{\epsilon}:=B-\int_0^\cdot b_u^{\epsilon}\,du-\int_0^\cdot d^{\epsilon,ij}_u\,d\langle B^i,B^j\rangle_u$ is a $G$-Brownian motion under $\tilde{\mathbb{E}}^{b^\epsilon,d^{\epsilon,ij}}$. Let $\Gamma^{\epsilon}$ be a solution to the $G$-SDE   $$\Gamma_\cdot^{\epsilon}=1+\int_0^\cdot a_u^{\epsilon}\Gamma_u^{\epsilon}\,du+\int_0^{\cdot} c_u^{\epsilon,ij}\Gamma_u^\epsilon\,d\langle B^{\epsilon,i},B^{\epsilon,j} \rangle_u \,.$$  We then  have
	\begin{equation}\begin{aligned}
	\hat{Y}_s\leq \tilde{\mathbb{E}}_s^{b^\epsilon,d^{\epsilon,ij}}\Big[\frac{\Gamma_T^{\epsilon}}{\Gamma_s^{\epsilon}}\hat{Y}_T+\int_s^T\frac{\Gamma_u^{\epsilon}}{\Gamma_s^{\epsilon}}m_u^\epsilon\,du+\int_s^T\frac{\Gamma_u^{\epsilon}}{\Gamma_s^{\epsilon}}n_u^{\epsilon,ij}\,d\langle B^{\epsilon,i},B^{\epsilon,j} \rangle _u\Big]\,.
	\end{aligned}\end{equation}
	Since $$\mu+\eta>(1+\overline{\sigma}^2)\Bigg(C_\sigma C_3+4\sqrt{C_\sigma C_1C_3\frac{\overline{\sigma}M_\sigma}{\underline{\sigma}}} \,\Bigg)\,,$$
	we can choose a constant $k$ such that $0<k<\mu$ and $k+\eta-(1+\overline{\sigma}^2)C_\sigma C_3(1+2L_1)>0$. By Lemma \ref{infinite mean is finite}, we have 
	   \begin{equation}
	\begin{aligned}
	\hat{Y}_s&\leq e^{-(\mu-k)(T-s)}\tilde{\mathbb{E}}_s^{b^\epsilon,d^{\epsilon,ij}}[e^{-k(T-s)}\hat{Y}_T]\\
	&\quad +
	\tilde{\mathbb{E}}_s^{b^\epsilon,d^{\epsilon,ij}}\Big[\int_s^Te^{-\mu(u-s)}m_u^\epsilon\,du+\int_s^Te^{-\mu(u-s)}n_u^{\epsilon,ij}\,d\langle B^{\epsilon,i},B^{\epsilon,j} \rangle _u\Big]\\
	&\leq 2Le^{-\mu(T-s)}+2\tilde{L}Le^{-(\mu-k)(T-s)}(1+|X_s|)\\
	&\quad +\tilde{\mathbb{E}}_s^{b^\epsilon,d^{\epsilon,ij}}\Big[\int_s^Te^{-\mu(u-s)}m_u^\epsilon\,du+\int_s^Te^{-\mu(u-s)}n_u^{\epsilon,ij}\,d\langle B^{\epsilon,i},B^{\epsilon,j} \rangle _u\Big]
	\end{aligned}   
	\end{equation} 
	for some constant $\tilde{L}>0.$ 
	Letting $\epsilon \rightarrow 0$, we have
	\begin{equation}\begin{aligned}
	Y^1_s-Y^2_s=\hat{Y}_s\leq2Le^{-\mu(T-s)}+2\tilde{L}Le^{-(\mu-k)(T-s)}(1+|X_s|)\,.
	\end{aligned}\end{equation}
	Thus, $Y^1_s-Y^2_s\leq 0$ by letting $T\rightarrow \infty$, and a similar procedure provides $Y^1_s-Y^2_s\geq 0$, which implies $Y^1_s-Y^2_s=0.$ 
	Using the continuity of $Y^1$ and $Y^2$, we obtain that both are indistinguishable under all $P\in \mathcal{P},$ which implies that $Y^1=Y^2$ by Proposition \ref{relation P and sublinear}. From \cite[Proposition 3.8]{hu2014backward}, we have $Z^1=Z^2,$ and thus  $K^1=K^2.$

	We now prove the existence of solutions. For each $\ell \in\mathbb{N}$,
	there exists a unique solution  $(Y^\ell,Z^\ell,K^\ell)=(Y_s^\ell,Z_s^\ell,K_s^\ell)_{0\le s\le \ell}$  to
	the finite horizon $G$-BSDE 
	\begin{equation}\begin{aligned}\label{eqn}
	Y_s^\ell=\int_s^\ell f(X_u,Y_u^\ell,Z_u^\ell)\,du&+\int_s^\ell g_{ij}(X_u,Y_u^\ell,Z_u^\ell)\,d\langle B^i,B^j \rangle _u\\&-\int_s^\ell Z^\ell_u\,dB_u-(K^{\ell}_\ell-K^\ell_s)\,,\;0\le s\le \ell
	\end{aligned}\end{equation}
	such that $|Z^\ell |\leq L_1$
	by Theorem \ref{Finite QBSDE exis and unique}.
	We extend the solution $(Y_s^\ell  ,Z_s^\ell ,K_s^\ell )_{0\le s\le \ell }$ to $(Y_s^\ell ,Z_s^\ell ,K_s^\ell )_{0\le s< \infty}$ by defining 
	\begin{equation}\begin{aligned}
	Y_s^\ell =0, \;Z_s^\ell =0, \; K_s^\ell =K_\ell ^\ell \; \textnormal{  for all }  s>\ell \,.
	\end{aligned}\end{equation}
	Then we can rewrite \eqref{eqn} as
	\begin{equation}\begin{aligned}
	Y^\ell _s&=Y_T^\ell +\int_s^T f(X_u,Y_u^\ell ,Z_u^\ell )-\mathds{1}_{\{u\geq \ell \}}f(X_u,0,0)\,du\\&
    \quad+\int_s^Tg_{ij}(X_u,Y_u^\ell ,Z_u^\ell )-\mathds{1}_{\{u\geq \ell \}}g_{ij}(X_u,0,0)\,d\langle B^i,B^j \rangle _u\\
    &\quad-\int_s^TZ^\ell _u\,dB_u-(K^\ell _T-K^\ell _s)\,
	\end{aligned}\end{equation}
	for all $T\geq 0$.
	For $\ell \leq r$, let $(\hat{Y},\hat{Z},\hat{K}):=(Y^\ell -Y^r,Z^\ell -Z^r,K^\ell -K^r)$. We define processes $a^{\ell ,r,\epsilon},$ $b^{\ell ,r,\epsilon},$
	$m^{\ell ,r,\epsilon},$
	$c^{\ell ,r,\epsilon,ij},$ $d^{\ell ,r,\epsilon,ij},$
	$n^{\ell ,r,\epsilon,ij},$ ${B}^{\ell ,r,\epsilon},$ ${\Gamma}^{\ell ,r,\epsilon}$ and the sublinear expectation $\tilde{\mathbb{E}}^{b^{\ell ,r,\epsilon},d^{\ell ,r,\epsilon,ij}}$ similarly as above.
	Then
	\begin{equation}\begin{aligned}
	\hat{Y_s}&=\int_s^r a^{\ell ,r,\epsilon}_u\hat{Y}_u+b_u^{\ell ,r,\epsilon}\hat{Z}_u+m_u^{\ell ,r,\epsilon}-\mathds{1}_{\{u>\ell \}}f(X_u,0,0)\,du \\
	&\quad+\int_s^rc_u^{\ell ,r,\epsilon,ij}\hat{Y}_u+d_u^{\ell ,r,\epsilon,ij}\hat{Z_u}+n_u^{\ell ,r,\epsilon,ij}-\mathds{1}_{\{u>\ell \}}g_{ij}(X_u,0,0)\,d\langle B^i,B^j\rangle_u\\
    &\quad-\int_s^r\hat{Z}_u\,dB_u-(\hat{K}_r-\hat{K}_s)\,,
	\end{aligned}\end{equation}
	which gives
	\begin{equation}\begin{aligned}
	\hat{Y}_s&\leq \frac{(1+\overline{\sigma}^2)\max(|f(0,0,0)|,|g_{ij}(0,0,0)|)} {\mu}e^{\mu s}(e^{-\mu \ell }-e^{-\mu r})\\&
    \quad+\frac{(1+\overline{\sigma}^2)C_1L}{\mu-k}e^{(\mu-k) s}(e^{-(\mu-k) \ell }-e^{-(\mu-k) r})(1+|X_s|)
	\end{aligned}\end{equation}
	for some constant $k>0$ from a similar argument in the proof of uniqueness above.
	Repeating this procedure to $-\hat{Y}$, we have
	\begin{equation}\begin{aligned}
	|\hat{Y}_s|&\leq \frac{(1+\overline{\sigma}^2)\max(|f(0,0,0)|,|g_{ij}(0,0,0)|)} {\mu}e^{\mu s}(e^{-\mu \ell }-e^{-\mu r})\\&\quad+\frac{(1+\overline{\sigma}^2)C_1L}{\mu-k}e^{(\mu-k) s}(e^{-(\mu-k) \ell }-e^{-(\mu-k) r})(1+|X_s|)\,.\label{eqn: Y=0}
	\end{aligned}\end{equation}
	Therefore, for $0\leq T \leq \ell \leq r$, 
	\begin{equation}\begin{aligned}
	\lim_{\ell ,r\rightarrow \infty}\hat{\mathbb{E}}[\sup_{0\leq s\leq T}|Y_s^\ell -Y_s^r|^2 ]= \lim_{\ell ,r\rightarrow \infty}\hat{\mathbb{E}}[\sup_{0\leq s\leq T}|\hat{Y}_s|^2 ]=0\,.
	\end{aligned}\end{equation}
	Since the sequence $(Y^\ell )_{\ell \in\mathbb{N}}$ is Cauchy in   $\mathbb{S}^2_G(0,T) $ for each $T\geq 0$, we define 
	$    Y:=\lim_{\ell \rightarrow \infty} Y^\ell .$

	We now construct two processes $Z$ and $K$ so that 
	$(Y,Z,K)$ is a solution to \eqref{ergodic QBSDE2}.
	By \cite[Proposition 3.8]{hu2014backward} and \cite[Theorem 3.1]{hu2018ergodic}, for each $T>0,$ 
	the limit 
	$Z^T:=\lim_{\ell \rightarrow \infty}Z^\ell $
 exists	in $\mathbb{H}^2_G(0,T;\mathbb{R}^d)$ and
	the limit
	$K_s^T:=\lim_{\ell \rightarrow \infty}K_s^\ell $
	exists in $\mathbb{L}^2_G(\Omega_s)$ for all $s\leq T$.
	Moreover,  the process $K^T$ is a decreasing $G$-martingale, and  $(Y,Z^T,K^T)$ satisfies 
	\begin{equation}
    \begin{aligned}
	Y_s&=Y_T+\int_s^Tf(X_u,Y_u,Z_u^T)\,du+\int_s^Tg_{ij}(X_u,Y_u,Z_u^T)\,d\langle B^i,B^j\rangle_u\\&
    \quad-\int_s^TZ_u^T\,dB_u-(K_T^T-K_s^T)\,.
	\end{aligned}
    \end{equation}
	By Theorem \ref{linear FeyKac formula}, it is evident that $(Z_s^T,K_s^T)=(Z_s^S,K_s^S)$ for $s\le T<S$. Finally,  we define $(Z,K)=(Z_s,K_s)_{s\ge0}$ as $Z_s=Z_s^T\,,\,K_s=K_s^T$ for $s\leq T,$
	then the triplet $(Y,Z,K)$ is a solution to \eqref{ergodic QBSDE2}.

	We now prove the  inequalities for $Y$ and $Z$ in \eqref{eqn:inq}.
	Since  $|Z^\ell |\leq L_1$ and
	$Z=\lim_{\ell \rightarrow \infty}Z^\ell ,$
	it is evident that
	$|Z|\leq L_1.$
	To prove $|Y|\leq L(1+|X|)$ for some constant $L>0,$ since
	$Y=\lim_{\ell \rightarrow \infty} Y^\ell ,$
	it suffices to
	prove $|Y^\ell |\leq L(1+|X|).$ 
	Through a similar argument in Theorem \ref{Finite QBSDE exis and unique}, there exist processes $a^{\ell ,\epsilon}$, $b^{\ell ,\epsilon}$, $c^{\ell ,\epsilon,ij}$, $d^{\ell ,\epsilon,ij}$, $m^{\ell ,\epsilon}$, $n^{\ell ,\epsilon,ij}, \Gamma^{\ell ,\epsilon}, B^{\ell ,\epsilon}$ and sublinear expectation $\tilde{\mathbb{E}}^{b^{\ell ,\epsilon},d^{\ell ,\epsilon,ij}}$ such that 
	\begin{equation}\begin{aligned}
	Y^\ell _s&=\int_s^\ell  (a_u^{\ell ,\epsilon}Y_u^\ell +b_u^{\ell ,\epsilon}Z_u^\ell +m_u^{\ell ,\epsilon}+f(X_u,0,0))\,du\\
	&\quad+\int_s^\ell (c_u^{\ell ,\epsilon,ij}Y_u^\ell +d_u^{\ell ,\epsilon,ij}Z_u^\ell +n_u^{\ell ,\epsilon,ij}+g_{ij}(X_u,0,0))\,d\langle B^i,B^j\rangle_u\\
    &\quad
	-\int_s^\ell Z_u^\ell \,dB_u-(K_\ell ^\ell -K_s^\ell )\,.  
	\end{aligned}\end{equation}
	It follows that
	 \begin{equation}\begin{aligned}
	Y_s^\ell &=\tilde{\mathbb{E}}^{b^{\ell ,\epsilon},d^{\ell ,\epsilon,ij}}_s\Big[ \int_s^\ell  \frac{\Gamma_u^{\ell ,\epsilon}}{\Gamma_s^{\ell ,\epsilon}}\big(m_u^{\ell ,\epsilon}+f(X_u,0,0)\big)\,du \Big]\\&
    \quad+\tilde{\mathbb{E}}^{b^{\ell ,\epsilon},d^{\ell ,\epsilon,ij}}_s\Big[\int_s^\ell  \frac{\Gamma_u^{\ell ,\epsilon}}{\Gamma_s^{\ell ,\epsilon}}\big(n_u^{\ell ,\epsilon,ij}+g_{ij}(X_u,0,0)\big)\,d\langle B^{\ell ,\epsilon,i},B^{\ell ,\epsilon,j}\rangle_u  \Big]\\
	&\leq \tilde{\mathbb{E}}^{b^{\ell ,\epsilon},d^{\ell ,\epsilon,ij}}_s\Big[ \int_s^\ell  \frac{\Gamma_u^{\ell ,\epsilon}}{\Gamma_s^{\ell ,\epsilon}}\big(C_1|X_u|+|f(0,0,0)|\big)\,du\Big]\\
    &\quad +\tilde{\mathbb{E}}^{b^{\ell ,\epsilon},d^{\ell ,\epsilon,ij}}_s\Big[\int_s^\ell  \frac{\overline{\sigma}^2\Gamma_u^{\ell ,\epsilon}}{\Gamma_s^{\ell ,\epsilon}}\big(C_1|X_u|+|g_{ij}(0,0,0)|\big)\,du\Big]\\
	&\quad+\tilde{\mathbb{E}}^{b^{\ell ,\epsilon},d^{\ell ,\epsilon,ij}}_s\Big[ \int_s^\ell  \frac{\Gamma_u^{\ell ,\epsilon}}{\Gamma_s^{\ell ,\epsilon}}m_u^{\ell ,\epsilon}\,du+\int_s^\ell  \frac{\Gamma_u^{\ell ,\epsilon}}{\Gamma_s^{\ell ,\epsilon}}n_u^{\ell ,\epsilon,ij}\,d\langle B^{\ell ,\epsilon,i},B^{\ell ,\epsilon,j}\rangle_u\Big]\,.
	\end{aligned}\end{equation}
	Incorporating this inequality with Lemma \ref{infinite mean is finite}, 
	by the similar argument in the proof of uniqueness,
	we have
	\begin{equation}\begin{aligned}
	|Y_s^\ell |\leq \frac{(1+\overline{\sigma}^2)\max(|f(0,0,0)|,|g_{ij}(0,0,0)|)}{\mu}+\frac{(1+\overline{\sigma}^2)C_1\tilde{L}(1+|X_s|)}{\mu-k}\,
	\end{aligned}\end{equation}
	for some positive constants $\tilde{L}$ and $k.$

	We now prove the inequality
	$|Y_s^x-Y_s^{x'}|\leq L|X_s^x-X_s^{x'}|
	$ in \eqref{eqn:inq}.
	Let $(\tilde{Y},\tilde{Z},\tilde{K}):=(Y^x-Y^{x'},Z^x-Z^{x'},K^x-K^{x'})$. By a similar argument in the proof of Theorem \ref{Finite QBSDE exis and unique}, for each $\epsilon>0,$ we define   processes $\tilde{a}^{\epsilon},\tilde{b}^{\epsilon},\tilde{c}^{\epsilon,ij},\tilde{d}^{\epsilon,ij},\tilde{m}^{\epsilon},\tilde{n}^{\epsilon,ij}, \tilde{\Gamma}^\epsilon$,  
  a sublinear expectation $\tilde{\mathbb{E}}^{\tilde{b}^\epsilon, \tilde{d}^{\epsilon,ij}}$ and a $G$-Brownian motion $\tilde{B}^\epsilon $ such that 
	\begin{equation}\begin{aligned}
	\tilde{Y}_s&=\tilde{Y}_T+\int_s^T(\tilde{a}^\epsilon_u\tilde{Y}_u+\tilde{b}^\epsilon_u\tilde{Z}_u+\tilde{m}_u^\epsilon+h_u)\,du\\
   & \quad +\int_s^T(\tilde{c}^{\epsilon,ij}_u\tilde{Y}_u+\tilde{d}^{\epsilon,ij}_u\tilde{Z}_u+\tilde{n}_u^{\epsilon,ij}+k^{ij}_u)\,d\langle B^i,B^j\rangle_u-\int_s^T\tilde{Z}_u\,dB_u-(\tilde{K}_T-\tilde{K}_s)\,,
	\end{aligned}\end{equation}
	where 
	$   h_s:=f(X_s^x,Y_s^{x'},Z_s^{x'})-f(X_s^{x'},Y_s^{x'},Z_s^{x'})$ and $k^{ij}_s:=g_{ij}(X_s^x,Y_s^{x'},Z_s^{x'})-g_{ij}(X_s^{x'},Y_s^{x'},Z_s^{x'}).$
	Then \begin{equation}\begin{aligned}
	\tilde{Y}_s
	\leq \tilde{\mathbb{E}}_s^{\tilde{b}^\epsilon,\tilde{d}^{\epsilon,ij}}\Big[\frac{\tilde{\Gamma}_T^\epsilon}{\tilde{\Gamma}_s^\epsilon}\tilde{Y}_T+\int_s^T\frac{\tilde{\Gamma}_u^\epsilon}{\tilde{\Gamma}_s^\epsilon}(\tilde{m}_u^\epsilon+h_u)\,du+\int_s^T\frac{\tilde{\Gamma}_u^\epsilon}{\tilde{\Gamma}_s^\epsilon}(\tilde{n}_u^{\epsilon,ij}+k^{ij}_u)\,d\langle \tilde{B}^{\epsilon,i},\tilde{B}^{\epsilon,j} \rangle _u\Big]\,.\label{eqn:4}
	\end{aligned}\end{equation}
	Observe that 
	\begin{equation}
    \begin{aligned} &\quad\tilde{\mathbb{E}}_s^{\tilde{b}^\epsilon,\tilde{d}^{\epsilon,ij}}\Big[\int_s^T\frac{\tilde{\Gamma}_u^\epsilon}{\tilde{\Gamma}_s^\epsilon} h_u\,du+\int_s^T\frac{\tilde{\Gamma}_u^\epsilon}{\tilde{\Gamma}_s^\epsilon}k^{ij}_u\,d\langle \tilde{B}^{\epsilon,i},\tilde{B}^{\epsilon,j} \rangle _u\Big]\\
	&\leq\tilde{\mathbb{E}}_s^{\tilde{b}^\epsilon,\tilde{d}^{\epsilon,ij}}\Big[\int_s^TC_1\frac{\tilde{\Gamma}_u^\epsilon}{\tilde{\Gamma}_s^\epsilon}|X_u^x-X_u^{x'}|\,du+\int_s^TC_1\overline{\sigma}^2\frac{\tilde{\Gamma}_u^\epsilon}{\tilde{\Gamma}_s^\epsilon}|X_u^x-X_u^{x'}|\,du\Big]\\
	&\leq \frac{C_1(1+\overline{\sigma}^2)}{\mu+\eta-(1+\overline{\sigma}^2)C_\sigma C_3(1+2L_1)}|X_s^x-X_s^{x'}|\label{eqn:1}\,.
	\end{aligned}\end{equation}
    By a similar argument in the proof of uniqueness, choosing a constant $k$ such that $0<k<\mu$ and $k+\eta-(1+\overline{\sigma}^2)C_\sigma C_3(1+2L_1)>0$,
	we have
	\begin{equation}
	\begin{aligned}
	\tilde{\mathbb{E}}_s^{\tilde{b}^\epsilon,\tilde{d}^{\epsilon,ij}}\Big[\frac{\tilde{\Gamma}_T^\epsilon}{\tilde{\Gamma}_s^\epsilon}\tilde{Y}_T\Big]&\leq e^{(-\mu+k)(T-s)}\tilde{\mathbb{E}}_s^{\tilde{b}^\epsilon,\tilde{d}^{\epsilon,ij}}\Big[e^{-k(T-s)}\tilde{Y}_T\Big]\\
    &\leq e^{(-\mu+k)(T-s)}\tilde{\mathbb{E}}_s^{\tilde{b}^\epsilon,\tilde{d}^{\epsilon,ij}}\Big[e^{-k(T-s)}L(2+|X_T^x|+|X_T^{x'}|)\Big]\,.\\
	\end{aligned}
	\end{equation}
	Combining this inequality with Lemma \ref{infinite mean is finite}, 
	\begin{equation}\begin{aligned}
	\lim_{T\rightarrow\infty}\tilde{\mathbb{E}}_s^{\tilde{b}^\epsilon,\tilde{d}^{\epsilon,ij}}\Big[\frac{\tilde{\Gamma}_T^\epsilon}{\tilde{\Gamma}_s^\epsilon}\tilde{Y}_T\Big]=0\,.\label{eqn:2}
	\end{aligned}\end{equation}
	Since $ |\tilde{m}^\epsilon|,|\tilde{n}^{\epsilon,ij}|\leq (C_2(1+d)+2C_3(1+2L_1))\epsilon$, 
	\begin{equation}\begin{aligned}
	&\quad\tilde{\mathbb{E}}_s^{\tilde{b}^\epsilon,\tilde{d}^{\epsilon,ij}}\Big[\int_s^T\frac{\tilde{\Gamma}_u^\epsilon}{\tilde{\Gamma}_s^\epsilon} \tilde{m}_u^\epsilon\,du+\int_s^T\frac{\tilde{\Gamma}_u^\epsilon}{\tilde{\Gamma}_s^\epsilon}\tilde{n}_u^{\epsilon,ij}\,d\langle \tilde{B}^{\epsilon,i},\tilde{B}^{\epsilon,j} \rangle _u\Big]\\
    &\leq \frac{(1+\overline{\sigma}^2)(C_2(1+d)+2C_3(1+2L_1))\epsilon}{\mu}\label{eqn:3}\,.
	\end{aligned}\end{equation}
	Substituting \eqref{eqn:1}, \eqref{eqn:2}, \eqref{eqn:3}
	to \eqref{eqn:4},  we have   
	$\tilde{Y}_s\leq L|X_s^x-X_s^{x'}|$,
	where
	$$L:=\frac{C_1(1+\overline{\sigma}^2)}{\mu+\eta-(1+\overline{\sigma}^2)C_\sigma C_3(1+2L_1)}\,.$$
	Using a similar method, it can be verified that $-\tilde{Y}_s\leq L|X_s^x-X_s^{x'}|$.
	Therefore, we obtain the inequality  
	$|Y_s^x-Y_s^{x'}|=|\tilde{Y}_s|\leq L|X_s^x-X_s^{x'}|.$ 
\end{proof}

\begin{thm}\label{ergodic fey kac formula}
	Suppose Assumptions \ref{ergodic assumption1}-\ref{ergodic assumption2} hold and
	let $(Y^x,Z^x,K^x)$ be the unique solution to \eqref{ergodic QBSDE2} satisfying \eqref{eqn:inq}. Define a function $u$ on $\mathbb{R}^m$ as $u(x):=Y_0^x.$    Then $Y_t^x=u(X_t^x)$ for $(t,x)\in [0,\infty)\times\mathbb{R}^m $
	and $u$ is a viscosity solution to the PDE
	\begin{equation}\begin{aligned}
	G(H(x, u, D_xu, D_x^2u ))+\langle b(x),D_xu\rangle+f(x,u,\langle\sigma^1(x),D_xu\rangle,\cdots,\langle\sigma^d(x),D_xu\rangle)=0\label{infinite horizon elliptic PDE}\,,
	\end{aligned}\end{equation}
	where
	\begin{equation}\begin{aligned}
	H_{ij}(x, u, D_xu, D_x^2u )=&\;\langle D^2_xu\sigma^i(x),\sigma^j(x)\rangle+2\langle D_xu,h_{ij}(x)\rangle\\&+2g_{ij}(x,u,\langle\sigma^1(x),D_xu\rangle,\cdots,\langle\sigma^d(x),D_xu\rangle)\,,
	\end{aligned}\end{equation}
	and $\sigma^i$  is the $i$-th column vector of $\sigma$ for $i=1,2,\cdots,d.$    
\end{thm}

\begin{proof}
	Define $$L_1:=\frac{\mu+\eta-C_{\sigma}C_3(1+\overline{\sigma}^2)} {4(1+\overline{\sigma}^2)C_{\sigma}C_3} $$   and $f^{L_1}(x,y,z):=f(x,y,z^{L_1})$, $g_{ij}^{L_1}(x,y,z):=g_{ij}(x,y,z^{L_1})$, where $z^{L_1}:=\frac{|z|\wedge L_1}{|z|}z$ with the convention $\frac{\,0\,}{0}=0.$ Since $|Z^x|\leq L_1$,  the $G$-BSDE \eqref{ergodic QBSDE2} can be written 
    as 
    \begin{equation}\begin{aligned}
	Y_s^x&=Y_T^x+\int_s^Tf^{L_1}(X_u^x,Y_u^x,Z_u^x)\,du+\int_s^Tg^{L_1}_{ij}(X_u^x,Y_u^x,Z_u^x)\,d\langle B^i,B^j\rangle_u\\&\quad-\int_s^TZ_u^x\,dB_u-(K_T^x-K_s^x)\,.
	\end{aligned}\end{equation} 
	Similar to the proof of \cite[Lemma 4.2]{hu2018ergodic}, one can easily show that $Y_t^x=u(X_t^x)$ for $(t,x)\in [0,\infty)\times\mathbb{R}^m.$

	The continuity of $u$ is evident from \eqref{eqn:inq}.
	Observe that
	\begin{equation}\begin{aligned}
	u(x)&=u(X_\delta^x)+\int_0^\delta f(X_u^x,u(X_u^x),Z_u^x)\,du+\int_0^\delta g_{ij}(X_u^x,u(X_u^x),Z_u^x)\,d\langle B^i,B^j\rangle_u\\&
    \quad-\int_0^\delta Z_u^x\,dB_u-K_\delta^x
	\end{aligned}\end{equation}
	for $\delta>0$.
	From this equation, a similar argument in the proof of Theorem \ref{QBSDE feykac formula}
	shows that
	$u$ is a viscosity solution to \eqref{infinite horizon elliptic PDE}.

\end{proof}

\section{Ergodic $G$-BSDEs}\label{section 5}

Consider the $G$-SDE and $G$-BSDE
\begin{align} \label{ergodic SDE3}
X_s^x&=x+\int_0^sb(X_u^x)\,du+\int_0^sh_{ij}(X_u^x)\,d\langle B^i,B^j\rangle_u+\int_0^s\sigma(X_u^x)\,dB_u\,,\\
\label{ergodic QBSDE3} Y_s^{x}&=Y_T^{x}+\int_s^T(f(X_u^{x},Z_u^{x})+\gamma^1\lambda) \,du+\int_s^T(g_{ij}(X_u^{x},Z_u^{x})+\gamma^2_{ij}\lambda)\, d\langle B^i,B^j\rangle_u\\&
\quad-\int_s^TZ_u^{x}\,dB_u-(K_T^{x}-K_s^{x}) 
\end{align}
for  $0\leq s \leq T<\infty$, where $b,h_{ij}: \mathbb{R}^m\to \mathbb{R}^m$, $\sigma:\mathbb{R}^m\to \mathbb{R}^{m\times d}$, $f,g_{ij}: \mathbb{R}^m\times \mathbb{R}^d\to
\mathbb{R}$ are functions, $\gamma^1$ is a constant and  $\gamma^2=(\gamma_{ij}^2)_{i,j=1}^d$ is a symmetric matrix. Occasionally, we write $X^x,Y^x,Z^x,K^x$ as $X,Y,Z,K$, repectively, omitting the superscripts $x$.
\begin{defi}
	A quartet $(Y,Z,K,\lambda)$ is a solution to the $G$-BSDE \eqref{ergodic QBSDE3} if
	\begin{enumerate}[label=(\roman*)]
		\item  $(Y,Z)\in\mathbb{S}_G^2(0,\infty)\times\mathbb{H}_G^{2}(0,\infty;\mathbb{R}^d)\,,$ 
		\item  the process $K$ is a decreasing $G$-martingale with $K_0=0$ and $K_T\in \mathbb{L}^2_G(\Omega_T)$ for each $T>0$\,,
		\item $\lambda$ is a real number\,,
		\item  the quartet $(Y,Z,K,\lambda)$ satisfies \eqref{ergodic QBSDE3} quasi-surely for $0\leq s\leq T<\infty$\,.
	\end{enumerate}
\end{defi}

\begin{assume}\label{ergodic assumption3} Assume $b,h_{ij},\sigma,f,g_{ij},\gamma^1,\gamma^2$ satisfy the following properties.
	\begin{enumerate}[label=(\roman*)]%[font={\bfseries},label=({A\arabic*})] 
		\item   $\gamma^1+2G(\gamma^2)<0$. 
		\item  There exist  constants $C_1,C_3$ such that 
		\begin{equation}
        \begin{aligned}
		|f(x,z)-f(x'&,z')|+\sum_{i,j=1}^d|g_{ij}(x,z)-g_{ij}(x',z')|\\&\leq C_1|x-x'|+C_3(1+|z|+|z'|)|z-z'|\,,
		\end{aligned}
        \end{equation}
		for $x,x'\in \mathbb{R}^m$, $z,z'\in \mathbb{R}^d$.
		\item\label{eta ergodic} There is a constant $\eta>0$ such that
		\begin{equation}\begin{aligned}
		&G\big( \sum_{j=1}^m(\sigma_j(x)-\sigma_j(x'))^{\top} (\sigma_j(x)-\sigma_j(x')) +2(\langle x-x',h_{ij}(x)-h_{ij}(x')\rangle)_{i,j=1}^d\big)\\
		&+\langle x-x',b(x)-b(x')\rangle\leq -\eta|x-x'|^2\,
		\end{aligned}\end{equation}
		for $x,x'\in \mathbb{R}^m$,
		where $\sigma_j$ is the $j$-th row vector of $\sigma$ for $j=1,2,\cdots,m.$
		\item $
		\eta>(1+\overline{\sigma}^2)\Bigg(C_\sigma C_3+4\sqrt{C_\sigma C_1C_3\frac{\overline{\sigma}M_\sigma}{\underline{\sigma}}} \Bigg)\,.$
	\end{enumerate}
\end{assume}

The following lemma is crucial for our upcoming discussion as it can further refine the estimation of $|Y|$.
We recall 
the sublinear conditional expectation $\tilde{\mathbb{E}}_t^{b,d^{ij}}$  in Theorem \ref{SDE sublinear property}.

\begin{lemma}\label{Ergodic mean is finite}   Let Assumption \ref{ergodic assumption1} and \ref{eta ergodic} in Assumption \ref{ergodic assumption3} hold and  $X^{x}$ be a solution to \eqref{ergodic SDE3} for $x\in \mathbb{R}^m$. Suppose  $p\ge1$ and $b,d^{ij}\in \mathbb{M}^2_G(0,\infty;\mathbb{R}^d)$ are  processes satisfying $d^{ij}=d^{ji}$ and $|b|,|d^{ij}|\le C_4$ for some constant $C_4>0$.
	Then there exists a  constant $L>0$ depending only on $\underline{\sigma},\overline{\sigma}, C_{\sigma}, C_4,\eta, p$
	such that	
	\begin{equation}\begin{aligned}
	\tilde{\mathbb{E}}_t^{b,d^{ij}}[|X_s^x|^p]\le L(1+|X_t^x|^p) 
	\end{aligned}\end{equation}
	for all $0\le t\le s<\infty$ and $x\in \mathbb{R}^m.$ 
\end{lemma}

\begin{proof}
Denote as $X=X^x$ for convenience.
 Recall that $B_s^{b,d^{ij}}=B_s-\int_0^s b_u\,du-\int_0^s d^{ij}_u\,d\langle B^i,B^j\rangle_u,s\ge0$ is the $G$-Brownian motion under the sublinear expectation $\tilde{\mathbb{E}}^{b,d^{ij}}$.  Hereafter, $L$ is a generic constant depending only on $\underline{\sigma},\overline{\sigma}, C_{\sigma}, C_4,\eta,p$ and may differ line by line. 
	By Jensen's inequality, it suffices to prove 
	\begin{equation}\begin{aligned}
	\tilde{\mathbb{E}}^{b,d^{ij}}_t[|X_s|^{2p}]\leq L(1+|X_t|^{2p})
	\end{aligned}\end{equation} for positive even integers $2p$.  
	For $p\in \mathbb{N}$, applying the It\^{o} formula, we obtain
	\begin{equation}\begin{aligned}
	e^{p\eta s}|X_s|^{2p}&\le e^{p\eta t}|X_t|^{2p}+\int_t^s p\eta e^{p\eta u}|X_u|^{2p}\,du\\
    &\quad+\int_t^s 2pe^{p\eta u}|X_u|^{2p-2}\langle X_u,b(X_u)-b(0)\rangle\,du\\
	&\quad+\int_t^s 2pe^{p\eta u}|X_u|^{2p-2}G(\langle 2X_u,h_{ij}(X_u)-h_{ij}(0)\rangle)\,du+\Lambda_s+M_s\,,
	\end{aligned}\end{equation}
	where 
	\begin{equation}\begin{aligned}
	M_s&=\int_t^s 2pe^{p\eta u}|X_u|^{2p-2}X_u^{\top}\sigma(X_u)\,d B^{b,d^{ij}}_u\,,\\
	\Lambda_s&=\int_t^s 2pe^{p\eta u}|X_u|^{2p-2}\langle X_u, b(0)+\sigma(X_u)b_u\rangle \,du\\
    &\quad+\int_t^s 4pe^{p\eta u}|X_u|^{2p-2}G(\langle X_u,h_{ij}(0)\rangle)\,du\\
    &\quad+\int_t^s 4pe^{p\eta u}|X_u|^{2p-2}G(\langle X_u,\sigma(X_u)d_u^{ij}\rangle)\,du\\
    &\quad+\int_t^s2pe^{p\eta u}|X_u|^{2p-2} G(\sigma^{\top}(X_u)\sigma(X_u))\,du\\
    &\quad+\int_t^s4p(p-1)e^{p\eta u}|X_u|^{2p-4}G(\sigma^{\top}(X_u)X_uX_u^{\top}\sigma(X_u))\,du\,.
	\end{aligned}\end{equation}
	Since $G$ is a monotone increasing function, $\langle x,b(x)-b(0)\rangle+G(2\langle x,h_{ij}(x)-h_{ij}(0)\rangle )\leq -\eta|x|^2$. It follows that
	\begin{equation}\begin{aligned}
	e^{p\eta s}|X_s|^{2p}&\leq e^{p\eta t}|X_t|^{2p}-\int_t^s p\eta e^{p\eta u}|X_u|^{2p}\,du+\Lambda_s+M_s\,.
	\end{aligned}\end{equation} Observe that 
    \begin{align}
        \Lambda_s\le &\int_t^s 2pe^{p\eta u}|X_u|^{2p-2}|\langle X_u, b(0)+\sigma(X_u)b_u\rangle|\,du\\
    &+\int_t^s 2p\overline{\sigma}^2e^{p\eta u}|X_u|^{2p-2}|\langle X_u,h_{ij}(0)\rangle|\,du\\
    &+\int_t^s 2p\overline{\sigma}^2e^{p\eta u}|X_u|^{2p-2}|\langle X_u,\sigma(X_u)d_u^{ij}\rangle|\,du\\
    &+\int_t^sp\overline{\sigma}^2e^{p\eta u}|X_u|^{2p-2} |\sigma(X_u)|^2\,du\\
    &+\int_t^s2p(p-1)\overline{\sigma}^2e^{p\eta u}|X_u|^{2p-4}X_u^{\top}\sigma(X_u)\sigma^{\top}(X_u)X_u\,du\,\\
    \le& \int_t^s L_1 e^{p\eta u}|X_u|^{2p-2}(|X_u|+1)\,du\,,
    \end{align}
    where $L_1$ is a constant depending only on $\overline{\sigma}, M_\sigma,C_4,p.$
    From the inequality $ab\leq \frac{ca^2}{2}+\frac{b^2}{2c}$ for $c\geq 0$ with $a=X_u$, $b=L_1$, $c=2p\eta$,
     we have
	\begin{equation}\begin{aligned}
	\Lambda_s\leq p\eta\int_t^s e^{p\eta u}|X_u|^{2p}\,du+L\int_t^se^{p\eta u}|X_u|^{2p-2}\,du\,.
	\end{aligned}\end{equation}
	Since $M$ is a $G$-martingale, 
	\begin{equation}
    \begin{aligned}
	\tilde{\mathbb{E}}^{b,d^{ij}}_t[e^{p\eta s}|X_s|^{2p}]&\leq e^{p\eta t}|X_t|^{2p}+L\tilde{\mathbb{E}}^{b,d^{ij}}_t\Big[\int_t^s e^{p\eta u}|X_u|^{2p-2}\,du\Big]\\
    &\leq e^{p\eta t}|X_t|^{2p}+L\int_t^s e^{p\eta u}\tilde{\mathbb{E}}^{b,d^{ij}}_t[|X_u|^{2p-2}]\,du\,.
	\end{aligned}
    \end{equation}
	
	We now apply the mathematical induction on $p.$ For $p=1$,
	\begin{equation}\begin{aligned}
	\tilde{\mathbb{E}}^{b,d^{ij}}_t[e^{\eta s}|X_s|^{2}]\leq e^{\eta t}|X_t|^{2}+Le^{\eta s}\,,
	\end{aligned}\end{equation}
	thus $\tilde{\mathbb{E}}^{b,d^{ij}}_t[|X_s|^2]\leq L(1+|X_t|^2).$
	By the induction hypothesis for $p-1,$
	$$\tilde{\mathbb{E}}^{b,d^{ij}}_t[|X_u|^{2p-2}]\le L(1+|X_t|^{2p-2})\leq L(2+|X_t|^{2p})\,,$$
	which gives 
	$\tilde{\mathbb{E}}^{b,d^{ij}}_t[|X_s|^{2p}]\leq L(1+|X_t|^{2p}).$ This completes the proof.
\end{proof}

\begin{thm}\label{QBSDE with lambda sol exist}
	Let Assumptions \ref{ergodic assumption1} and \ref{ergodic assumption3} hold. Then,
	there exists a solution $(Y^x,Z^x,K^x,\lambda)$ to the ergodic $G$-BSDE
	\eqref{ergodic QBSDE3}    such that $|Y^x|\leq L|X^x|$ and $|Z^x|\le L$ for a constant $L>0$ depending only on $\overline{\sigma},  C_1,C_3,C_{\sigma},\eta$.
\end{thm}

\begin{proof}
	For simplicity,  $\gamma^1+2G(\gamma^2)=-1$. Let $\delta$ be a positive constant satisfying 
	\begin{equation}\begin{aligned}
	\eta-(1+\overline{\sigma}^2)C_\sigma C_3\Big(1+\frac{\delta+\eta-C_{\sigma}C_3(1+\overline{\sigma}^2)}{2(1+\overline{\sigma}^2)C_{\sigma}C_3}\Big)>0\,.    
	\end{aligned}\end{equation} 
	According to Theorem \ref{infinite horizon Quadartic BSDE exist and unique}, there exists a unique solution $(Y^{\delta},Z^{\delta},K^{\delta})=(Y^{x,\delta},Z^{x,\delta},K^{x,\delta})$ to the infinite-horizon $G$-BSDE
	\begin{equation}\begin{aligned}
	Y_s^{\delta}&=Y_T^{\delta}+\int_s^T(f(X_u,Z_u^{\delta})+\gamma^1\delta Y_u^{\delta})\,du\\&
    \quad+\int_s^T(g_{ij}(X_u,Z_u^{\delta})+\gamma^2_{ij}\delta Y_u^{\delta})\, d\langle B^i,B^j\rangle_u-\int_s^TZ_u^{\delta}\,dB_u-(K_T^{\delta}-K_s^{\delta}) \,.
	\end{aligned}\end{equation}
	We claim that there exists a constant $L>0,$ which does not depend on $\delta,$ such that 
	\begin{equation}\begin{aligned}\label{eqn: 1/delta}
	|Y^{\delta}|\leq \frac{L}{\delta}(1+|X|)\,.
	\end{aligned}\end{equation}
	Through a similar argument in the proof of Theorem \ref{Finite QBSDE exis and unique}, for each $\epsilon>0,$ we can construct processes $  b^{\delta,\epsilon},d^{\delta,\epsilon,ij},m^{\delta,\epsilon},n^{\delta,\epsilon,ij},\Gamma^{\delta,\epsilon},$  a sublinear expectation $\tilde{\mathbb{E}}^{b^{\delta,\epsilon},d^{\delta,\epsilon,ij}}$ and a $G$-Brownian motion $B^{\delta,\epsilon}$ such that
	\begin{equation}\begin{aligned}
	Y_s^{\delta} &=Y_T^\delta+\int_s^T (\gamma^1\delta Y_u^{\delta}+b_u^{\delta,\epsilon}Z_u^\delta+m_u^{\delta,\epsilon}+f(X_u,0))\,du\\&
    \quad+\int_s^T(\gamma^2_{ij}\delta Y_u^\delta+d_u^{\delta,\epsilon,ij}Z_u^\delta+n_u^{\delta,\epsilon,ij}+g_{ij}(X_u,0))\,d\langle B^i,B^j\rangle_u\\
    &\quad-\int_s^TZ_u^\delta\,dB_u-(K_T^\delta-K_s^\delta)\,.
	\end{aligned}\end{equation} 
	Then
\begin{equation}
	\begin{aligned}
	Y_s^{\delta}&\le \tilde{\mathbb{E}}^{b^{\delta,\epsilon},d^{\delta,\epsilon,ij}}_s\Big[\frac{\Gamma_T^{\delta,\epsilon}}{\Gamma_s^{\delta,\epsilon}}Y_T^\delta+ \int_s^T \frac{\Gamma_u^{\delta,\epsilon}}{\Gamma_s^{\delta,\epsilon}}\big(m_u^{\delta,\epsilon}+f(X_u,0)\big)\,du \Big]\\
    &\quad+\tilde{\mathbb{E}}^{b^{\delta,\epsilon},d^{\delta,\epsilon,ij}}_s\Big[\int_s^T \frac{\Gamma_u^{\delta,\epsilon}}{\Gamma_s^{\delta,\epsilon}}\big(n_u^{\delta,\epsilon,ij}+g_{ij} (X_u,0)\big)\,d\langle B^{\delta,\epsilon,i},B^{\delta, \epsilon,j}\rangle_u  \Big]\,.
	\end{aligned}
	\end{equation}
	It follows that
	\begin{equation}
	\begin{aligned}
	Y_s^{\delta} 
	&\leq\tilde{\mathbb{E}}^{b^{\delta,\epsilon},d^{\delta,\epsilon,ij}}_s\Big[\frac{\Gamma_T^{\delta,\epsilon}}{\Gamma_s^{\delta,\epsilon}}Y_T^\delta \Big]+\tilde{\mathbb{E}}^{b^{\delta,\epsilon},d^{\delta,\epsilon,ij}}_s\Big[ \int_s^T \frac{\Gamma_u^{\delta,\epsilon}}{\Gamma_s^{\delta,\epsilon}}m_u^{\delta,\epsilon}\,du+\int_s^T \frac{2\Gamma_u^{\delta,\epsilon}}{\Gamma_s^{\delta,\epsilon}}G(n_u^{\delta,\epsilon,ij})\,du\Big]\\
    &\quad+ \tilde{\mathbb{E}}^{b^{\delta,\epsilon},d^{\delta,\epsilon,ij}}_s\Big[ \int_s^T \frac{\Gamma_u^{\delta,\epsilon}}{\Gamma_s^{\delta,\epsilon}}|f(X_u,0)|\,du+\int_s^T \frac{2\Gamma_u^{\delta,\epsilon}}{\Gamma_s^{\delta,\epsilon}}G(g_{ij}(X_u,0))\,du\Big]\,.	\end{aligned}
	\end{equation}
	From \eqref{eqn:inq} and Lemma \ref{Ergodic mean is finite}, it can be shown that  
	\begin{equation}
	\begin{aligned}
	&\quad\tilde{\mathbb{E}}^{b^{\delta,\epsilon},d^{\delta,\epsilon,ij}}_s\Big[ \int_s^T \frac{\Gamma_u^{\delta,\epsilon}}{\Gamma_s^{\delta,\epsilon}}|m_u^{\delta,\epsilon}|\,du+\int_s^T \frac{2\Gamma_u^{\delta,\epsilon}}{\Gamma_s^{\delta,\epsilon}}G(n_u^{\delta,\epsilon,ij})\,du\Big]\leq \frac{L\epsilon}{\delta}\,,\\    
	&\quad\tilde{\mathbb{E}}^{b^{\delta,\epsilon},d^{\delta,\epsilon,ij}}_s\Big[\frac{\Gamma_T^{\delta,\epsilon}}{\Gamma_s^{\delta,\epsilon}}|Y_T^\delta | \Big]\leq e^{-\delta(T-s)}\tilde{\mathbb{E}}^{b^{\delta,\epsilon},d^{\delta,\epsilon,ij}}_s[|Y_T^\delta |]\leq e^{-\delta(T-s)}\tilde{L}(1+|X_s|)\,,\\
	&\quad\tilde{\mathbb{E}}^{b^{\delta,\epsilon},d^{\delta,\epsilon,ij}}_s\Big[ \int_s^T \frac{\Gamma_u^{\delta,\epsilon}}{\Gamma_s^{\delta,\epsilon}}|f(X_u,0)|\,du+\int_s^T \frac{2\Gamma_u^{\delta,\epsilon}}{\Gamma_s^{\delta,\epsilon}}G(g_{ij}(X_u,0))\,du\,\Big] 
	\\
    &\leq\frac{(1+\overline{\sigma}^2)L}{\delta}(1+|X_s|)\,,
	\end{aligned}
	\end{equation}
	where $L>0$ is a constant which does not depend on $\delta$ and $\tilde{L}>0$ is a constant which may depend on $\delta$. Therefore, letting $\epsilon \to 0$ and $T\to \infty,$ we obtain the inequality
	\eqref{eqn: 1/delta}.
	
We now construct a  solution $(Y,Z,K,\lambda)$ to the BSDE \eqref{ergodic QBSDE3}.
	Define $u^\delta(x)=Y_0^{\delta,x}$. According to Theorem \ref{infinite horizon Quadartic BSDE exist and unique} and \eqref{eqn: 1/delta}, we have
	\begin{equation}\begin{aligned}
	|u^{\delta}(x)|&\leq \frac{L}{\delta}(1+|x|)\,,\\
	|u^{\delta}(x)-u^{\delta}(x')|&\leq \frac{C_1(1+\overline{\sigma}^2)}{\eta-(1+\overline{\sigma}^2)C_\sigma C_3(1+2L_1)} |x-x'|\,,\\
	|Z^{\delta}|&\leq \frac{\delta+\eta-C_{\sigma}C_3(1+\overline{\sigma}^2)}{4(1+\overline{\sigma}^2)C_{\sigma}C_3}\,,
	\end{aligned}\end{equation}
	where $L>0$ is a constant which does not depend on $\delta$.
   Let $\overline{u}^\delta(x):=u^{\delta}(x)-u^{\delta}(0)$ and  $\overline{Y}^{\delta}:=Y^{\delta}-u^{\delta}(0)$. 
By the standard diagonal procedure, there exists a sequence $(\overline{u}^{\delta_n})_{n\in\mathbb{N}}$ such that $\overline{u}^{\delta_n}\to u$ locally uniformly and  $\delta_n u^{\delta_n}(0)\rightarrow \lambda$ as $\delta_n \rightarrow 0$ for some continuous function $u$ and some real number $\lambda$.
	We define a process $Y_s=u(X_s)$ for $s\ge0$.
Following the proof of
	\cite[Theorem 5.1]{hu2018ergodic}, for all $T>0,$  the limit 
 $Y:=\lim_{\delta_n \rightarrow 0}\overline{Y}^{\delta_n}$ exists in $\mathbb{S}^2_G(0,T)$,  $Z:=\lim_{\delta_n \rightarrow 0}Z^{\delta_n}$ exists in $\mathbb{H}^2_G(0,T;\mathbb{R}^d)$, and $K_s:=\lim_{\delta_n \rightarrow 0}K_s^{\delta_n}$ exists
 in $\mathbb{L}^2_G(\Omega_s)$ for all $s\le T.$
Moreover, the quartet $(Y,Z,K,\lambda)$ is a solution to \eqref{ergodic QBSDE3}.
Since  $(Y,Z,K)$ is the limit of $(\overline{Y}^{\delta_n},Z^{\delta_n},K^{\delta_n})$ as $\delta_n\to0$, it follows that  
	\begin{equation}\begin{aligned}
	|Y|&\leq \frac{C_1(1+\overline{\sigma}^2)}{\eta-(1+\overline{\sigma}^2)C_\sigma C_3(1+2L_1)} |X|\,,\\
	|Z|&\leq \frac{\eta-C_{\sigma}C_3(1+\overline{\sigma}^2)}{4(1+\overline{\sigma}^2)C_{\sigma}C_3}\,.
	\end{aligned}\end{equation}
	This completes the proof.
\end{proof}

A solution to \eqref{ergodic QBSDE3} is not unique. 
If $(Y,Z,K,\lambda)$ is a solution, then  $(Y+c,Z,K,\lambda)$ is also a solution
for any $c\in \mathbb{R}$. However,   $\lambda$ is unique when the solution space is restricted. 
The following theorem is derived directly by combining Lemma \ref{Ergodic mean is finite} with the proof of \cite[Theorem 5.3]{hu2018ergodic}, therefore we omit the proof.

\begin{thm}\label{lambda unique}
	Suppose Assumptions \ref{ergodic assumption1} and \ref{ergodic assumption3} hold.
	For $i=1,2,$ let   $(Y^i,Z^i,K^i,\lambda^i)$ be a solution to \eqref{ergodic QBSDE3} such that $|Y^i|\leq L_i(1+|X|^{p})$ and $|Z^i|\leq L_i$ for some constant $L_i>0$, $p\geq 1$.
	Then $\lambda^1=\lambda^2.$  
\end{thm}

We now investigate the uniqueness of solutions to the ergodic $G$-BSDE.
Let
 \begin{equation}
\begin{aligned}
\mathbb{L}^1(\Omega)&:=\{X:\Omega\to [-\infty,\infty]\,|\,  X \text{ is } \mathcal{F}\text{-measurable} \}\,,\\
%\mathbb{L}^1(\Omega_t)&=\{X:\Omega_t\to [-\infty,\infty]\,|\, X \text{ is } \mathcal{F}_t\text{-measurable} \}\,,\\
\mathbb{L}^{1^\ast}_G(\Omega)&:=\{X\in \mathbb{L}^1(\Omega)\,|\, \hat{\mathbb{E}}[|X|]<\infty,\, \text{ there exists } (X_n)_{n\in\mathbb{N}}\subseteq \mathbb{L}^1_G(\Omega) \\&\quad\quad\quad\quad\quad\quad\quad\quad\quad\quad\text{ such that } X_n\searrow X \text{ quasi-surely} \}\,,\\
%\mathbb{L}^{1^\ast}_G(\Omega_t)&=\{X\in \mathbb{L}^1(\Omega_t)\,|\, \hat{\mathbb{E}}[|X|]<\infty,\, \exists (X_n)_{n\in\mathbb{N}}\subseteq \mathbb{L}^1_G(\Omega_t) \text{ such that } X_n\searrow X \text{ quasi-surely} \}\,.
\end{aligned}
\end{equation}
and define two spaces
$\mathbb{L}^1(\Omega_t)$ and $
\mathbb{L}^{1^\ast}_G(\Omega_t)$ similarly.
We define a sublinear conditional expectation $\hat{\mathbb{E}}_t:\mathbb{L}^{1^\ast}_G(\Omega)\to \mathbb{L}^{1^\ast}_G(\Omega_t)$
as 
\begin{equation}\begin{aligned}
\hat{\mathbb{E}}_t[X]=\lim_{n\to \infty}\hat{\mathbb{E}}_t[X_n] \text{  quasi-surely}   
\end{aligned}\end{equation}
for $X\in \mathbb{L}^{1^{\ast}}_G(\Omega)$\,,
where $(X_n)_{n\in\mathbb{N}}\subseteq \mathbb{L}_G^1(\Omega)$ and $X_n \searrow X$ quasi-surely. For more details, refer to  \cite{hu2021extended}. 
We say a stopping time $\tau:\Omega\to [0,\infty)$ is a $\ast$-stopping time if $\mathds{1}_{\{\tau\geq t\}}\in \mathbb{L}^{1^\ast}_G(\Omega_t)$ for all $t\geq 0$.

\begin{assume}\label{Assumption ergodic unique}
	Assume that the function $\sigma$ satisfies 
	\begin{equation}\begin{aligned}
	v^\top\sigma(x)\sigma^{\top}(x)v>0
	\end{aligned}\end{equation}
	for all $v \in \mathbb{R}^m\setminus{\{0\}}$ and $x\in \mathbb{R}^m.$
\end{assume}

\begin{lemma}\label{ergodic recurrent}
	Let Assumptions \ref{ergodic assumption1},  \ref{Assumption ergodic unique} and \ref{eta ergodic} in Assumption \ref{ergodic assumption3} hold and $X^{x}$ be a solution to \eqref{ergodic SDE3} for $x\in \mathbb{R}^m.$
 Suppose $b,d^{ij}\in \mathbb{M}^2_G(0,\infty;\mathbb{R}^d)$ are processes satisfying $d^{ij}=d^{ji}$ and $|b|,|d^{ij}|\le C_4$ for some constant $C_4>0$.
 Then for any bounded open subset $U$ in $\mathbb{R}^m$ and a stopping time $\tau_U:=\inf\{t\geq0\, |\, X_t^x\in U\}$,  we have
	\begin{equation}\begin{aligned}
	\lim_{T\to\infty}\tilde{\mathbb{E}}^{b,d^{ij}}[\mathds{1}_{\{\tau_U\geq T\}}]=0\,.
	\end{aligned}\end{equation} 
\end{lemma}
\begin{proof}
 Let $\tau:=\tau_U$. 
	Without loss of generality, we may assume that  $U$ is an open ball of radius $r>0$ centered at the origin and $x\in \mathbb{R}^m\setminus U.$
	Recall that $B^{b,d^{ij}}=B-\int_0^\cdot b_u\,du-\int_0^\cdot d^{ij}_u\,d\langle B^i,B^j\rangle_u$ is a $G$-Brownian motion under the sublinear expectation $\tilde{\mathbb{E}}^{b,d^{ij}}.$
	Let $\tilde{\mathcal{P}}$ be a set that represents $\tilde{\mathbb{E}}^{b,d^{ij}}.$ 
	Define 
	$f(\cdot)=\frac{2}{r^{2p}}-\frac{1}{\,|\,\cdot\,|^{2p}}$  for a positive integer $p$ which will be chosen later.
	Applying It\^{o}'s formula to $f(X)$ for each $\tilde{P}\in\tilde{\mathcal{P}}$, we have
	\begin{equation}\begin{aligned}
	f(X_{T\wedge \tau})&\le f(x)+\int_0^{T\wedge\tau}2p|X_u|^{-2p-2} \langle X_u,b(X_u)-b(0)\rangle\,du\\
    &\quad+\int_0^{T\wedge\tau}2p|X_u|^{-2p-2}G(\langle 2X_u,h_{ij}(X_u)-h_{ij}(0) \rangle )  \,du\\
    &\quad-\int_0^{T\wedge \tau}2p(p+1)\underline{\sigma}^2|X_u|^{-2p-4}X_u^{\top}\sigma(X_u)\sigma^{\top}(X_u)X_u\,du\\ 
	&\quad+\Lambda_{T\wedge\tau}+M_{T\wedge\tau}\,\,\, \tilde{P}\text{-almost surely, }
	\end{aligned}\end{equation}
	where 
	\begin{equation}
	\begin{aligned}
	M&=\int_0^{\cdot}2p|X_u|^{-2p-2}X_u^{\top}\sigma(X_u)\,dB^{b,d^{ij}}_u\,,\\
	\Lambda&=\int_0^{\cdot}2p|X_u|^{-2p-2}\langle X_u,b(0)+\sigma(X_u)b_u\rangle\,du+\int_0^{\cdot}4p|X_u|^{-2p-2}G(\langle X_u,h_{ij}(0)\rangle )\,du\\
    &\quad+\int_0^{\cdot} 4p |X_u|^{-2p-2}G(\langle X_u, \sigma(X_u)d^{ij}_u \rangle) \,du+\int_0^{\cdot}2p|X_u|^{-2p-2}G(\sigma^{\top}(X_u)\sigma(X_u))\,du\,.
	\end{aligned}
	\end{equation}
	Since $G$ is monotone increasing, we know $\langle x,b(x)-b(0)\rangle+G(2\langle x,h_{ij}(x)-h_{ij}(0)\rangle )\leq -\eta|x|^2.$ Then 
	\begin{equation}\begin{aligned}
	f(X_{T\wedge\tau})&\leq f(x)-\int_0^{T\wedge\tau}2p\eta|X_u|^{-2p}\,du\\
    &\quad-\int_0^{T\wedge \tau}2p(p+1)\underline{\sigma}^2|X_u|^{-2p-4}X_u^{\top}\sigma(X_u)\sigma^{\top}(X_u)X_u\,du+\Lambda_{T\wedge\tau}+M_{T\wedge\tau}\,.
	\end{aligned}\end{equation}
	
	Recalling the inequality $ab\leq \frac{ca^2}{2}+\frac{b^2}{2c}$ for $c\geq0$, we can choose a constant $L,$ which depends only on $|b(0)|,|h_{ij}(0)|,C_4,M_\sigma, \eta,$ such that
	\begin{equation}\begin{aligned}
	f(X_{T\wedge\tau})&\leq f(x)+M_{T\wedge\tau}\\
	&\quad+\int_0^{T\wedge\tau}p|X_u|^{-2p-2}\Big(-\eta|X_u|^2+L-2(p+1)\underline{\sigma}^2 \frac{X_u^{\top}\sigma(X_u)\sigma^{\top}(X_u)X_u}{|X_u|^2}\Big)\,du\,.\\
	\end{aligned}\end{equation}
	By Assumption \ref{Assumption ergodic unique},
	we can choose a positive integer $p$ such that  $-\eta|y|^2+L-4(p+1)\underline{\sigma}^2\frac
    {y^\top\sigma(y)\sigma^{\top}(y)y}{|y|^2}\leq -1$ for all $|y|\ge r.$ 
	This yields that
	\begin{equation}\begin{aligned}\label{eqn:recur}
	E^{\tilde{P}}\Big[\int_0^{T}\mathds{1}_{[0,\tau]}(u)|X_u|^{-2p-2}\,du\Big]\leq \frac{f(x)-E^{\tilde{P}}[f(X_{T\wedge\tau})]}{p}\leq \frac{f(x)}{p}\,. 
	\end{aligned}\end{equation}
	Using Jensen's inequality, the reverse H\"{o}lder inequality  
	\begin{equation}\begin{aligned}\label{reverse holder}
	\Big(E^{\tilde{P}}[\mathds{1}_{[0,\tau]}(u)]\Big)^2\Big(E^{\tilde{P}}[|X_u|^{2p+2}]\Big)^{-1} \leq  E^{\tilde{P}}[\mathds{1}_{[0,\tau]}(u)|X_u|^{-2p-2}]
	\end{aligned}\end{equation}
	and $E^{\tilde{P}}[|X_u|^{2p+2}]\leq L(1+|x|^{2p+2})$ 
	for the constant $L$ in 
	Lemma \ref {Ergodic mean is finite},
	we obtain
	\begin{equation}\begin{aligned}
	\Big(E^{\tilde{P}}[T\wedge\tau]\Big)^2&=\Big(E^{\tilde{P}}\Big[\int_0^T\mathds{1}_{[0,\tau]}(u)\,du\Big]\Big)^2\leq T\int_0^T \Big(E^{\tilde{P}}[\mathds{1}_{[0,\tau]}(u)]\Big)^2\,du\\&\leq \frac{TL(1+|x|^{2p+2})f(x)}{p}\,.
	\end{aligned}\end{equation}
	Taking the supremum over $\tilde{\mathcal{P}}$,  
	\begin{equation}\begin{aligned}
	\tilde{\mathbb{E}}^{b,d^{ij}}[\mathds{1}_{\{\tau\geq T\}}]\leq \frac{1}{T}\tilde{\mathbb{E}}^{b,d^{ij}}[T\wedge\tau]\leq \sqrt{
		\frac{L(1+|x|^{2p+2})f(x)}{Tp}}\,.
	\end{aligned}\end{equation}
	This completes the proof.
\end{proof}

The following theorem establishes the uniqueness of solutions to \eqref{ergodic QBSDE3} under suitable conditions. While the uniqueness of  $\lambda$ has been addressed in the existing literature, the uniqueness of the processes $Y$ (up to an additive constant), $Z$, and $K$ has not been previously studied.

\begin{thm}\label{markovian ergodic solution is unique}
	Suppose Assumptions \ref{ergodic assumption1}, \ref{ergodic assumption3}, \ref{Assumption ergodic unique} hold. For $i=1,2$, let $(Y^i,Z^i,K^i,\lambda^i)$ be a solution to \eqref{ergodic QBSDE3} such that $Z^i$ is bounded. If there exists a continuous function  $u^i:\mathbb{R}^m\to\mathbb{R}$ with polynomial growth such that $u^i(X)=Y^i$ for $i=1,2,$  then $u^1=u^2+L$ for some constant $L$ and $Z^1=Z^2,$ $K^1=K^2,$ $\lambda^1=\lambda^2$.
\end{thm}
\begin{proof}
 It is evident that $\lambda^1=\lambda^2$ according to Theorem \ref{lambda unique}. We first prove $u^1=u^2+L$ for some constant $L$. Since $u^i(X)-u^i(0)$ is also a solution to \eqref{ergodic QBSDE3}, we may assume $u^i(0)=0$ for $i=1,2$.  For any $\delta> 0$, there exists $r>0$ such that $|u^1(x)-u^2(x)|\leq \delta$ for $|x|\leq r$.
	Let $\tau:=\inf\{t\geq0 \,|\, |X_t|<r\}$.
	Then $\tau\wedge T$ is a $\ast$-stopping time by \cite[Example 46]{hu2021extended}.
	Define $(\hat{Y},\hat{Z},\hat{K})=(Y^1-Y^2,Z^1-Z^2,K^1-K^2)$.   
	Hereafter, $L>0$ is a generic constant depending only on $C_1,C_3,C_\sigma, M_\sigma,\eta, x$ and may differ line by line.

	By a similar argument in the proof of Theorem \ref{infinite horizon Quadartic BSDE exist and unique}, we can construct four processes $b^{\epsilon}, d^{\epsilon,ij}, m^{\epsilon}, n^{\epsilon,ij}, $ a sublinear expectation $\hat{\mathbb{E}}^{b^\epsilon,d^{\epsilon,ij}}$  and a $G$-Brownian motion $B^\epsilon$ such that 
	\begin{equation}\begin{aligned}
	\hat{Y}_0=\hat{Y}_{T\wedge\tau}+\int_0^{T\wedge\tau}m_u^{\epsilon}\,du+\int_0^{T\wedge\tau}n_u^{\epsilon,ij}\,d\langle B^{\epsilon,i},B^{\epsilon,j}\rangle_u-\int_0^{T\wedge\tau}\hat{Z}_u\,dB^\epsilon_u-\hat{K}_{T\wedge\tau}\,.
	\end{aligned}\end{equation}
	Then     
	\begin{equation} 
	\begin{aligned}
	\hat{Y}_0+\tilde{\mathbb{E}}^{b^\epsilon,d^{\epsilon,ij}}[{K}^1_{T\wedge\tau}]&\le \tilde{\mathbb{E}}^{b^\epsilon,d^{\epsilon,ij}}[|\hat{Y}_{T\wedge\tau}|]+\tilde{\mathbb{E}}^{b^\epsilon,d^{\epsilon,ij}}\Big[\int_0^{T\wedge\tau}|m_u^{\epsilon}|\,du\Big]
	\\
	&\;\;+\tilde{\mathbb{E}}^{b^\epsilon,d^{\epsilon,ij}}\Big[\Big|\int_0^{T\wedge\tau}n_u^{\epsilon,ij}\,d\langle B^{\epsilon,i},B^{\epsilon,j}\rangle_u\Big|\Big]+\tilde{\mathbb{E}}^{b^\epsilon,d^{\epsilon,ij}}\Big[-\int_0^{T\wedge\tau}\hat{Z}_u\,dB^\epsilon_u\Big] \\
	&\le \tilde{\mathbb{E}}^{b^\epsilon,d^{\epsilon,ij}}[|\hat{Y}_{T\wedge\tau}|]
	+L\epsilon T
	+\tilde{\mathbb{E}}^{b^\epsilon,d^{\epsilon,ij}}\Big[-\int_0^{T\wedge\tau}\hat{Z}_u\,dB^\epsilon_u\Big] \,.
	\end{aligned}  
	\end{equation} 
	Using that $\tau\wedge T$ is a $\ast$-stopping time, by   \cite[Theorem 48]{hu2021extended}, we know
	$$\tilde{\mathbb{E}}^{b^\epsilon,d^{\epsilon,ij}}[{K}^1_{T\wedge\tau}]=\tilde{\mathbb{E}}^{b^\epsilon,d^{\epsilon,ij}}\Big[-\int_0^{T\wedge\tau}\hat{Z}_u\,dB^\epsilon_u\Big]=0\,.$$
	Since $u^1$ and $u^2$ have polynomial growth, we know 
	$|u^1(x)-u^2(x)|\leq L(1+|x|^p)$ for some  constants $p\in\mathbb{N}$ and $L>0.$ 
	By Lemma \ref{ergodic recurrent}, we have
	\begin{equation}\begin{aligned}
	\hat{Y}_0\leq \tilde{\mathbb{E}}^{b^\epsilon,d^{\epsilon,ij}}[|\hat{Y}_{T \wedge\tau}|]+L T\epsilon
	&\leq \delta +\tilde{\mathbb{E}}^{b^\epsilon,d^{\epsilon,ij}}[\mathds{1}_{\{\tau\geq T \}}|\hat{Y}_T|]+L T\epsilon\\
	&\leq \delta +\tilde{\mathbb{E}}^{b^\epsilon,d^{\epsilon,ij}}[\mathds{1}_{\{\tau\geq T \}}|u^1(X_T)-u^2(X_T)|]+L T\epsilon\\   
	&\leq \delta +L(\tilde{\mathbb{E}}^{b^\epsilon,d^{\epsilon,ij}}[\mathds{1}_{\{\tau\geq T \}}])^{\frac{1}{2}}(\tilde{\mathbb{E}}^{b^\epsilon,d^{\epsilon,ij}}[1+|X_T|^{2p}])^{\frac{1}{2}}+L T\epsilon\\
	&\leq \delta+L(T^{-\frac{1}{4}}(1+|x|^p)+ T\epsilon)\,.
	\end{aligned}\end{equation} 
	By sending $\epsilon\to 0$ and  $T\to\infty$,  we get $\hat{Y}_0\leq \delta$. 
	A similar argument gives $-\hat{Y}_0\leq \delta$ and thus $|Y_0^1-Y_0^2|=|\hat{Y}_0|\leq \delta$.
	Since $\delta>0$ is arbitrary, $u^1(x)=Y_0^1=Y_0^2=u^2(x)$.
	Because the quadratic variation of $\hat{Y}$ is zero, we obtain $Z^1=Z^2,$ and this induces $K^1=K^2$.
\end{proof}

Consider the PDE
\begin{equation}\begin{aligned}\label{ergodic PDE}
G(H(\lambda, x, D_xu, D_x^2u))+\langle b(x),D_xu\rangle+f(x,\langle\sigma^1(x),D_xu\rangle,\cdots,\langle\sigma^d(x),D_xu\rangle)+\gamma^1\lambda=0\,,
\end{aligned}\end{equation}
where
\begin{equation}\begin{aligned}
H_{ij}(\lambda, x, D_xu, D_x^2u)&=\langle D^2_xu\sigma^i(x),\sigma^j(x)\rangle+2\langle D_xu,h_{ij}(x)\rangle\\&\quad+2 g_{ij}(x,\langle\sigma^1(x),D_xu\rangle,\cdots,\langle\sigma^d(x),D_xu\rangle) +2\gamma_{ij}^2\lambda \,,
\end{aligned}\end{equation}
and  $\sigma^i$  is the $i$-th column vector of $\sigma$ for $i=1,2,\cdots,d.$
For a continuous 
function $u$ and a real number $\lambda,$ we say the pair $(u,\lambda)$ is a viscosity subsolution pair (respectively, viscosity supersolution pair)  to the PDE \eqref{ergodic PDE} if
\begin{equation}\begin{aligned}
&G(H( \lambda,\overline{x} ,D_x\overline{\varphi}(\overline{x}) ,D^2_x\overline{\varphi}(\overline{x})  )   )+\langle b(\overline{x}),D_x\overline{\varphi}(\overline{x})\rangle\\&+f(\overline{x},\langle\sigma^1,D_x\overline{\varphi}(\overline{x})\rangle,\cdots,\langle\sigma^d,D_x\overline{\varphi}(\overline{x})\rangle)+\gamma^1\lambda\geq 0\;\;(\text{respectively,}\leq 0)
\end{aligned}\end{equation}
for all $(\overline{\varphi},\overline{x})\in C^2(\mathbb{R}^m)\times \mathbb{R}^m$ satisfying
$\overline{\varphi}(\overline{x})=u(\overline{x})$ and $\overline{\varphi} -u\geq 0$ (respectively, $\overline{\varphi} -u\leq 0$).
If $(u,\lambda)$ 
is both a viscosity   subsolution pair and 
a viscosity   supersolution pair
to the PDE,
then
$(u,\lambda)$ is called a viscosity solution pair to the PDE.
The following theorem
can be proven by a 
similar argument in the proof of Theorem \ref{QBSDE feykac formula}.
\begin{thm}\label{ergo_pde}
	Let Assumptions \ref{ergodic assumption1} and  \ref{ergodic assumption3} hold. Then there exists a viscosity solution pair to the PDE \eqref{ergodic PDE}.  
\end{thm}

\begin{remark}
	\label{remark:infinite_BSDEs}	
Assumptions \ref{ergodic assumption1} and \ref{ergodic assumption2}   for infinite-horizon $G$-BSDEs differ from those in the existing literature in several key aspects.  While \cite{hu2018ergodic} investigates infinite-horizon $G$-BSDEs under the assumption that the drivers $f(x, y, z)$ and $g_{ij}(x,y,z)$ are Lipschitz continuous in $z$, our framework allows for drivers with quadratic growth in $z$.  In \cite{sun2024g}, the analysis 
	requires the boundedness of the mappings $x \mapsto f(x, 0, 0)$ and $x\mapsto g_{ij}(x,0,0)$, whereas our results do not rely on this condition.  
Assumption \ref{ergodic assumption3}  for ergodic $G$-BSDEs similarly  departs from the corresponding assumptions in both \cite{hu2018ergodic} and \cite{sun2024g}.
\end{remark}

\section{Robust pricing kernels}\label{section 6}

This section studies the long-term decomposition of  pricing kernels.
Consider the $G$-SDE and the pricing kernel of the form
\begin{align}
\label{SDE appli}X_s&=x+\int_0^sb(X_u)\,du+\int_0^sh_{ij}(X_u)\,d\langle B^i,B^j\rangle_u+\int_0^s\sigma(X_u)\,dB_u\,,\\
\label{SDF}
D_s&=e^{-\int_0^sr(X_u)\,du-\int_0^sk_{ij}(X_u)\,d\langle B^j,B^j\rangle_u-\int_0^sv(X_u)\,dB_u}\,
\end{align}
for $s\ge0$,
where $b,h_{ij}:\mathbb{R}^m\to\mathbb{R}^m$, $\sigma:\mathbb{R}^m\to\mathbb{R}^{m\times d}$, $r,k_{ij}:\mathbb{R}^m\to \mathbb{R}$ and $v:\mathbb{R}^m\to\mathbb{R}^d.$
The process $D$ can be understood as a robust version of pricing kernels under volatility uncertainty.
We define a function $d_{ij}:\mathbb{R}^m\to \mathbb{R}^m$ as $d_{ij}(x)=(\frac{1}{2}(\sigma_{\ell,i}(x)v_j(x)+\sigma_{\ell,j}(x)v_i(x)))_{1\leq \ell \leq m}$, where $\sigma_{\ell,i}$ is the 
entry in the $\ell$-th row and $i$-th column of the matrix $\sigma$ and $v_j$ is the $j$-th component of $v$.

\begin{assume}\label{sdf assumption}
	Assume the functions $b,h_{ij},\sigma,r,k_{ij},v,d_{ij}$ satisfy the following properties. 
	\begin{enumerate}[label=(\roman*)]%[font={\bfseries},label=({A\arabic*})] 
		\item\label{sdf symm} For $1\leq i,j \leq d$, $h_{ij}=h_{ji}$, $k_{ij}=k_{ji}$.
		\item \label{aa} There exist constants $C_1,C_\sigma,M_\sigma$ such that 
		\begin{equation} 
		\begin{aligned}
		&|b(x)-b(x')|+|r(x)-r(x')|+\sum_{i,j=1}^d|k_{ij}(x)-k_{ij}(x')|\\
        &+\frac{1}{2}\sum_{i,j=1}^d|v_i(x)v_j(x)-v_i(x')v_j(x')|\\
        &+\sum_{i,j=1}^d|h_{ij}(x)-h_{ij}(x')|+\sum_{i,j=1}^d|d_{ij}(x)-d_{ij}(x')|\leq C_1|x-x'|\,,\\
		&|\sigma(x)-\sigma(x')|\leq C_\sigma |x-x'|\,,\\
		&|\sigma(x)|\leq M_\sigma
		\end{aligned}
		\end{equation}
		for all $x,x'\in \mathbb{R}^m$.
		\item \label{eta appli} There exists a constant $\eta>0$ such that \begin{equation}\begin{aligned}
		&G\Big(\sum_{j=1}^m(\sigma_j(x)-\sigma_j(x'))^\top(\sigma_j(x)-\sigma_j(x'))\\
		&+2(\langle x-x',h_{ij}(x)-d_{ij}(x)-h_{ij}(x')+d_{ij}(x') \rangle)_{i,j=1}^d\Big)\\
		&+\langle x-x',b(x)-b(x')\rangle \leq -\eta|x-x'|^2
		\end{aligned}\end{equation} 
		for all $x,x'\in \mathbb{R}^m$, where $\sigma_j$ is the $j$-th row vector of $\sigma$ for $j=1,2,\cdots,m.$
		\item $\eta >\frac{1+\overline{\sigma}^2}{2}\Big(C_\sigma d+4\sqrt{2C\sigma  C_1 d \frac{\overline{\sigma}M_\sigma}{\underline{\sigma}}}\Big)\,.$
	\end{enumerate}
\end{assume}

\begin{lemma}\label{application lemma for expoential growth}
	Suppose \ref{aa} in Assumption \ref{sdf assumption} holds and $X$ is a solution to \eqref{SDE appli}. Then 
	\begin{equation}\begin{aligned}
	\hat{\mathbb{E}}[\sup_{0\leq s\leq T}e^{p |X_s|}]<\infty
	\end{aligned}\end{equation}
	for all $p\geq 1$ and $T\geq 0$. Moreover 
	$e^{|X|}\in \mathbb{S}^p_G(0,T)$\,.
\end{lemma}
\begin{proof}
	We first prove that $\hat{\mathbb{E}}[\sup_{0\leq s\leq T}e^{p |X_s|}]<\infty$ for all $p\geq 1$ and $T\ge0.$ In this proof, $L>0$ is a generic constant depending only on $\overline{\sigma},C_1, M_\sigma ,p ,x,T$ and may differ line by line.  Let $\mathcal{P}$ be a set that represents $\hat{\mathbb{E}}. $ For each $P\in \mathcal{P}$, the process $B$ is a continuous martingale and $\langle B\rangle$ is bounded, thus there exists a $d\times d$ matrix-valued process  $\alpha^P$ such that  $B=\int_0^\cdot \alpha_u^P\,dW^P_u $ and $\underline{\sigma}^2I_d\le (\alpha^P)^{\top}\alpha^P\leq \overline{\sigma}^2I_d$, where $W^P=(W^{P,i})_{1\le i\le d}$ is a $d$-dimensional   Brownian motion under $P$  and $I_d$ is the $d\times d$ identity matrix. Under each $P$,
	\begin{equation}\begin{aligned}
	\sup_{0\leq s\leq T}|X_s|\le|x|+LT+L\int_0^T\sup_{0\leq s\leq u}|X_s|\,du+\sum_{j=1}^m \sup_{0\leq s\leq T}\Big|\int_0^s\sigma_j(X_u)\alpha^P_u\,dW^P_u\Big|\,.
	\end{aligned}\end{equation}
	Using   Gr\"{o}nwall's inequality, 
	\begin{equation}\begin{aligned}
	\sup_{0\leq s\leq T}|X_s|\le \Big(|x|+LT+\sum_{j=1}^m \sup_{0\leq s\leq T}\Big|\int_0^s\sigma_j(X_u)\alpha^P_u\,dW^P_u\Big|\Big)e^{LT}\,.
	\end{aligned}\end{equation}
	Since,  for each $j=1,2,\cdots ,m$, the process $$M^j:=\int_0^\cdot \big(\sigma_j(X_u)\mathds{1}_{\{0\le u\le T\}}+e_1\mathds{1}_{\{u>T\}} \big)\alpha^P_u\,dW^P_u$$
	is a continuous martingale and $\langle M^j\rangle_{\infty}=\infty$, the Dambis-Dubins-Schwarz theorem 
	implies that the processes 
	$M^j$ and $W^{P,1}_{\langle M^j\rangle} $ have the same law,  where $e_1\in\mathbb{R}^d$ denotes the first standard basis vector.
	Since the quadratic variation $\langle M^j\rangle _T$ is less than or equal to $\overline{\sigma}^2M_\sigma^2  T,$     applying H\"{o}lder inequality, we have 
	\begin{equation}\begin{aligned}
	E^P\Big[\sup_{0\leq s \leq T}e^{p|X_s|}\Big]&\leq  L\prod_{j=1}^m \Big(E^P\Big[ \sup_{0\leq s\leq  \overline{\sigma}^2M_\sigma^2 T} e^{mpe^{LT}|W_s^{P,1}|}\Big]   \Big)^{\frac{1}{m}}\leq L_1\,
	\end{aligned}\end{equation} for some constant $L_1>0$.
	Thus, $\hat{\mathbb{E}}[\sup_{0\leq s\leq T}e^{p |X_s|}]<\infty$ because $L_1$ is independent of the measure $P.$ 
	Moreover, by \cite[Lemma 3.3]{hu2022quadratic}, we have   $e^{|X|}\in \mathbb{S}_G^p(0,T)$.
\end{proof}

We can define a  sublinear expectation and sublinear conditional expectation on the $G$-expectation space $(\Omega,Lip(\Omega),\hat{\mathbb{E}})$ 
like 
\eqref{M_z}. Consider the process  
$$M^v=e^{-\int_0^\cdot\frac{1}{2} v_i(X_u)v_j(X_u)\,d\langle B^i,B^j\rangle_u-\int_0^\cdot v(X_u)\,dB_u}\,.$$ 
By \ref{aa} in Assumption \ref{sdf assumption} and 
Lemma \ref{application lemma for expoential growth}, 
we obtain that for any $p\geq 1$ 
\begin{equation}\begin{aligned}
\hat{\mathbb{E}}[e^{p\int_0^Tv_i(X_u)v_j(X_u)\,d\langle B^i, B^j\rangle_u}]<\infty\,,
\end{aligned}\end{equation}
which implies that
$M^v$ is a symmetric $G$-martingale
by \cite[Proposition 5.10]{osuka2013girsanov}. Define a sublinear expectation $\hat{\mathbb{E}}^v$ and sublinear conditional expectation $\hat{\mathbb{E}}^v_s$ as 
\begin{equation}\begin{aligned}\label{new sublinear expectation appli}
\hat{\mathbb{E}}^v[\xi]=\hat{\mathbb{E}}[M^v_T \xi]\,,\,\,\,\hat{\mathbb{E}}^v_s[\xi]=(M^v_s)^{-1}\hat{\mathbb{E}}_s[M^v_T \xi] 
\end{aligned}\end{equation}
for $\xi\in Lip(\Omega_T)$. We say $\hat{\mathbb{E}}^v$(respectively, $\hat{\mathbb{E}}^v_s$) is the sublinear expectation (respectively, the sublinear conditional expectation) induced by $v(X)$. The process
\begin{equation}\begin{aligned}
B^v=B+\Big(\sum_{j=1}^d\int_0^\cdot v_j(X_u)\,d\langle B^i,B^j\rangle_u\Big)_{1\leq i\leq d}
\end{aligned}\end{equation}
is a $G$-Brownian motion under the sublinear expectation $\hat{\mathbb{E}}^v$. We can construct the spaces $\mathbb{L}^{p,v}_G(\Omega_T),$ $ \mathbb{M}^{p,v}_G(0,T),$ $\mathbb{H}^{p,v}_G(0,T),$ $\mathbb{S}^{p,v}_G(0,T)$ for each $p\geq 1$ in a way analogous to  Section \ref{section 2}. If there is no confusion, we will omit the superscript $v$. Refer to \cite{osuka2013girsanov} for more details.

\begin{lemma}\label{K is MG under v}
	Suppose \ref{aa} in Assumption \ref{sdf assumption} holds. Let $K$ be a decreasing $G$-martingale on the $G$-expectation space $(\Omega,\mathbb{L}^1_G(\Omega),\hat{\mathbb{E}})$ with $K_0=0$ and $K_T\in \mathbb{L}^p(\Omega_T)$ for some $p>1$, and let $\hat{\mathbb{E}}^v$ be the sublinear expectation induced by $v(X)$. Then $K$ is a decreasing $G$-martingale under the sublinear expectation $\hat{\mathbb{E}}^v.$ 
\end{lemma}

\begin{proof}
	 From Lemma \ref{application lemma for expoential growth}, observe that for any $p>1$
	\begin{equation}\begin{aligned}
	&\;\quad\hat{\mathbb{E}}_s[(e^{-\int_s^T\frac{1}{2}v_i(X_u)v_j(X_u)\,d\langle B^i,B^j\rangle_u-\int_s^Tv(X_u)\,dB_u})^p]\\
	&=\hat{\mathbb{E}}_s[e^{-\int_s^T \frac{p^2q}{2}v_i(X_u)v_j(X_u)\,d\langle B^i,B^j\rangle_u-\int_s^Tpv(X_u)\,dB_u} e^{\int_s^T \frac{p(pq-1)}{2}v_i(X_u)v_j(X_u)\,d\langle B^i,B^j\rangle_u  }]\\
	&\leq \hat{\mathbb{E}}_s[e^{-\int_s^T\frac{p^2q^2}{2}v_i(X_u)v_j(X_u)\,d\langle B^i,B^j\rangle_u-\int_s^Tpqv(X_u)\,dB_u}]^{\frac{1}{q}}\\
    &\quad\cdot \hat{\mathbb{E}}_s[ e^{\int_s^T\frac{pq(pq-1)}{2(q-1)} v_i(X_u)v_j(X_u)\,d\langle B^i,B^j\rangle_u  }]^{\frac{q-1}{q}}<\infty\,, 
	\end{aligned}\end{equation}
	where $q>0$ is a constant satisfying $\frac{1}{p}+\frac{1}{q}<1$.  
	By combing this inequality with the proof of
	\cite[Lemma 3.4]{hu2018quadratic}, we obtain that 
	$K$ is a decreasing $G$-martingale under the sublinear expectation $\hat{\mathbb{E}}^v.$ 
\end{proof}

The following theorem is the main result of this study. 
The pricing kernel is decomposed into a discounting component with a long-term exponential rate $\lambda$, a transitory component $e^{u(X_0)-u(X)}$, a positive symmetric $G$-martingale $M$, and a decreasing process $e^{K}$ that captures the volatility uncertainty of the $G$-Brownian motion. Furthermore, the components $M$ and $e^K$ are represented in terms of $u$ when $u$ is twice continuously  differentiable. 
It is noteworthy that 
the function $u$ 	
is obtained from the ergodic $G$-BSDE \eqref{eqn:appli egbsde}.

\begin{thm} \label{sdf thm}
	Under Assumption \ref{sdf assumption}, the pricing kernel in \eqref{SDF}
	has the long-term  decomposition 
	\begin{equation}
	\label{fac}
	D_s=e^{\lambda s}e^{u(X_0)-u(X_s)}
	M_se^{K_s}\,,\;s\ge0
	\end{equation}
	satisfying
	\begin{enumerate}[label=(\roman*)] 	
		\item\label{sdf lambda} $\lim_{s\to\infty} \frac{1}{s}\ln\hat{\mathbb{E}}[{D_s}]=\lambda,$
		\item\label{sdf Z}  $(M_s)_{s\ge0}$ is a positive symmetric  $G$-martingale of  the form  $$M_s=e^{-\int_0^s\frac{1}{2}(Z^i_u-v_i(X_u))(Z^j_u-v_j(X_u))\,d\langle B^i,B^j\rangle_u+\int_0^s Z_u-v(X_u)\,dB_u}$$ for some bounded process  $Z=(Z^1,\cdots,Z^d)$,
		\item\label{sdf growth u} $u:\mathbb{R}^m\to\mathbb{R}$   is a continuous function with   polynomial growth,		
		\item  $u$  is a viscosity solution to  the PDE
		\begin{equation}\begin{aligned}\label{sdf pde}
		G(H(x,D_xu,D_x^2u))+\langle b(x),D_xu\rangle-r(x)-\lambda=0\,,
		\end{aligned}\end{equation}
		where 
		 \begin{equation}\begin{aligned}
		H_{ij}(x,D_xu,D_x^2u)&=\langle D^2_xu\sigma^i(x),\sigma^j(x)\rangle+2\langle D_xu,h_{ij}(x)-d_{ij}(x)\rangle-2k_{ij}(x)\\
        &\quad+v_i(x)v_j(x)+\langle\sigma^i(x),D_xu\rangle\langle\sigma^j(x),D_xu\rangle\,,
		\end{aligned}\end{equation}
		\item\label{sdf K} $(K_s)_{s\ge0}$ is a decreasing $G$-martingale  and $K_s\in \mathbb{L}^2_G(\Omega_s)$,
		\item\label{MeK} $(M_se^{K_s})_{s\ge0}$ is a $G$-martingale.
	\end{enumerate}
	In addition, if $u$ is twice continuously  differentiable, then 
	$$Z_s=(\langle\sigma^j(X_s),D_xu(X_s)\rangle)_{1\leq j\leq d}$$
	and
	 \begin{align}
	\label{K}
	K_s&=\int_0^s\frac{1}{2}H_{ij}(X_u,D_xu(X_u) ,D_x^2u(X_u)  )\,d\langle B^i,B^j\rangle_u\\
    &\quad-\int_0^sG(H(  X_u,D_xu(X_u) ,D_x^2u(X_u) )  )\,du\,.
	\end{align} 
\end{thm}

\begin{proof}
	In this proof, $L>0$ is a generic constant depending only on $\overline{\sigma},C_1,C_{\sigma},\eta$ and may differ  line by line.
	Let $\hat{\mathbb{E}}^v$ be the sublinear expectation induced by $v(X)$. Then $B^v=B+(\sum_{j=1}^d\int_0^\cdot v_j(X_u)\,d\langle B^i,B^j\rangle_u)_{1\leq i\leq d}$ is a $G$-Brownian motion under the sublinear expectation $\hat{\mathbb{E}}^v$.
	It is evident that  $X$ satisfies the $G$-SDE
	\begin{equation}\begin{aligned}\label{eqn:X on v}
	X_s=x+\int_0^sb(X_u)\,du+\int_0^s (h_{ij}(X_u)-d_{ij}(X_u))\,d\langle B^{v,i},B^{v,j}\rangle_u+\int_0^s\sigma(X_u)\,dB^v_u\,.
	\end{aligned}\end{equation}       
	For any $T>0,$ consider the ergodic $G$-BSDE 
	\begin{equation}\begin{aligned}
	\label{eqn:appli egbsde}
	Y_s&=Y_T+\int_s^T(-r(X_u)-\lambda)\,du+\int_s^T\Big(-k_{ij}(X_u)+\frac{1}{2}v_i(X_u)v_j(X_u)\Big)\,d\langle B^{v,i},B^{v,j}\rangle_u\\
    &\quad+\int_s^T\frac{1}{2}Z^i_uZ^j_u\,d\langle B^{v,i},B^{v,j}\rangle_u-\int_s^TZ_u\,dB^v_u-(K_T-K_s)\,. 
	\end{aligned}\end{equation}
	By Theorem \ref{QBSDE with lambda sol exist}, there is a solution  
	$(Y,Z,K,\lambda)$ to \eqref{eqn:appli egbsde} such that
	$Z$ is bounded by $L$ and $Y$ is bounded by $L(1+|X|)$.
	Define a function $u$  as $u(x)=Y_0^x$. 
	Using $Y_s=u(X_s)$ and \eqref{eqn:appli egbsde}, 
	we have
	\begin{equation}\begin{aligned}\label{eqn:decompose}
	D_s&=e^{-\int_0^sr(X_u)\,du-\int_0^sk_{ij}(X_u)\,d\langle B^i,B^j\rangle_u-\int_0^sv(X_u)\,dB_u}\\
	&=e^{\lambda s+u(X_0)-u(X_s)-\int_0^s\frac{1}{2}(v_i(X_u)v_j(X_u)+Z_u^iZ_u^j)\,d\langle B^{v,i},B^{v,j}\rangle_u+\int_0^s Z_u\,dB^v_u-\int_0^s v(X_u)\,dB_u+K_s}\\
	&=e^{\lambda s+u(X_0) -u(X_s)-\int_0^s\frac{1}{2}(Z^i_u-v_i(X_u))(Z^j_u-v_j(X_u))\,d\langle B^i,B^j\rangle_u+\int_0^s Z_u-v(X_u)\,dB_u+K_s}\,.
	\end{aligned}\end{equation}
	Define $
	M_s=e^{-\int_0^s\frac{1}{2}(Z^i_u-v_i(X_u))(Z^j_u-v_j(X_u))\,d\langle B^i,B^j\rangle_u+\int_0^s Z_u-v(X_u)\,dB_u}.$
	By Lemma \ref{application lemma for expoential growth} and \cite[Proposition 5.10]{osuka2013girsanov}, $M$ is a symmetric $G$-martingale under $\hat{\mathbb{E}}$.
	Moreover, by Lemma \ref{K is MG under v}, $K$ is a decreasing $G$-martingale under $\hat{\mathbb{E}}$.  Since $u$ is constructed in the proof of Theorem \ref{QBSDE with lambda sol exist} as a continuous function with polynomial growth and is a viscosity solution to \eqref{sdf pde}, the composition in \eqref{fac} satisfies 
	\ref{sdf Z}-\ref{sdf K}. Moreover, \ref{MeK} is easily derived from It\^o's formula and Lemma \ref{int with respect to K is martingale}.

    The next step is to prove  \ref{sdf lambda}.
    Theorem \ref{Finite QBSDE exis and unique} ensures the existence of a solution $(\overline{Y},\overline{Z},\overline{K})$  to the $G$-BSDE
	\begin{equation}\begin{aligned}
	\label{eqn:appli fintie gbsde}Y_s&=-\int_s^Tr(X_u)\,du+\int_s^T\Big(-k_{ij}(X_u)+\frac{1}{2}v_i(X_u)v_j(X_u)+\frac{1}{2}Z_u^iZ^j_u\Big)\,d\langle B^{v,i},B^{v,j}\rangle_u\\&
    \quad-\int_s^TZ_u\,dB^v_u-(K_T-K_s)\,,\;0\le s\le T
	\end{aligned}\end{equation}
	such that
	$\overline{Z}$ is bounded.  
    Using a similar argument in the proof of Theorem \ref{Finite QBSDE exis and unique}, for each $\epsilon>0,$ we can construct two processes $  d^{\epsilon,ij},n^{\epsilon,ij}$, a sublinear expectation $\hat{\mathbb{E}}^{v,d^{\epsilon,ij}}$ and a $G$-Brownian motion $B^{v,\epsilon}$ such that 
	\begin{equation}\begin{aligned}
	Y_0-\overline{Y}_0+\lambda T
	&\leq \hat{\mathbb{E}}^{v,d^{\epsilon,ij}}[|Y_T|]+\hat{\mathbb{E}}^{v,d^{\epsilon,ij}}\Big[\int_0^T n_u^{\epsilon,ij}\,d\langle B^{v,\epsilon,i},B^{v,\epsilon,j}\rangle_u\Big]\\
	&\leq L(1+\hat{\mathbb{E}}^{v,d^{\epsilon,ij}}[|X_T|])+\hat{\mathbb{E}}^{v,d^{\epsilon,ij}}\Big[\int_0^T n_u^{\epsilon,ij}\,d\langle B^{v,\epsilon,i},B^{v,\epsilon,j}\rangle_u\Big] \,.
	\end{aligned}\end{equation}
	By letting $\epsilon\to 0$ and applying Lemma \ref{Ergodic mean is finite}, we obtain $Y_0-\overline{Y}_0+\lambda T\leq L(1+|x|)$. Similarly we can obtain $-Y_0+\overline{Y}_0-\lambda T\leq L(1+|x|)$, and thus $|-Y_0+\overline{Y}_0-\lambda T|\leq L(1+|x|)$. Observe from $|Y_0|\leq L(1+|x|)$ that 
	\begin{equation}\begin{aligned}
	\Big|\frac{\overline{Y}_0}{T}-\lambda\Big|\leq \frac{L(1+|x|)}{T}\,.\label{eqn: i}
	\end{aligned}\end{equation}
	
	From  It\^{o}'s formula, we obtain 
	 \begin{equation}\begin{aligned}
	e^{\overline{Y}_0}&=e^{-\int_0^Tr(X_u)\,du-\int_0^T k_{ij}(X_u)-\frac{1}{2}v_i(X_u)v_j(X_u)\,d\langle B^{v,i},B^{v,j}\rangle_u}\\
    &\quad-\int_0^T e^{-\int_0^ur(X_l)\,dl-\int_0^uk_{ij}(X_l)-\frac{1}{2}v_i(X_l)v_j(X_l)\,d\langle B^{v,i},B^{v,j}\rangle_l}e^{\overline{Y}_u}\overline{Z}_u\,dB^v_u\\&
    \quad-\int_0^T e^{-\int_0^ur(X_l)\,dl-\int_0^uk_{ij}(X_l)-\frac{1}{2}v_i(X_l)v_j(X_l)\,d\langle B^{v,i},B^{v,j}\rangle_l}e^{\overline{Y}_u}\,d\overline{K}_u\,.
	\end{aligned}\end{equation}
	Applying Lemma \ref{application lemma for expoential growth} to  the solution $X$ of \eqref{eqn:X on v},  we know that $$e^{-\int_0^Tr(X_u)\,du-\int_0^Tk_{ij}(X_u)-\frac{1}{2}v_i(X_u)v_j(X_u)\,d\langle B^{v,i},B^{v,j}\rangle_u}e^{\overline{Y}_T}\in \mathbb{S}^p_G(0,T)$$ for every $p\geq 1$, which implies that 
	$$\int_0^\cdot e^{-\int_0^ur(X_l)\,dl-\int_0^uk_{ij}(X_l)-\frac{1}{2}v_i(X_l)v_j(X_l)\,d\langle B^{v,i},B^{v,j}\rangle_l}e^{\overline{Y}_u}\,d\overline{K}_u$$
	is a decreasing $G$-martingale by Lemma \ref{int with respect to K is martingale}.
	Thus, $$e^{-\int_0^\cdot r(X_u)\,du-\int_0^\cdot k_{ij}(X_u)-\frac{1}{2}v_i(X_u)v_j(X_u)\,d\langle B^{v,i},B^{v,j}\rangle_u}-e^{\overline{Y}_0}$$ is a $G$-martingale, and we have \begin{equation}\begin{aligned}
	\hat{\mathbb{E}}[D_T]&=\hat{\mathbb{E}}[e^{-\int_0^T r(X_u)\,du-\int_0^T k_{ij}(X_u)\,d\langle B^i,B^j\rangle_u-\int_0^T v(X_u)\,dB_u}]\\
	&=\hat{\mathbb{E}}^{v}[e^{-\int_0^Tr(X_u)\,du-\int_0^Tk_{ij}(X_u)-\frac{1}{2}v_i(X_u)v_j(X_u)\,d\langle B^{v,i},B^{v,j}\rangle_u}]=   e^{\overline{Y}_0}\,.
	\end{aligned}\end{equation}
	Combining this with \eqref{eqn: i}, we obtain $\lim_{T\to\infty} \frac{1}{T}\ln\hat{\mathbb{E}}[{D_T}]=\lambda.$
	This gives 
	\ref{sdf lambda} by replacing  $T$ with $s.$

	Suppose $u$ is twice continuously  differentiable.   According to   It\^{o}'s formula and \eqref{eqn:appli egbsde}, we have
	\begin{equation}\begin{aligned}
	u(X_0)&=u(X_s)-\int_0^s \langle D_xu(X_u),b(X_u)\rangle \,du\\
    &\quad-\int_0^s\langle D_xu(X_u),h_{ij}(X_u)-d_{ij}(X_u)\rangle \,d\langle B^{v,i},B^{v,j}\rangle_u\\
    &\quad-\int_0^s \frac{1}{2} \langle D^2_xu (X_u)\sigma^i(X_u),\sigma^j(X_u)\rangle\,d\langle B^{v,i},B^{v,j}\rangle_u-\int_0^s D_xu^{\top}(X_u)\sigma(X_u)\,dB^v_u 
	\end{aligned}\end{equation}
	and
	\begin{equation}\begin{aligned}
	u(X_0) 
	&=u(X_s)+\int_0^s-r(X_u)-\lambda\,du
	\\
    &\quad+\int_0^s-k_{ij}(X_u)+\frac{1}{2}(v_i(X_u)v_j(X_u)+Z^i_uZ^j_u)\,d\langle B^{v,i},B^{v,j}\rangle_u-\int_0^s Z_u\,dB^v_u-K_s\,.
	\end{aligned}\end{equation}    
	Since $u(X)$ is a special semimartingale for each $P\in\mathcal{P}^v$, comparing these two equations, we know $Z_s=D_xu^{\top}(X_s)\sigma(X_s)$   and
	\begin{equation} 
	\begin{aligned}
	K_s&=\int_0^s \langle D_xu(X_u),b(X_u)\rangle-r(X_u)-\lambda\,du\\
	&\quad+\int_0^s\frac{1}{2}\langle D^2_xu(X_u)\sigma^i(X_u),\sigma^j(X_u)\rangle+\langle D_xu(X_u),h_{ij}(X_u)-d_{ij}(X_u)\rangle\,d\langle B^{v,i},B^{v,j}\rangle_u\\
    &\quad+\int_0^s-k_{ij}(X_u)+\frac{1}{2}(v_i(X_u)v_j(X_u)+Z^i_uZ^j_u) \,d\langle B^{v,i},B^{v,j}\rangle_u\,,
	\end{aligned}    
	\end{equation}
	where $\mathcal{P}^v$ is a set that represents $\hat{\mathbb{E}}^v$.  Because $u$ is a classical solution to \eqref{sdf pde}, we obtain \eqref{K}.   
	This completes the proof. 
\end{proof}

\begin{thm}\label{thm:unique}
	Under Assumptions \ref{Assumption ergodic unique} and \ref{sdf assumption}, the
	decomposition  of the pricing kernel in \eqref{fac}  is unique. More precisely, if 
	\begin{equation}
	\label{unique}
	D_s=e^{\lambda s}e^{u(X_0)-u(X_s)}
	M_se^{K_s}=e^{\overline{\lambda} s}e^{\overline{u}(X_0)-\overline{u}(X_s)}
	\overline{M}_se^{\overline{K}_s}\,,\;s\ge0
	\end{equation}
	and if  $(u,M,K,\lambda)$ and $(\overline{u},\overline{M},\overline{K},\overline{\lambda})$
	satisfy \ref{sdf Z}, \ref{sdf growth u}, \ref{sdf K} in Theorem \ref{sdf thm}, then 
	$\lambda=\overline{\lambda},$  $M=\overline{M},$ $K=\overline{K},$ $u-u(X_0)=\overline{u}-\overline{u}(X_0).$ 
\end{thm}

\begin{proof}  
    Let $Z$ and $\overline{Z}$ be  the processes  satisfying \ref{sdf Z} in Theorem \ref{sdf thm}, associated with $M$ and $\overline{M}$, respectively.
	A simple calculation gives that both $(u(X),Z,K,\lambda)$
	and
	$(\overline{u}(X),\overline{Z},\overline{K},\overline{\lambda})$  
	are solutions to \eqref{eqn:appli egbsde}. By Theorem \ref{markovian ergodic solution is unique}, we obtain the uniqueness. \end{proof}

The following theorem provides an alternative characterization of the long-term decomposition  in terms of solutions to PDEs.
\begin{thm}\label{feynmac kac formula reverse} 
	Under \ref{sdf symm}, \ref{aa}, \ref{eta appli} in Assumption \ref{sdf assumption}, suppose that 
	the PDE \eqref{sdf pde}
	has a solution pair $(u,\lambda)\in C^2(\mathbb{R}^m)\times\mathbb{R}$ 
	such that $D_xu$ is bounded and for all $T>0,$
	the PDE
	\begin{equation}\label{eqn:PDE}
	\begin{aligned}
	\begin{cases}\; 
	\partial_t w+G(H( x,D_xw,D_x^2w) )+\langle b(x),D_xw\rangle-r(x)=0\,,  \\
	\; w(T,x)=0\,,
	\end{cases}
	\end{aligned}\end{equation}
	has a solution 
	$w\in C^{1,2}([0,T]\times\mathbb{R}^m)$ 
	such that
	$D_xw$ is bounded, 	where 
	\begin{equation}\begin{aligned}
	H_{ij}(x,D_xw,D_x^2w)&=\langle D^2_xw\sigma^i(x),\sigma^j(x)\rangle+2\langle D_xw,h_{ij}(x)-d_{ij}(x)\rangle-2k_{ij}(x)\\&\quad+v_i(x)v_j(x)+\langle\sigma^i(x),D_xw\rangle\langle\sigma^j(x),D_xw\rangle\,.
	\end{aligned}\end{equation}
	Define
	\begin{equation} 
	\begin{aligned}
	M_s&=e^{-\int_0^s\frac{1}{2} (  
 \langle  \sigma^i(X_u) , D_xu(X_u)   \rangle -v_i(X_u) ) (  \langle \sigma^j(X_u) , D_xu(X_u) \rangle- v_j(X_u) ) \,d\langle B^i,B^j\rangle_u}\\&\quad \cdot e^{\int_0^s \langle \sigma^j(X_u) , D_xu(X_u) \rangle-v_j(X_u) \,dB^j_u   }
	\end{aligned}	\end{equation}  
	and
	\begin{equation}\begin{aligned}
	K_s&=\int_0^s\frac{1}{2}H_{ij}(   X_u,D_xu(X_u),D_x^2u(X_u)     )\,d\langle B^i,B^j\rangle_u\\
    &\quad-\int_0^sG(H(  X_u,D_xu(X_u),D_x^2u(X_u)    ))\,du
	\end{aligned}\end{equation} for $s\ge0.$
	Then we have 
	\begin{equation}
	D_s=e^{\lambda s}e^{u(X_0)-u(X_s)}
	M_se^{K_s}\,,\;s\ge0
	\end{equation}
	and
	\begin{enumerate}[label=(\roman*)] 
		\item \label{PDE_thm_lambda} $\lim_{s\to\infty} \frac{1}{s}\ln\hat{\mathbb{E}}[{D_s}]=\lambda,$
		\item $M$ is a positive symmetric $G$-martingale,
		\item  $K$ is a decreasing $G$-martingale.
	\end{enumerate} 
\end{thm}
\begin{proof}
 Let $H_{ij}(X_u):=H_{ij}(X_u,D_xu(X_u),D_x^2u(X_u) )$ for simplicity. Since $D_xu$ is bounded and $|v_i(x)v_j(x)|\leq C_1(1+|x|)$, by Lemma \ref{application lemma for expoential growth} and \cite[Proposition 5.10]{osuka2013girsanov}, $M$ is a positive symmetric $G$-martingale.
	
	Now we show that $K$ is a decreasing $G$-martingale. 
	For each $k\in\mathbb{N}$ and $T>0$, define   
	\begin{equation}\begin{aligned}
	\tau_k:=\inf\{s\geq 0| \, |X_s|>k\}\wedge T\,.
	\end{aligned}\end{equation}
	Then $\tau_k$ is a $\ast$-stopping time by \cite[Example 46]{hu2021extended}.
	It follows that
	\begin{equation}\begin{aligned}
	K_{s\wedge \tau_k}=\int_0^{s\wedge\tau_k} \frac{1}{2}H_{ij}(X_u)\,d\langle B^i,B^j\rangle_u-\int_0^{s\wedge\tau_k}G(H_{ij}(X_u))\,du\,.
	\end{aligned}\end{equation}
The process
	\begin{equation}\begin{aligned}
	\int_0^\cdot \frac{1}{2}H_{ij}\Big(\frac{|X_u|\wedge k}{|X_u|}X_u\Big)\,d\langle B^i,B^j\rangle_u-\int_0^\cdot G\Big(H_{ij}\Big(\frac{|X_u|\wedge k}{|X_u|}X_u\Big)\Big)\, du
	\end{aligned}\end{equation}
	is a decreasing $G$-martingale and coincides with $K$ on $[0,T\wedge \tau_k]$. 
	We have used 
  the convention $\frac{\,0\,}{0}=0$.
	Then $\hat{\mathbb{E}}_t[K_{s\wedge\tau_k}]=K_{t\wedge\tau_k}$ by \cite[Theorem 48]{hu2021extended}, which implies that $K_{\cdot\wedge\tau_k}$ is  a decreasing $G$-martingale.

	Let $\hat{\mathbb{E}}^v$ be the sublinear expectation induced by $v(X),$ then $B^v=B+(\sum_{j=1}^d\int_0^\cdot v_j(X_u)\,d\langle B^i,B^j\rangle_u)_{1\leq i\leq d}$ is a $G$-Brownian motion under the sublinear expectation $\hat{\mathbb{E}}^v$. Using It\^{o}'s formula, \eqref{sdf pde} and the fact 
	that the quadratic variations of $B$ and $B^v$ coincide under both $\hat{\mathbb{E}}$ and $\hat{\mathbb{E}}^v$, we have
	\begin{equation}\begin{aligned}\label{eqn: appli ergodic}
	u(X_s)&=u(X_{T\wedge\tau_k})-\int_s^{T\wedge\tau_k} \langle D_xu(X_u),b(X_u)\rangle \,du\\
    &\quad-\int_s^{T\wedge\tau_k}\langle D_xu(X_u),h_{ij}(X_u)-d_{ij}(X_u)\rangle\, d\langle B^{v,i},B^{v,j}\rangle_u\\
    &\quad-\int_s^{T\wedge\tau_k}\frac{1}{2} \langle D^2_x u(X_u)\sigma^i(X_u),\sigma^j(X_u)\rangle \,d\langle B^{v,i},B^{v,j}\rangle_u\\
    &\quad    -\int_s^{T\wedge\tau_k} D_xu^{\top}(X_u)\sigma(X_u)\,dB^v_u\\ 
	&=u(X_{T\wedge\tau_k})+\int_s^{T\wedge\tau_k}-r(X_u)-\lambda\,du
	\\
    &\quad+\int_s^{T\wedge\tau_k}\big(-k_{ij}(X_u)+\frac{1}{2}v_i(X_u)v_j(X_u)+\frac{1}{2}Z^i_uZ^j_u \big)\,d\langle B^{v,i},B^{v,j}\rangle_u\\
    &\quad-\int_s^{T\wedge\tau_k} Z_u\,dB^v_u-(K_{T\wedge\tau_k}-K_s)\,.
	\end{aligned}\end{equation} 
	Since $u,r,b,k_{ij},h_{ij},d_{ij},v_iv_j$ has linear growth and $D_xu$ is bounded, it follows that
	\begin{equation}\begin{aligned}\label{eqn:K Lp}
	|K_{T\wedge\tau_k}|&\leq L(1+\sup_{0\leq s\leq T}|X_s|)+\sup_{0\leq s\leq T}\Big|\int_0^sD_xu^{\top}(X_u)\sigma(X_u)\,dB^v_u\Big|	\,,
	\end{aligned}\end{equation}
	where $L>0$ is a   constant  depending only on $C_1,\lambda, M_\sigma, \overline{\sigma},T,x$.
	Combining this with Lemma \ref{Ergodic mean is finite} and Proposition \ref{in sublinear DCT}, we have 
	$\lim_{k\to\infty}\hat{\mathbb{E}}^v[|K_T-K_{T\wedge\tau_k}|]=0.$
	Thus, $K$ is a decreasing $G$-martingale under $\hat{\mathbb{E}}^v$. Using $\hat{\mathbb{E}}[K_T^p]<\infty$ for $p\geq 1$ and Lemma \ref{K is MG under v}, we obtain that $K$ is a decreasing $G$-martingale under $\hat{\mathbb{E}}$.

    The next step is to prove $\lim_{s\to\infty} \frac{1}{s}\ln\hat{\mathbb{E}}[{D_s}]=\lambda$.  
    Since	  
    $K_T$ is  quasi-continuous  and $\hat{\mathbb{E}}^v[|K_T|^p]<\infty$ for all $p\geq 1$ and $T>0,$ 	by \cite[Theorem 25]{denis2011function}, we have  $K_T\in \mathbb{L}^{2,v}_G(\Omega_T)$ 
    for all $T>0$. 
	Using this and  \eqref{eqn: appli ergodic}, it can be shown that  $(u(X), D_xu^{\top}(X)\sigma(X),K,\lambda)$ is a solution to \eqref{eqn:appli egbsde}. Define 
    \begin{align}
        \overline{K}&=\int_0^\cdot \frac{1}{2}H_{ij}(     X_u,D_xw(u,X_u), D_x^2w(u,X_u)    )\,d\langle B^{v,i},B^{v,j}\rangle_u\\
        &\quad-\int_0^\cdot G(H_{ij}(X_u,D_xw(u,X_u), D_x^2w(u,X_u))  )\,du\,.
    \end{align} Then one can easily prove that $(w(\,\cdot,X),D_xw^{\top}(\,\cdot,X)\sigma(X),\overline{K})$ is a solution to  \begin{equation}\begin{aligned}
	Y_s&=\int_s^T-r(X_u)\,du+\int_s^T \big( -k_{ij}(X_u)+\frac{1}{2}v_i(X_u)v_j(X_u)+\frac{1}{2}Z_u^iZ^j_u \big)\,d\langle B^{v,i},B^{v,j}\rangle_u\\&
    \quad-\int_s^TZ_u\,dB^v_u-(K_T-K_s)\,.
	\end{aligned}\end{equation}
	By applying the same argument in the proof of Theorem \ref{sdf thm} to $(u(X), D_xu^{\top}(X)\sigma(X),K,\lambda)$ and $(w(\,\cdot,X),D_xw^{\top}(\,\cdot,X)\sigma(X),\overline{K})$, we obtain $\lim_{T\to\infty} \frac{1}{T}\ln\hat{\mathbb{E}}[{D_T}]=\lambda$. This gives the desired result by replacing  $T$ with $s.$
\end{proof}

We now provide sufficient conditions under which the PDEs mentioned in Theorem \ref{feynmac kac formula reverse} admit classical solutions.
The following lemma is an immediate consequence of \cite[Theorems 6.1.5, 6.4.3, and 6.4.5]{krylov1987nonlinear}, and its proof is therefore omitted.

\begin{lemma}\label{thm:kryl}
Let $T,r>0,$ $a\in \mathbb{R}^d$ and let $\Sigma$ be an index set. For  $A\in \Sigma$ and $n\in \mathbb{N}$, consider a smooth function $F^{A,n}: (0,T)\times B_r(a)\times \mathbb{R}\times\mathbb{R}^d\times\mathbb{S}^d \to \mathbb{R}$ where $B_r(a)$ is the open ball centered at $a$ with radius $r$. 
Suppose  the following conditions are satisfied.
\begin{enumerate}[label=(\roman*)]
    \item \label{eqn:clas 1} There exist  positive constants $\epsilon,K,L$ % independent of $A\in \Sigma$ and $n\in \mathbb{N}$ 
    such that
    \begin{align}
        &\epsilon|\zeta|^2\le \sum_{i,j=1}^d\partial_{\varphi_{ij}} F^{A,n}\zeta_i\zeta_j\le K|\zeta|^2\,,\\   
        &|F^{A,n}-\sum_{i,j=1}^d\partial_{\varphi_{ij}} F^{A,n}\varphi_{ij}|\le L(1+|\varphi|)\Big(1+\sum_{i=1}^d|\varphi_i|^2\Big)\,,\\
	&|(\partial_{\varphi_i}F^{A,n})_{1\le i\le d}|\Big(1+\sum_{i=1}^d|\varphi_i|\Big)+|\partial_\varphi F^{A,n}|+|(\partial_{x_i}F^{A,n})_{1\le i\le d}|\Big(1+\sum_{i=1}^d|\varphi_i|\Big)^{-1}\\
	&\;\;\quad\quad\quad\le L(1+|\varphi|)\Big(1+\sum_{i=1}^d|\varphi_i|^2+\sum_{i,j=1}^d|\varphi_{ij}|\Big)\,,\\
    &|\partial_t F^{A,n}|\le L(1+|\varphi|+|\varphi_i|^2)
    \Big(1+\sum_{i,j=1}^d|\varphi_{ij}|^2\Big)\,,
    \end{align}
    for all $A\in \Sigma$, $n\in \mathbb{N}$,  $\zeta=(\zeta_1,\cdots \zeta_d)\in \mathbb{R}^d$ and $(t,x,\varphi,\varphi_i,\varphi_{ij})\in(0,T)\times B_r(a)\times \mathbb{R}\times\mathbb{R}^d\times\mathbb{S}^d$.
    \item \label{eqn:clas 2} There exists a positive constant $\tilde{L}$ %independent of $A\in \Sigma$ and $n\in \mathbb{N}$ 
    such that
    \begin{equation} 
    \begin{aligned}
        \Big(\tilde{L}\Big(1+|\varphi|+\sum_{i=1}^d|\varphi_i|^2\Big)\Big)^{-1}F^{A,n}_{(\eta)(\eta)}\le& \sum_{i,j=1}^d|\tilde{\varphi}_{ij}|\Big(\sum_{i=1}^d|\tilde{\varphi}_i|+\Big( 1+\sum_{i,j=1}^d|\varphi_{ij}|\Big)(|\tilde{\varphi}|+|\tilde{x}|) \Big)\\
        +\sum_{i=1}^d|\tilde{\varphi}_i|^2\Big(1+\sum_{i,j=1}^d|\varphi_{ij}|\Big)&+\Big(1+\sum_{i,j=1}^d|\varphi_{ij}|^3\Big)(|\tilde{\varphi}|^2+|\tilde{x}|^2)
    \end{aligned}    
    \end{equation}
    for all $A\in \Sigma$, $n\in \mathbb{N}$, $(t,x,\varphi,\varphi_i,\varphi_{ij})\in(0,T)\times B_r(a)\times \mathbb{R}\times\mathbb{R}^d\times\mathbb{S}^d$ and $\eta=(\tilde{x},\tilde{\varphi},\tilde{\varphi}_i,  \tilde{\varphi}_{ij}         )\in \mathbb{R}^d\times\mathbb{R}\times \mathbb{R}^d\times \mathbb{S}^d $, where
    \begin{equation}\label{eqn:F^A,N}
    \begin{aligned}
F^{A,n}_{(\eta)(\eta)}
&:=\sum_{i,j,\ell,r =1}^d\partial_{\varphi_{ij},\varphi_{\ell r}}F^{A,n}\tilde{\varphi}_{ij}\tilde{\varphi}_{\ell r}+2\sum_{i,j,\ell=1}^d \partial_{\varphi_{ij},\varphi_\ell}F^{A,n}\tilde{\varphi}_{ij}\tilde{\varphi}_\ell
\\&\quad+2\sum_{i,j=1}^d\partial_{\varphi_{ij},\varphi}F^{A,n}\tilde{\varphi}_{ij}\tilde{\varphi}+2\sum_{i,j,\ell=1}^d\partial_{\varphi_{ij},x_\ell}F^{A,n}\tilde{\varphi}_{ij}\tilde{x}_\ell+\sum_{i,j=1}^d\partial_{\varphi_i,\varphi_j}F^{A,n}\tilde{\varphi}_i\tilde{\varphi}_j
\\&\quad+2\sum_{i=1}^d\partial_{\varphi_i,\varphi}F^{A,n}\tilde{\varphi}_i\tilde{\varphi}+2\sum_{i,j=1}^d\partial_{\varphi_i,x_j}F^{A,n}\tilde{\varphi}_i\tilde{x}_j+\partial_{\varphi, \varphi}F^{A,n}(\tilde{\varphi})^2
\\&\quad+2\sum_{i=1}^d\partial_{\varphi, x_i}F^{A,n}\tilde{\varphi}\tilde{x}_i+\sum_{i,j=1}^d\partial_{x_i,x_j}F^{A,n}\tilde{x}_i\tilde{x}_j\,.
\end{aligned}
    \end{equation}
\item \label{eqn:clas 3} There exist positive  constants $\delta, M$ %independent of $A\in \Sigma$ and $n\in \mathbb{N}$
such that
\begin{align}
	F^{A,n}(t,x,-M,0,\varphi_{ij})\ge \delta\,,\;
	F^{A,n}(t,x,M,0,-\varphi_{ij})\le- \delta\,,
	\end{align}
	for all $A\in \Sigma$, $n\in \mathbb{N}$,  $(t,x)\in (0,T)\times B_r(a)$ and 
	 all symmetric nonnegative matrices
	$(\varphi_{ij})_{i,j=1}^d$.
\end{enumerate}
Define
\begin{align}
    F(t,x,\varphi,\varphi_i,\varphi_{ij})=\sup_{A\in \Sigma}\lim_{n\to \infty}  F^{A,n}(t,x,\varphi,\varphi_i,\varphi_{ij})\,.
\end{align}
Then for any $u\in C([0,T]\times\overline{B_r(a)})$, the boundary value problem
\begin{align}
	\begin{cases}
	\partial_t\varphi +F(t,x,\varphi,D_x\varphi,D_x^2\varphi)=0\,\,\text{ in } (0,T)\times B_r(a),\\
	\varphi=u\,\,\text{ on }((0,T)\times \partial B_r(a))\cup (\{T\}\times \overline{B_r(a)})
	\end{cases}
	\end{align} 
   admits a classical solution $\varphi\in C^{1,2}([0,T)\times  B_r(a))\cap C( [0,T]\times \overline{B_r(a)} )$.
Moreover if each $F^{A,n}$ is independent of $t$, then for any $u\in C(\overline{B_r(a)})$, the boundary value problem 
\begin{align}
	\begin{cases}
	F(x,\varphi,D_x\varphi,D_x^2\varphi )=0\,\,\text{ in } B_r(a)\,,\\
	\varphi=u\,\,\text{ on }\partial B_r(a)
	\end{cases}
	\end{align} 
admits a classical solution $\varphi\in C^2(B_r(a))\cap C(\overline{B_r(a)})$.
\end{lemma}

\begin{thm}\label{thm:classical}
	Suppose that $m = d$, Assumptions \ref{Assumption ergodic unique} and \ref{sdf assumption} hold, and the functions $b,h_{ij},\sigma,r,k_{ij},v$ are  twice continuously differentiable.  Then the PDE \eqref{sdf pde}
	has a solution pair $(u,\lambda)$ in $C^2(\mathbb{R}^d)\times\mathbb{R}$, and 
	the PDE \eqref{eqn:PDE} 
	has a solution 
	$w$ in $C^{1,2}([0,T]\times\mathbb{R}^d)$.
\end{thm}

\begin{proof}
	For simplicity, assume that $v=0.$ We begin by introducing several functions and constants that will be used in the proof. 
	Let $X^x$ be a solution to the SDE \eqref{SDE appli} and consider the ergodic $G$-BSDE 
	\begin{align}
	Y_s&=Y_T+\int_s^T(-r(X^x_u)-\lambda)\,du+\int_s^T\Big(-k_{ij}(X^x_u)+\frac{1}{2}Z^i_uZ^j_u\Big)\,d\langle B^{i},B^{j}\rangle_u\\
	&\quad-\int_s^TZ_u\,dB_u-(K_T-K_s)\,.\label{eqn: pde clas} 
	\end{align}
	By Theorems \ref{QBSDE with lambda sol exist} and \ref{markovian ergodic solution is unique}, there is a unique solution  
	$(Y^x,Z^x,K^x,\lambda)$ to \eqref{eqn: pde clas} such that
	$Z^x$ is bounded by a constant $L_1>0$. Define a function $u:\mathbb{R}^d\to \mathbb{R}$  as $u(x)=Y_0^x$. 
	Let $a=(a_1,\cdots,a_d)^{\top}\in \mathbb{R}^d$ and let $r$ and $ L_2$ be positive constants  to be specified later.   Set
	\begin{align}
	M:=\sup \{|z \sigma^{-1}(x)| : x\in B_r(a), z\in \mathbb{R}^d,|z|\le L_1 \} \,.
	\end{align}
Define $q:\mathbb{R} \to \mathbb{R}$ by
	\begin{align}
	q(x_1)=-\frac{1}{4r}(x_1-a_1)^2+(x_1-a_1)+L_2 \,,
	\end{align}
and let $p:\mathbb{R}^d \to \mathbb{R}^d$ be a smooth function such that
	\begin{align}
	p(\tilde{x}) =
	\begin{cases}
	\tilde{x}=(\tilde{x}_1,\cdots,\tilde{x}_d)^{\top}, & \text{for } |\tilde{x}| \le M, \\
	0, & \text{for } |\tilde{x}| \ge M+1,
	\end{cases}    
	\end{align}
and which smoothly interpolates between $\tilde{x}$ and $0$ for $M < | \tilde{x} | < M+1$.
For $n\in \mathbb{N}$, let $\psi^n$ be the standard mollification of $\psi$  
for $\psi=b,h_{ij},\sigma,r,k_{ij}$.  
Then $\psi^n$ is smooth on $\overline{B_r(a)}$, and
\begin{align}
\psi^n \to \psi\,,\; D_x \psi^n \to D_x \psi\,,\; D_x^2 \psi^n \to D_x^2 \psi \; \text{ uniformly on } \overline{B_r(a)}.    
\end{align}

We now show that, for sufficiently large $L_2>0$ and small $r>0$, the boundary value problem
	\begin{equation}
    \begin{aligned}\label{pde:bdy pde}
	\begin{cases}
	G(H(   x,D_x\tilde{u},D_x^2\tilde{u}   ) )+\langle b(x),D_x\tilde{u}\rangle-r(x)-\lambda=0\,\,\text{ in } B_r(a),\\
	\tilde{u}=u\,\,\text{ on }\partial B_r(a)
	\end{cases}
	\end{aligned} 
	\end{equation}
	admits a classical solution $\tilde{u}\in C^2(B_r(a))\cap C(\overline{B_r(a)})$.
	 Define $\hat{u}(x)=\frac{u(x)}{q(x_1)}$ for  $x=(x_1,\cdots x_d)\in \overline{B_r(a)}$. Setting  $\tilde{u}(x)=q(x_1)\varphi(x)$,  it suffices to show that the PDE
	\begin{equation} 
\begin{aligned}\label{eqn:PDE_tilde}
\begin{cases}    
G(\tilde{H}(         x,\varphi, D_x\varphi,D_x^2\varphi          ))+\langle b(x),D_x q(x_1)\varphi +D_x\varphi q(x_1)\rangle-r(x)-\lambda=0\,\,\text{ in } B_r(a),\\
\varphi=\hat{u}\,\,\text{ on }\partial B_r(a),
\end{cases}
\end{aligned}
\end{equation}
has a classical solution  $\varphi\in C^2(B_r(a))\cap C(\overline{B_r(a)})$ where 
\begin{equation}\begin{aligned}
&\quad\tilde{H}_{ij}(   x,\varphi, D_x\varphi,D_x^2\varphi    )\\
&=\langle 
(D^2_x q(x_1) \varphi +D_xq(x_1)D_x\varphi^{\top}+D_x\varphi D_xq(x_1)^{\top}+q(x_1)D^2_x\varphi)   
\sigma^{i}(x),\sigma^{j}(x)\rangle\\
&\quad+2\langle D_xq(x_1)\varphi+D_x\varphi q(x_1)   ,h_{ij}(x)\rangle-2k_{ij}(x)\\
&\quad+\langle  \sigma^{i}(x) , p (    D_x q(x_1)\varphi +D_x\varphi q(x_1) )  \rangle  \langle  \sigma^{j}(x) , p(    D_x q(x_1)\varphi +D_x\varphi q(x_1) )\rangle\,.
\end{aligned}\end{equation}
 For each $n\in \mathbb{N}$, let $\tilde{H}^n( x,\varphi, D_x\varphi,D_x^2\varphi )$ denote the function obtained from $\tilde{H}$ by replacing the coefficients $\sigma,h_{ij},k_{ij}$ with  $\sigma^n,h^n_{ij},k^n_{ij}$, respectively. Let $A$ be a symmetric matrix satisfying $\underline{\sigma}^2I_d\le A\le \overline{\sigma}^2I_d$ and define $F^{A,n}: B_r(a)\times \mathbb{R} \times \mathbb{R}^d\times \mathbb{S}^d\to \mathbb{R}$ as
\begin{equation}
	\begin{aligned}
	F^{A,n}(x,\varphi,\varphi_i,\varphi_{ij}    )&= \frac{1}{2}\text{tr}A\tilde{H}^n(x,\varphi,\varphi_i,\varphi_{ij})\\
	&+\langle b^n(x), \partial_{x_1}q(x_1)(\varphi,0,\cdots,0)^{\top}+q(x_1)(\varphi_1,\cdots,\varphi_d)^{\top}\rangle-r^n(x)-\lambda\,.
	\end{aligned}
\end{equation}
We now apply Lemma~\ref{thm:kryl} to this function $F^{A,n}$. 
Since conditions \ref{eqn:clas 1} and \ref{eqn:clas 3} in Lemma \ref{thm:kryl} are readily satisfied for sufficiently large $L_2 > 0$ and small $r > 0$, it remains to verify condition \ref{eqn:clas 2}.
We focus on the term  $\partial_{x_ix_j}F^{A,n}$, which is the most delicate part in  \eqref{eqn:F^A,N}.
Hereafter, $L>0$ denotes a generic constant, whose value may vary from line to line.
By direct calculation, we have
\begin{align}\label{eqn:clas bound}
&\quad |  \langle   \sigma^{n,i}(x),  \partial_{x_1,x_{1}}p (    \partial_{x_1}q(x_1)(\varphi,0,\cdots,0  )^{\top}+q(x_1)(\varphi_1,\cdots,\varphi_d)^{\top} )  \rangle \\
&\quad\cdot\langle\sigma^{n,j}(x),p(   \partial_{x_1}q(x_1)(\varphi,0,\cdots,0  )^{\top}+q(x_1)(\varphi_1,\cdots,\varphi_d)^{\top})\rangle \big|\\
&\le L | \partial_{x_1,x_1}p(  \partial_{x_1}q(x_1)(\varphi,0,\cdots,0  )^{\top}+q(x_1)(\varphi_1,\cdots,\varphi_d)^{\top}  )|\\
&\le L \big( | \text{tr} D^2_xp(  \partial_{x_1}q(x_1)(\varphi,0,\cdots,0  )^{\top}+q(x_1)(\varphi_1,\cdots,\varphi_d)^{\top} )  \alpha\alpha^{\top}|\\
&\quad+\big| D_x p(    \partial_{x_1}q(x_1)(\varphi,0,\cdots,0  )^{\top}+q(x_1)(\varphi_1,\cdots,\varphi_d)^{\top}  )   \big(\frac{1}{2r}(\varphi_1,\cdots,\varphi_d)\big)  \big| \Big)\,,
\end{align}
where $\alpha=(-\frac{1}{2r}\varphi,0,\cdots,0)^{\top}+\partial_{x_1}q(x_1)(\varphi_1,\cdots \varphi_d)^{\top}$.
Observe that  $\frac{1}{2}\le \partial_{x_1}q(x_1)\le 1 $ for $x=(x_1,\cdots,x_d)^{\top}\in B_r(a)$, the support of the function $p$ is contained in the ball $B_{M+1}(0)$, and  
\begin{align}
\alpha&=-\frac{1}{2r\partial_{x_1}q(x_1)} \big(\partial_{x_1}q(x_1)(\varphi,0,\cdots,0  )^{\top}+q(x_1)(\varphi_1,\cdots,\varphi_d)^{\top} \big)  \\&\quad+\Big(\frac{q(x_1)}{2r\partial_{x_1}q(x_1)}+\partial_{x_1}q(x_1)\Big)(\varphi_1,\cdots \varphi_d)^{\top}\,.
\end{align}
It follows that 
\begin{align}
&\quad\big|\text{tr}\big(D^2_xp  ( \partial_{x_1}q(x_1)(\varphi,0,\cdots,0  )^{\top}+q(x_1)(\varphi_1,\cdots,\varphi_d)^{\top})            \alpha\alpha^{\top} \big)\big|\\
&\le L\Big(L+\Big(\frac{q(x_1)}{2r\partial_{x_1}q(x_1)}+\partial_{x_1}q(x_1)\Big)^2(\varphi_1,\cdots \varphi_d)(\varphi_1,\cdots \varphi_d)^{\top}\Big)\\
&\le \sum_{\ell=1}^d L(1+|\varphi_\ell|^2)\,.
\end{align}
By substituting the expression above into \eqref{eqn:clas bound}, we have
\begin{align}
&\quad \big|\langle \sigma^{n,i}(x), \partial_{x_1,x_{1}} p (    \partial_{x_1}q(x_1)(\varphi,0,\cdots,0  )^{\top}+q(x_1)(\varphi_1,\cdots,\varphi_d)^{\top} )  \rangle \\
&\quad\cdot\langle \sigma^{n,j}(x),p(   \partial_{x_1}q(x_1)(\varphi,0,\cdots,0  )^{\top}+q(x_1)(\varphi_1,\cdots,\varphi_d)^{\top})\rangle \big|\le \sum_{\ell=1}^d L (1+|\varphi_\ell|^2)\,.
\end{align}
The other terms in  \eqref{eqn:F^A,N} can be estimated analogously. Hence,  $F^{A,n}$ satisfies condition \ref{eqn:clas 2} in Lemma \ref{thm:kryl}.
It follows that the PDE \eqref{eqn:PDE_tilde} admits a classical solution $\varphi$, and consequently the PDE \eqref{pde:bdy pde} also admits a classical solution $\tilde{u}$.

We now show that the function 
$u$ is twice continuously differentiable. First we prove that $u = \tilde{u}$ on $B_{r'}(a)$ for all $r'<r$. For any $x\in B_{r'}(a)$, we define the stopping time  
\begin{align}
    \tau_{r'}=\inf\{s>0\,|\,r'<|X^x_s-a|<r\}\,.
\end{align}
Recall that $(u(X^x),Z^x,K^x,\lambda)$ is the unique solution to the $G$-BSDE \eqref{eqn: pde clas}.
For simplicity, we define $\tilde{X}^x=\frac{|X^x|\wedge r'}{|X^x|}X^x$ with the convention that $\frac{0}{0}=0$.
Through a similar argument in the proof of Theorem \ref{feynmac kac formula reverse}, we have
\begin{align}\label{eqn: clas bsde}
    \tilde{u}(x)&=\tilde{u}(X^x_{T\wedge\tau_{r'}})-\int_0^{T\wedge\tau_{r'}}r(X_u^x)+\lambda\,du-\int_0^{T\wedge\tau_{r'}}k_{ij}(X^x_u)\,d\langle B^i,B^j\rangle_u\\
    &\quad+\int_0^{T\wedge\tau_{r'}}\frac{1}{2} \langle \sigma^i(X_u^x),p(D_x\tilde{u}(X_u^x))\rangle \langle \sigma^j(X_u^x),p(D_x\tilde{u}(X_u^x)) \rangle\,d\langle B^i,B^j\rangle_u\\
    &\quad-\int_0^{T\wedge\tau_{r'}}(D_x\tilde{u}^{\top}\sigma)(X_u^x)\,dB_u-\tilde{K}^x_{T\wedge\tau_{r'}}\,,
\end{align}
where
\begin{equation}\label{eqn:K}
\begin{aligned}
\tilde{K}^x_{\cdot\wedge\tau_{r'}}&=\int_0^{\cdot\wedge\tau_{r'}}\frac{1}{2}H_{ij}(      \tilde{X}_u^x    , D_x\tilde{u}(\tilde{X}_u^x), D_x^2\tilde{u}(\tilde{X}_u^x))\,d\langle B^i,B^j\rangle_u\\
&\quad-\int_0^{\cdot\wedge\tau_{r'}} G(H_{ij}(      \tilde{X}_u^x    , D_x\tilde{u}(\tilde{X}_u^x), D_x^2\tilde{u}(\tilde{X}_u^x)))\,du\,.
\end{aligned}
\end{equation}  Observe that $Z^x_s=  p^{\top}(   ( Z^x_s\sigma^{-1}(\tilde{X}_s^x)  )^{\top}  )  \sigma(\tilde{X}_s^x)$, $X^x_s=\tilde{X}^x_s$ for $s\in[0,\tau_{r'}]$ and
\begin{align}\label{eqn:32}
    &\quad\sum_{i,j=1}^d\big|Z_s^{x,i}Z_s^{x,j}-\langle \sigma^i(\tilde{X}_s^x),p(D_x\tilde{u}(\tilde{X}_s^x))\rangle\langle \sigma^j(\tilde{X}_s^x),p(D_x\tilde{u}(\tilde{X}_s^x) )\rangle\big|\\
    &\le L\big|Z^x_s- p^{\top}(D_x\tilde{u}(\tilde{X}_s^x))\sigma(\tilde{X}_s^x)\big|\\
    &\le LM_\sigma \big | p((  Z_s^x\sigma^{-1}(\tilde{X}_s^x) )^{\top} )-p(D_x\tilde{u}(\tilde{X}^x_s))    \big|\\
    &\le L
    \big|Z_s^x\sigma^{-1}(\tilde{X}_s^x)-(D_x\tilde{u}^{\top}\sigma \sigma^{-1})(\tilde{X}_s^x)\big|\le L\big|Z_s^x-(D_x\tilde{u}^{\top}\sigma)(\tilde{X}_s^x)\big|\,.    
\end{align}
Then, using a similar argument in the proof of Theorem \ref{Finite QBSDE exis and unique}, for each $\epsilon>0,$ we can construct two processes $  d^{\epsilon,ij},n^{\epsilon,ij}$, a sublinear expectation $\hat{\mathbb{E}}^{d^{\epsilon,ij}}$ and a $G$-Brownian motion $B^{\epsilon}$ such that
\begin{align}
    \tilde{u}(x)-u(x)+\tilde{K}^x_{T\wedge\tau_{r'}} &\le\tilde{u}(X^x_{T\wedge\tau_{r'}})-u(X^x_{T\wedge\tau_{r'}})+\int_0^{T\wedge\tau_{r'}}n_u^{\epsilon,ij}\,d\langle B^{\epsilon,i},B^{\epsilon,j}\rangle_u\\
    &\quad-\int_0^{T\wedge\tau_{r'}}(D_x\tilde{u}^{\top}\sigma)(X_u^x)-Z_u^x\,dB^{\epsilon}_u\,.
\end{align}
 By \cite[Theorem 48]{hu2021extended}, Lemma \ref{K is MG in new sublinear} and Lemma \ref{ergodic recurrent}, we have
 \begin{align}
     &\quad\tilde{u}(x)-u(x)\\
     &\le \hat{\mathbb{E}}^{d^{\epsilon,ij}} [(\tilde{u}(X_T^x)-u(X_T^x))\mathds{1}_{\{T\le \tau_{r'} \}} ] +\hat{\mathbb{E}}^{d^{\epsilon,ij}} [ (\tilde{u}(X_{\tau_{r'}}^x)-u(X_{\tau_{r'}}^x))\mathds{1}_{\{T\ge \tau_{r'}\}   } ]+L \epsilon T\,\\
     &\le
		\frac{L}{\sqrt{T}}+\hat{\mathbb{E}}^{d^{\epsilon,ij}} [ (\tilde{u}(X_{\tau_{r'}}^x)-u(X_{\tau_{r'}}^x))\mathds{1}_{\{T\ge \tau_{r'}  \}} ]+L\epsilon T\,.
 \end{align}
Since $\tilde{u}=u$ on $\partial B_r(a) $ and both  $u$ and $\tilde{u}$ are uniformly continuous on $\overline{B_r(a)}$, we have that for any $\delta>0$, there exists $r'<r$ such that 
\begin{align}\label{eqn:aaa}
    \Big|\hat{\mathbb{E}}^{d^{\epsilon,ij}} [ (\tilde{u}(X_{\tau_{r'}}^x)-u(X_{\tau_{r'}}^x))\mathds{1}_{ \{T\ge \tau_{r'}\} } ]\Big|
    &\le \Big|\hat{\mathbb{E}}^{d^{\epsilon,ij}} \Big[ \tilde{u}(X_{\tau_{r'}}^x)-\tilde{u}\Big( X^x_{\tau_{r'}} + \frac{r-r'}{r'}(  X_{\tau_{r'}}^x -a )  \Big)\Big]\Big|\\
    +\,\Big|\hat{\mathbb{E}}^{d^{\epsilon,ij}}& \Big[{u}\Big(X^x_{\tau_{r'}} + \frac{r-r'}{r'}(  X_{\tau_{r'}}^x -a )\Big)-u(X_{\tau_{r'}}^x) \Big]\Big|\le \delta\,.
\end{align} 
By letting $\epsilon \to 0$ and then $T \to \infty$, we obtain $\tilde{u}(x) - u(x) \le \delta$.
By a similar argument, $u(x) - \tilde{u}(x) \le \delta$, and hence $|u(x) - \tilde{u}(x)| \le \delta$ for $x \in B_{r'}(a)$.
Since $\delta$ is arbitrary, we conclude that $u(x) = \tilde{u}(x)$ for all $x \in B_{r'}(a)$ and $0 < r' < r$.
Therefore, $u = \tilde{u}$ on $B_r(a)$, implying that $u$ is twice continuously differentiable on $B_r(a)$.
As this argument applies to any $a \in \mathbb{R}^d$, we have $u \in C^{2}(\mathbb{R}^d)$.

We now prove that the PDE \eqref{eqn:PDE} admits a classical solution.
Since the proof is similar to the preceding argument, we present only a sketch.
By Theorem \ref{QBSDE feykac formula}, there is a viscosity solution 
$w\in C([0,T]\times\mathbb{R}^d)$ 
to  \eqref{eqn:PDE}.   First we verify that the PDE
\begin{align}\label{eqn:fine pde clas}
	\begin{cases}
	 \partial_t\overline{w}+G(H(    x,D_x\overline{w}, D_x^2\overline{w}            )             ) +\langle b(x),D_x\overline{w}\rangle-r(x)=0\,\,\text{ in } (0,T)\times B_r(a),\\
	\overline{w}=w\,\,\text{ on }((0,T)\times \partial B_r(a))  \cup (\{T\}\times \overline{B_r(a)})
	\end{cases}
	\end{align} 
admits a classical solution $\overline{w}$. Define $\hat{w}(t,x)=e^{L_2t}w(t,x)$ for $(t,x)\in [0,T]\times \overline{B_r(a)}$.  Then, by setting $\overline{w}(t,x)=e^{-L_2t}\overline{\varphi}(t,x)$, it suffices to prove that the PDE
\begin{align}
	\begin{cases}
	 \partial_t\overline{\varphi}+e^{L_2t}G(\overline{H}(               x,e^{-L_2t}D_x\overline{\varphi}, e^{-L_2t}D_x^2\overline{\varphi}           ))\\
    \quad\quad +\langle b(x),D_x\overline{\varphi}\rangle-e^{L_2t}r(x)-L_2\overline{\varphi}=0\,\,\text{ in } (0,T)\times B_r(a),\\
	\overline{\varphi}=\hat{w}\,\,\text{ on }((0,T)\times \partial B_r(a)) \cup (\{T\}\times \overline{B_r(a)})
	\end{cases}
	\end{align} 
    has a classical solution $\overline{\varphi}$,
    where 
   \begin{equation}\begin{aligned}
	\overline{H}_{ij}(   x,e^{-L_2t}D_x\overline{\varphi}, e^{-L_2t}D_x^2\overline{\varphi} )&=\langle e^{-L_2t} D^2_x\overline{\varphi}\sigma^i(x),\sigma^j(x)\rangle+2\langle e^{-L_2t} D_x\overline{\varphi},h_{ij}(x)\rangle\\&-2k_{ij}(x)+\langle\sigma^i(x),p(e^{-L_2t}D_x\overline{\varphi})\rangle\langle\sigma^j(x),p(e^{-L_2t}D_x\overline{\varphi})\rangle\,.
	\end{aligned}\end{equation}
	For each $n\in \mathbb{N}$, let $\overline{H}^n$ denote the function obtained from $\overline{H}$ by replacing the coefficients $\sigma,h_{ij},k_{ij}$ with  $\sigma^n,h^n_{ij},k^n_{ij}$, respectively. Let $A$ be a symmetric matrix satisfying $\underline{\sigma}^2I_d\le A\le \overline{\sigma}^2I_d$ and  define $\overline{F}^{A,n}:(0,T)\times B_r(a)\times \mathbb{R}\times \mathbb{R}^d\times\mathbb{S}^d\to \mathbb{R}$ as
	\begin{align}
	\overline{F}^{A,n}(       t,x,\overline{\varphi}, \overline{\varphi}_i,   \overline{\varphi}_{ij})&= \frac{1}{2}\text{tr}A(e^{L_2t}\overline{H}^n  (   x,e^{-L_2t}\overline{\varphi}_i,e^{-L_2t}\overline{\varphi}_{ij}))\\
	&\quad+\langle b^n(x), D_x\overline{\varphi}\rangle-e^{L_2t}r^n(x)-L_2\overline{\varphi}\,.
	\end{align}
It follows that conditions \ref{eqn:clas 1}, \ref{eqn:clas 2} and \ref{eqn:clas 3} in Lemma \ref{thm:kryl}  are satisfied, and hence the PDE \eqref{eqn:fine pde clas} admits a classical solution $\overline{w}$.

We now prove that $w$ is twice continuously differentiable.
It suffices to verify that
  $w(t,x)=\overline{w}(t,x)$ for all $(t,x)\in [0,T]\times B_r(a)$. We set $t = 0$ for simplicity and select $0 < r' < r$ such that $x \in B_{r'}(a)$.
  Define the stopping time
$    \tau_{r'}=\inf\{s>0\,|\,r'<|X^x_s-a|<r\}.$
Applying It\^o's formula, we have
\begin{equation} 
\begin{aligned}
   \overline{w}(0,x)
    &=\overline{w}({T\wedge\tau_{r'}},X^x_{T\wedge\tau_{r'}})-\int_0^{T\wedge\tau_{r'}}r(X_u^x)\,du-\int_0^{T\wedge\tau_{r'}}k_{ij}(X^x_u)\,d\langle B^i,B^j\rangle_u\\
    &\quad+\int_0^{T\wedge\tau_{r'}}\frac{1}{2}\langle \sigma^i(X_u^x),p(D_x\overline{w}(u,X_u^x))\rangle\langle \sigma^j(X_u^x),p(D_x\overline{w}(u,X_u^x))\rangle\,d\langle B^i,B^j\rangle_u\\
    &\quad-\int_0^{T\wedge\tau_{r'}}D_x\overline{w}^{\top}(u,X_u^x)\sigma(X_u^x)\,dB_u-\overline{K}_{T\wedge\tau_{r'}}
\end{aligned}
\end{equation}
where
$\overline{K}_{\cdot \wedge \tau_{r'}}$ is defined analogously to \eqref{eqn:K}.
Construct a process $\overline{d}^{\epsilon,ij}$ and a sublinear expectation $\hat{\mathbb{E}}^{\overline{d}^{\epsilon,ij}}$ such that
\begin{align}
    \overline{w}(0,x)-w(0,x) &\le \hat{\mathbb{E}}^{\overline{d}^{\epsilon,ij}} [(\overline{w}(T,X_T^x)-w(T,X_T^x))\mathds{1}_{ \{T\le \tau_{r'} \} } ]\\
    &\quad +\hat{\mathbb{E}}^{\overline{d}^{\epsilon,ij}} [ (\overline{w}(\tau_{r'},X_{\tau_{r'}}^x)-w(\tau_{r'},X_{\tau_{r'}}^x))\mathds{1}_{\{T\ge \tau_{r'} \}} ]+L \epsilon T\,.
\end{align} 
Since $\overline{w}=w$ on $ ((0,T) \times  \partial B_r(a)) \cup (\{T\} \times B_r(a))$, and both  $w$ and $\overline{w}$ are uniformly continuous on $[0,T]\times \overline{B_r(a)}$, by applying a similar argument to that in \eqref{eqn:aaa}, we obtain the desired result.  
\end{proof}

Theorems \ref{sdf thm}, \ref{thm:unique}, and \ref{feynmac kac formula reverse} are 	novel, even in the specific case where the G-Brownian motion is the standard Brownian motion. 
Theorem \ref{sdf thm} establishes a sufficient condition for achieving the long-term
decomposition  based on the parameters of pricing kernels.
The long-term decomposition of pricing kernels is not always guaranteed to exist, making it essential to find conditions under which such a decomposition is valid.
\cite{hansen2009long} and \cite{qin2017long} suggested Lyapunov-type criteria  
to ensure the existence of the decomposition. However, these criteria are not explicitly based on model parameters and are typically difficult to verify in practice.

It is desirable to formulate sufficient conditions directly in terms of model parameters.
Although previous studies have examined this problem in specific settings such as the HJM framework, there has been no work addressing it within the broader context of general Markovian market models.
Our study is the first to provide a sufficient condition for the decomposition  in terms of model parameters under a general Markovian framework.

Theorem \ref{thm:unique} demonstrates the uniqueness of
the long-term decomposition.
While the uniqueness of the long-term exponential rate 
$\lambda=\lim_{T\to\infty}\frac{1}{T}\ln \hat{\mathbb{E}}[D_T]$
is straightforward, the uniqueness of the other components is more subtle. This issue has not been explicitly addressed in the existing literature. Our study is the first to provide sufficient conditions for uniqueness in terms of model parameters.

Furthermore, Theorem \ref{feynmac kac formula reverse} offers an
alternative characterization of the decomposition through the solution to the PDE. This result provides a novel representation of the components in the decomposition.
Our findings reveal how these components can be expressed in terms of the solution to a second-order parabolic PDE.
This type of PDE-based representation has not been previously explored in the literature.

\section{Conclusion}
\label{sec:con}

This study investigates the long-term decomposition of pricing kernels under the $G$-expectation framework.
Our main result, 
Theorem \ref{sdf thm}, demonstrates that the pricing kernel can be decomposed
into a discounting component with a long-term exponential rate, a transitory component, a positive symmetric $G$-martingale, and a decreasing process.  The uniqueness of this decomposition is demonstrated in Theorem \ref{thm:unique}.
Theorem \ref{feynmac kac formula reverse}  provides an alternative characterization of the decomposition through a solution pair to the PDE \eqref{sdf pde}.
These findings extend previous results obtained under a single fixed  probability framework to the $G$-expectation context.

$ $

\noindent\textbf{Acknowledgement.} Hyungbin Park was supported by the National Research Foundation of Korea (NRF) grants funded by the Ministry of Science and ICT (Nos. 2021R1C1C1011675 and 2022R1A5A6000840). Financial support from the Institute for Research in Finance and Economics of Seoul National University is gratefully acknowledged.

\appendix

\section{Pricing kernels in equilibrium}

Pricing kernels  reflect  how economic agents value uncertain future cash flows and play a central role in asset pricing. In the classical framework, they are closely linked to general equilibria and linear pricing rules. 
This section discusses how pricing kernels are endogenously determined in equilibrium under the $G$-expectation framework. 
The results presented here correspond to a special case of the framework developed in \cite[Section 3.2]{epstein2013ambiguous}.

We consider an equilibrium in a representative agent economy. Let $B$ be a $d$-dimensional $G$-Brownian motion.
Suppose the endowment process $(w_s)_{s \ge 0}$ is strictly positive and follows the dynamics
$$
\frac{dw_s}{w_s} = b_s\,ds + \sigma_s\,dB_s,\quad w_0 = 1,
$$
where $b \in \mathbb{M}_G^2(0,T;\mathbb{R})$ and $\sigma \in \mathbb{M}_G^2(0,T;\mathbb{R}^{1\times d})$.
The agent's preferences
are described by the utility functional 
$$
U(c) = -\hat{\mathbb{E}}\Big[-\int_0^T e^{-\beta s} u(c_s)\,ds\Big]\,,
$$
for nonnegative consumption processes $c\in \mathbb{M}_G^2(0,T;\mathbb{R}) $
where $\beta > 0$ is the discount rate and $u:\mathbb{R} \to \mathbb{R}$ is a strictly increasing, concave, and three-times continuously differentiable utility function.
The market consists of a single consumption good, a risk-free asset with short rate $r_s$, and $d$ risky securities in zero net supply. The returns on the risky securities evolve as
$$
dR_s = \mu_s\,ds + \Sigma_s\,dB_s,
$$
where $r \in \mathbb{M}_G^2(0,T;\mathbb{R})$, $\mu \in \mathbb{M}_G^2(0,T;\mathbb{R}^{d \times 1})$ and $\Sigma \in \mathbb{M}_G^2(0,T;\mathbb{R}^{d \times d})$ are determined endogenously in equilibrium.
We define a symmetric, positive-definite, matrix-valued process $\eta = (\eta^{ij})_{1 \le i,j \le d}$ by
$$
\eta_s^{ij} := \limsup_{\epsilon \searrow 0} \frac{1}{\epsilon} ( \langle B^i, B^j \rangle_s - \langle B^i, B^j \rangle_{s - \epsilon} )$$
for $s > 0$ and $\eta_0^{ij}:=0.$

By \cite[Theorem 3]{epstein2013ambiguous}, under sequential equilibrium, we have
$$\theta_s:=\Sigma_s^{-1}(\mu_s-r_s\mathbf{1})=-\frac{u''(w_s)}{u'(w_s)}w_s\eta_s\sigma_s^\top$$
and
$$r_s=\frac{1}{2}u'''(w_s)\sigma_s\sigma_s^\top-\frac{u''(w_s)}{u'(w_s)}b_s-\beta\,.$$
The corresponding pricing kernel is given by
\begin{equation}\label{eqn:PK}
\begin{aligned}
D_s
&=e^{-\int_0^sr_u\,du-\frac{1}{2}\int_0^s\theta_u^\top \eta_u^{-1}\theta_u\,du-\int_0^s\theta_u^\top \eta_u^{-1}\,dB_u}\\
&=e^{-\int_0^sr_u\,du-\frac{1}{2}\int_0^sv_{i,u}v_{j,u}\,d\langle B^i,B^j \rangle_u-\int_0^sv_u \,dB_u}
\end{aligned}
\end{equation}
where 
$v_s=(v_{1,s},v_{2,s},\cdots,v_{d,s})=-\frac{u''(w_s)}{u'(w_s)}w_s\sigma_s.$
In particular, if the economy is driven by a state variable process $X = (X_s)_{s \ge 0}$ such that there exist continuous functions $r(\cdot)$ and $v(\cdot)$ with
$$
r_s = r(X_s) \quad \text{and} \quad v_s = v(X_s) $$
for all $s$,
then the pricing kernel in \eqref{eqn:PK} reduces to the expression given in \eqref{SDF}.

\section{Stochastic deflators}

We examine the stochastic deflator used by conservative traders under volatility uncertainty and demonstrate that it admits the representation given in \eqref{SDF}.
To begin, we consider a pricing kernel in a standard Brownian setting. Let $(D_s)_{s \ge 0}$ be a pricing kernel of the form
$$D_s=e^{-\int_0^sr(X_u)\,du-\int_0^sv(X_u)\,dB_u}$$ 
where $B$ is a $d$-dimensional standard Brownian motion, and $r$ and $v$ are discount rate functions. The process $X$ is a Markov process defined as the solution to the SDE
\begin{align}
X_s = x + \int_0^s b(X_u)\,du + \int_0^s \sigma(X_u)\,dB_u\,,\; s \ge 0
\end{align}
with drift function $b$ and diffusion function $\sigma$.
A trading agent uses this pricing kernel to convert a future payoff $\Phi$ at time $T$ into its present value by computing the expectation
$$\mathbb{E}[D_T\Phi]=\mathbb{E}\big[e^{-\int_0^Tr(X_u)\,du-\int_0^Tv(X_u)\,dB_u}\Phi\big]\,.$$

Under volatility uncertainty, a conservative trading agent evaluates upper prices by taking the supremum over a family of expectations associated with different volatility scenarios. Suppose the agent considers a bounded, convex and closed set $\Sigma$ of $d \times d$ symmetric positive definite matrices as the range of admissible Brownian covariance structures. Then, the upper price of a future payoff $\Phi$ is given by
\begin{equation}
\label{vol_uncertain}
\sup_{  \eta\in \Sigma}\mathbb{E}\big[e^{-\int_0^Tr(X_u^{\eta})\,du-\int_0^Tv(X_u^{\eta})\,dB_u^{\eta}}\Phi\big] \end{equation}
where $\eta$ is a progressively measurable process taking values in $\Sigma$,  $B^{\eta} = \int_0^\cdot \sqrt{\eta_u}\,dB_u$ is the corresponding  
distorted Brownian motion,
  and  $X^{\eta}$ is a solution to  the SDE
\begin{align} 
X^{\eta}_s=x+\int_0^sb(X^{\eta}_u)\,du+\int_0^s\sigma(X^{\eta}_u)\,dB_u^{\eta}\,,\;s\ge0\,.
\end{align}

This problem can be formulated more generally and concisely within the framework of $G$-expectation. Let $B$ be a $G$-Brownian motion with uncertain covariance range $\Sigma$ under a $G$-expectation $\hat{\mathbb{E}}$. Consider the process
\begin{equation} 
D_s:=e^{-\int_0^sr(X_u)\,du-\int_0^sk_{ij}(X_u)\,d\langle B^j,B^j\rangle_u-\int_0^sv(X_u)\,dB_u}\,,\;s\ge0
\end{equation}
where $r, k_{ij},$ and $v$ are discount rate functions, and $X$ satisfies the $G$-SDE
\begin{align} 
X_s=x+\int_0^sb(X_u)\,du+\int_0^sh_{ij}(X_u)\,d\langle B^i,B^j\rangle_u+\int_0^s\sigma(X_u)\,dB_u\,,\;s\ge0
\end{align}
for drift functions $b, h_{ij}$ and a volatility function $\sigma$.
Then, the $G$-expectation $\hat{\mathbb{E}}[D_T \Phi]$ coincides with the upper price  given in \eqref{vol_uncertain} when   $k_{ij} = h_{ij} = 0$.
Accordingly,  the process $(D_s)_{s \ge 0}$ can be interpreted as a stochastic deflator for conservative traders under volatility uncertainty.

\bibliographystyle{apalike}
\bibliography{references}

\end{document}